\documentclass[11 pt]{article}
\usepackage{hyperref}
\usepackage{graphicx,color,xcolor} 
\usepackage{amsmath,amssymb,amsfonts,amsthm}
\usepackage[title]{appendix}

\newtheorem{thm}{Theorem}
\newtheorem{lemma}[thm]{Lemma}
\newtheorem{prop}[thm]{Proposition}

\newtheorem{coro}[thm]{Corollary} 
\newtheorem{remark}[thm]{Remark}

\newtheorem{assumption}[thm]{Assumption}

\newcommand{\HhalfG}{\cH^{1/2}(\Gamma)}
\newcommand{\HmhalfG}{\cH^{-1/2}(\Gamma)}

\newcommand{\dpt}{d\ell}

\newcommand{\im}{\textrm i} 
\renewcommand{\ker}{\text{Ker}}

\newcommand{\chomo}{\cC_z}

\newcommand{\nuGb}{\nu_2}
\newcommand{\rflc}{F}

\newcommand{\eps}{\varepsilon}

\newcommand{\bbeta}{\boldsymbol\beta}

\newcommand{\Kone}{K}
\newcommand{\Ktwo}{K'}
\newcommand{\sgn}{\text{sgn}}
\newcommand{\kp}{k_{\parallel}}

\newcommand{\bbZ}{\mathbb Z}

\newcommand{\be}{\mathbf e}

 \newcommand{\bn}{\mathbf n}
 \newcommand{\bp}{\mathbf p}

\newcommand{\bw}{\mathbf w} \newcommand{\bx}{\mathbf x} 
\newcommand{\by}{\mathbf y}  
\newcommand{\bA}{\mathbf A}

 \newcommand{\cB}{\mathcal B}
\newcommand{\cC}{\mathcal C} \newcommand{\cD}{\mathcal D} 
 
 \newcommand{\cH}{\mathcal H}
\newcommand{\cI}{\mathcal I} 
\newcommand{\cK}{\mathcal K} \newcommand{\cL}{\mathcal L}
\newcommand{\cM}{\mathcal M} \newcommand{\cN}{\mathcal N}
\newcommand{\cO}{\mathcal O}  
 \newcommand{\cR}{\mathcal R}
\newcommand{\cS}{\mathcal S}

\newcommand{\A}{A}
\newcommand{\B}{B}
\newcommand{\C}{C}

\newcommand{\J}{J}
\newcommand{\Lint}{\cL^{\text{int}}}

\usepackage{epstopdf}

\title{Mathematical Theory for Photonic Hall Effect in Honeycomb Photonic Crystals}
\author{
    Wei Li    
    \thanks{Department of Mathematical Sciences, DePaul University, Chicago, IL 60614.
    \tt wei.li@depaul.edu.} \,\,\,  
     Junshan Lin
  \thanks{Department of Mathematics and Statistics, Auburn University, Auburn, AL 36849.  \tt jzl0097@auburn.edu.}\,\,\,
  Jiayu Qiu\thanks{Department of Mathematics, ETH Z\"{u}rich, R\"{a}mistrasse 101, CH-8092 Z\"{u}rich, Switzerland.
    \tt jiayu.qiu@sam.math.ethz.ch.}\,\,\,
  Hai Zhang
  \thanks{Department of Mathematics, 
 HKUST,  Clear Water Bay, Kowloon, Hong Kong S.A.R., China.  \tt haizhang@ust.hk.}}
\date{}

\begin{document}

\maketitle
\begin{abstract}
In this work, we develop a mathematical theory for the photonic Hall effect and prove the existence of guided electromagnetic waves at the interface of two honeycomb photonic crystals. The guided wave resembles the edge states in electronic systems: it is induced by the topological Hall effect, and the wave propagates along the interface but not in the bulk media. Starting from a symmetric honeycomb photonic crystal that attains Dirac points at the high-symmetry points of the Brillouin zone,  $K$ and $K'$, we introduce two classes of perturbations for the periodic medium. The perturbations lift the Dirac degeneracy, forming a spectral band valley at the points $K$ and $K'$ with well-defined topological phase that depends on the sign of the perturbation parameters.
By employing the layer potential techniques and spectral analysis, we investigate the existence of guided wave along an interface when two honeycomb photonic crystals are glued together. In particular, we elucidate the relationship between the existence of the interface mode and the nature of perturbations imposed on the two periodic media separated by the interface.

\end{abstract}

\section{Introduction}
\subsection{Background}
Recent developments in topological insulators have opened up new avenues for guiding classical waves robustly. In topological insulators, an insulating bulk electronic material supports localized edge states on its surface that is immune to backscattering and system disorder \cite{bernevig-13, Hasan-Kane-10,  Qi-Zhang-11}. Such features can also be realized in photonic structures by mimicking the quantum Hall effect in topological insulator using active components to break the time-reversal symmetry of the photonic system or by relying on an analogue of the quantum valley Hall or spin Hall effect using passive components to break the spatial symmetry of the photonic system \cite{Khanikaev-12, Khanikaev-15, Lu-14, Ma_Shvets-15, Chan-19, Ozawa-19, raghu-08, Yang-15, Wang-09,  Wu_Hu-15}.  Photonic crystals engineered in such a way attain nontrivial topological phases, and they support localized electromagnetic wave modes propagating along the medium interface, which is referred to as the photonic Hall effect.

In general, the photonic Hall effect is built upon a periodic dielectric  medium with certain symmetries, such as photonic crystal with a honeycomb lattice. The spectral band structure of the corresponding elliptic operator attains Dirac points where two dispersion surfaces cross linearly. For honeycomb photonic crystals,
Dirac points generically appear at the high-symmetry points of the Brillouin zone, $K$ and $K'$. By introducing specific perturbations to the high-symmetry periodic medium that break time-reversal symmetry or inversion symmetry of the photonic system, the degeneracy at the Dirac points is lifted. Meanwhile, the local extrema of the dispersion surfaces is attained at $K$ and $K'$, forming spectral band valleys. The Berry curvature near the valley locations $K$ and $K'$ carry the topological phase of the periodic media. Indeed, a valley Chern number can be defined by integrating the Berry curvature around $K$ and $K'$, and it can be shown that two opposite perturbations of the periodic medium give rise to opposite Chern numbers at $K$ or $K'$. When two photonic crystals with opposite valley Chern numbers at $K$ or $K'$ are connected together, a localized electromagnetic wave mode can be supported along their interface.

The objective of this paper is to develop a mathematical theory for the existence of the interface (edge) modes when two honeycomb photonic crystals attaining spectral band valleys are connected together.
 On the two sides of the interface, the honeycomb photonic crystals are perturbed from 
 the periodic medium with the Dirac points. The  perturbations are either from the same class or two different classes. 
 We examine the spectral gap opening near the Dirac points and the formation of spectral band valleys at $K$ and $K'$ points of the Brillouin zone when the perturbation is introduced. It is shown that the Berry curvature swap signs when the sign of the perturbation parameter swaps.
 It is then proved that interface (edge) modes  bifurcating from $\Kone$ and $\Ktwo$ emerge for the joint photonic structure when the two joining periodic media attain various perturbations; see Theorem \ref{thm:edge} and Corollary \ref{thm:edgearm} for details. The study provides a rigorous mathematical theory for the magneto-optical photonic Hall effect and valley Hall effect that have been explored experimentally in \cite{Lu-14, Ma_Shvets-15, Yang-15, Wang-09}. In a forthcoming paper, we shall develop mathematical theory to investigate the photonic spin Hall effect that relies on double-degenerate Dirac points and a geometric perturbation of the periodic medium.

We would like to point out that significant progress has made in recent years regarding the mathematical theory of the Hall effect and the existence of the topological edge state in graphene; see, for instance, \cite{Drouot-Wenstein-20, Fefferman-Thorp-Weinsein-16, Lee-Thorp-Weinstein-Zhu-19}. Therein the authors consider the continuum Schr\"{o}dinger operators and the joint edge operators using the domain wall-models, which assume that two bulk media are ``connected'' adiabatically over a length scale that is much larger than the period of the structure. A two-scale analysis is presented in \cite{Drouot-Wenstein-20, Fefferman-Thorp-Weinsein-16, Lee-Thorp-Weinstein-Zhu-19} for the edge operator, which yields an effective Dirac equation for the slowly varying amplitudes. In this work, we consider photonic crystal models wherein two different periodic media are connected directly along the interface and the separation of the scale is no longer present. Hence the multiscale expansion method that have been developed for the domain-wall models can not be applied to study the edge modes directly, and the wave in the joint structure considered in the work can no longer be described by an effective Dirac operator. To address the new challenges in the spectral analysis brought by the discontinuities of coefficients and the presence of the sharp interface in the PDE model, we apply the mathematical framework based on a combination of layer potential theory, asymptotic analysis, and the generalized Rouch\'e theorem.
In addition, our result implies the \textit{bulk-edge correspondence} for the topological photonic crystals. More precisely, the number of the interface modes is determined by the difference of the bulk invariants (valley Chern number) across the interface. We refer to
\cite{Bal-19, Druout-20-3, Thiang-Zhang-22,bal2024continuous} and references therein for the bulk-edge correspondence in several other PDE models by using different techniques.

\subsection{Main results}
\subsubsection{Periodic elliptic operators}
We consider a family of elliptic operators modeling the propagation of transverse magnetic (TM) wave in honeycomb photonic crystals  wherein the electric vectors are in the plane of propagation. The periodic differential operator is defined by
\begin{equation}\label{eq:pert}
\cL(\eps,\delta):= - \nabla\cdot\A(\eps,\delta;\bx) \nabla,
\end{equation}
where the material tensor $\A(\eps,\delta;\bx)$ depends on two parameters, $\eps$ and $\delta$, each representing one class of perturbation for the periodic medium specified in \eqref{eq:BC} and Assumption \ref{ass:symBC}.

The honeycomb lattice is denoted by $\Lambda := \bbZ \be_1 \oplus  \bbZ \be_2 := \{ \ell_1 \be_1 + \ell_2 \be_2 :  \ell_1, \ell_2 \in \bbZ \}$,
where the lattice vectors $\be_1=(\frac{\sqrt3}{2},-\frac{1}{2})^T$ and $\be_2=(\frac{\sqrt3}{2},\frac{1}{2})^T$. 
Here and henceforth, $\cC_z := \{ \ell_1 \be_1 + \ell_2 \be_2 :  \ell_1, \ell_2\in [-1/2,1/2) \}$ represents
the fundamental cell of the lattice shown in Figure \ref{fig:periodic_cell} (left).
The reciprocal lattice is given by $\Lambda^* = \{ 2\pi\ell_1 \bbeta_1 + 2\pi\ell_2 \bbeta_2 :  \ell_1, \ell_2\in\mathbb Z\}$, where the reciprocal lattice vectors
$\bbeta_1=(\frac{1}{\sqrt3},-1)^T$ and  $\bbeta_2=(\frac{1}{\sqrt3},1)^T$
satisfy $\be_i\cdot\bbeta_j=\delta_{ij}$ for $i,j=1,2$.
The hexagon-shaped fundamental cell of $\Lambda^*$, or the Brillouin zone, is denoted by $\cB_z$ and shown in Figure \ref{fig:periodic_cell} (right).

\begin{figure}[h]
\begin{center}
\includegraphics[height=4cm]{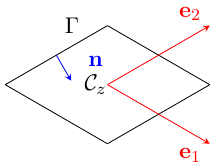} 
\includegraphics[height=5cm]{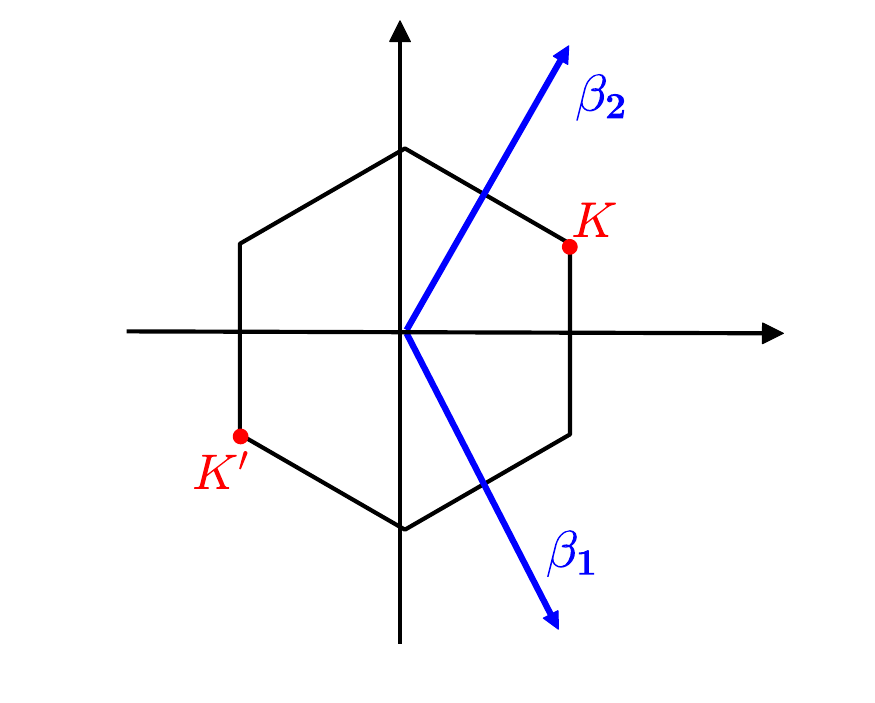}
\caption{The fundamental cell $\cC_z$ (left) and the Brillouin zone $\cB$ (right).}
\label{fig:periodic_cell}
\end{center}
\end{figure}

Let $R$ be the clockwise rotation matrix with an angle of $\frac{2\pi}{3}$ and $F$ the reflection matrix about the $x_2$-axis defined as follows:
\begin{equation}\label{eq:RF}
R\bx := \begin{pmatrix}-\frac{1}{2} & \frac{\sqrt3}{2} \\ -\frac{\sqrt3}{2} & -\frac{1}{2}\end{pmatrix}\bx,\quad
\rflc \bx := \begin{pmatrix} -1&0 \\ 0 &1\end{pmatrix}\bx.
\end{equation}
For each $2\times2$ orthogonal matrix $O$, the induced transform on functions is defined as
\[
\cO f(\bx) := f(O^{-1}\bx).
\]

\begin{assumption}\label{ass:symFull}
The material tensor $\A(\eps,\delta;\cdot)\in L^{\infty}(\mathbb R^2, \mathbb M_{2\times2})$ attains the following properties:
\begin{enumerate}
\item \textbf{Ellipticity:} There exists a constant $\gamma>0$ such that for all $\bx\in\mathbb{R}^2$ and for all $\xi\in\mathbb{R}^2$,
\[
\xi^T \A(\eps,\delta;\bx) \xi \geq \gamma |\xi|^2.
\]
\item \textbf{Hermiticity:} $\A(\eps,\delta;\bx)$ satisfies $\overline{\A(\eps,\delta;\bx)}^T = \A(\eps,\delta;\bx)$.
\item \textbf{Honeycomb periodicity:} The tensor is $\Lambda$-periodic such that
\[
\A(\eps,\delta;\bx+\be)=\A(\eps,\delta;\bx)\quad \text{for all } \be\in\Lambda.
\]
\item \textbf{$\frac{2\pi}{3}$-rotational invariance:} 
\[
\A(\eps,\delta;R^{-1}\bx) = R^{-1} \A(\eps,\delta;\bx)R.
\]

\end{enumerate}
\end{assumption}

\medskip

For brevity we denote the unperturbed operator by 
\[
\cL_0:=\cL(0,0),
\]
which attains further symmetries as follows. 

\begin{assumption}\label{ass:symA}
The unperturbed material tensor, $\A_0(\bx):= \A(0,0;\bx)$, has the following symmetries:
\begin{enumerate}
\item \textbf{Reflection invariance:} 
$\A_0(\rflc^{-1}\bx) = \rflc^{-1} \A_0(\bx) \rflc $.
\item \textbf{Time-reversal/Conjugate invariance:} $\overline{\A_0(\bx)} = \A_0(\bx)$.
\end{enumerate}
\end{assumption}
We refer to Figure \ref{fig:periodic_media} for an illustration of the symmetries imposed on the tensor $\A_0(\bx)$. Let $ K:=2\pi(\frac{1}{\sqrt3},\frac{1}{3})=2\pi(\frac{1}{3}\bbeta_1+\frac{2}{3}\bbeta_2)$ and $K':=-K$ be two high symmetry points of the Brillouin zone as shown in Figure \ref{fig:periodic_cell}. Then under a generic nondegeracy condition, the spectral band structure of $\cL_0$ exhibits Dirac points $(\Kone,\lambda_*)$ and $(\Ktwo,\lambda_*)$. More precisely, the dispersion relations near $\bp_* = K$ are expressed by
\begin{equation*}
(\lambda-\lambda_*)^2 = m_*^2 \, |\bp-\Kone|^2+O(|\bp-\Kone|^3),
\end{equation*}
wherein $\lambda_*$ is an eigenvalue of multiplicity two for the operator $\cL_0$. We refer to Theorem \ref{lem:Dirac} for more details.

\begin{assumption}[The no-fold condition along the direction $\bbeta$]\label{lem:assNoFold}
Let $\bbeta\in\mathbb R^2$ be a fixed Bloch wave vector and $\lambda_*$ denote the Dirac eigenfrequency at $K$ or $K'$. For any quasi-momentum of the form $\bp = \Kone + \ell\bbeta$ with $\ell\in \mathbb{R}$, the spectral band of $\cL_0$ attains the value $\lambda_*$ only if
\[
\bp\in (\Kone+\Lambda^*)\cup(\Ktwo+\Lambda^*).
\]
This condition ensures the absence of band folding along the $\bbeta$ direction.
\end{assumption}

\begin{figure}[!htbp]
\begin{center}
\includegraphics[width=5cm]{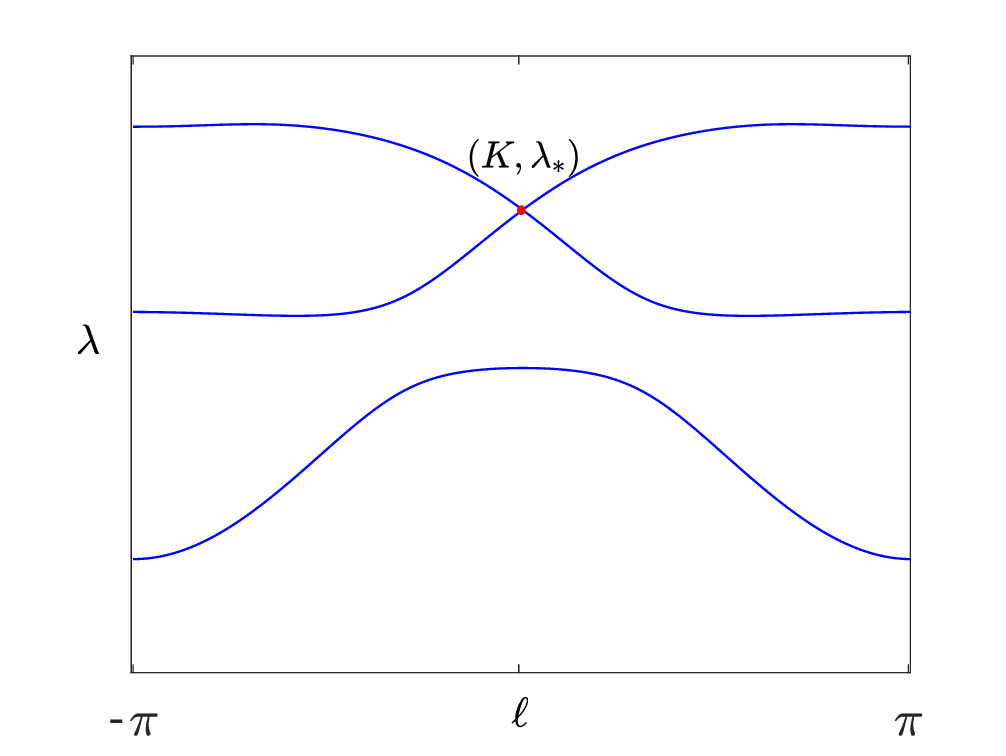}
\includegraphics[width=5cm]{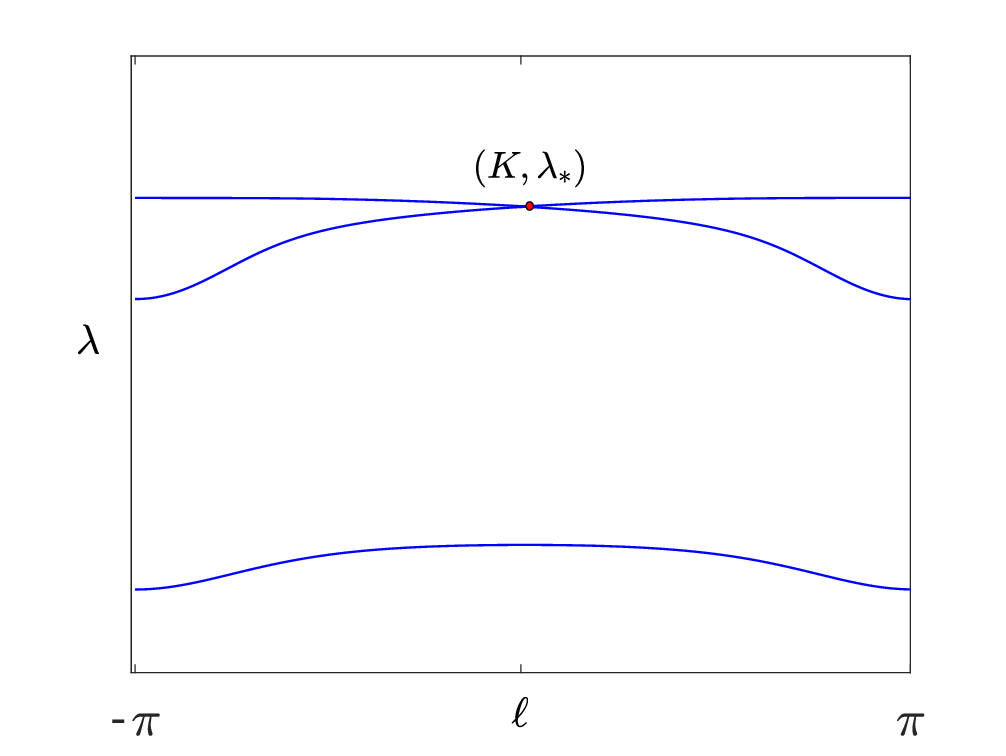}
\includegraphics[width=5cm]{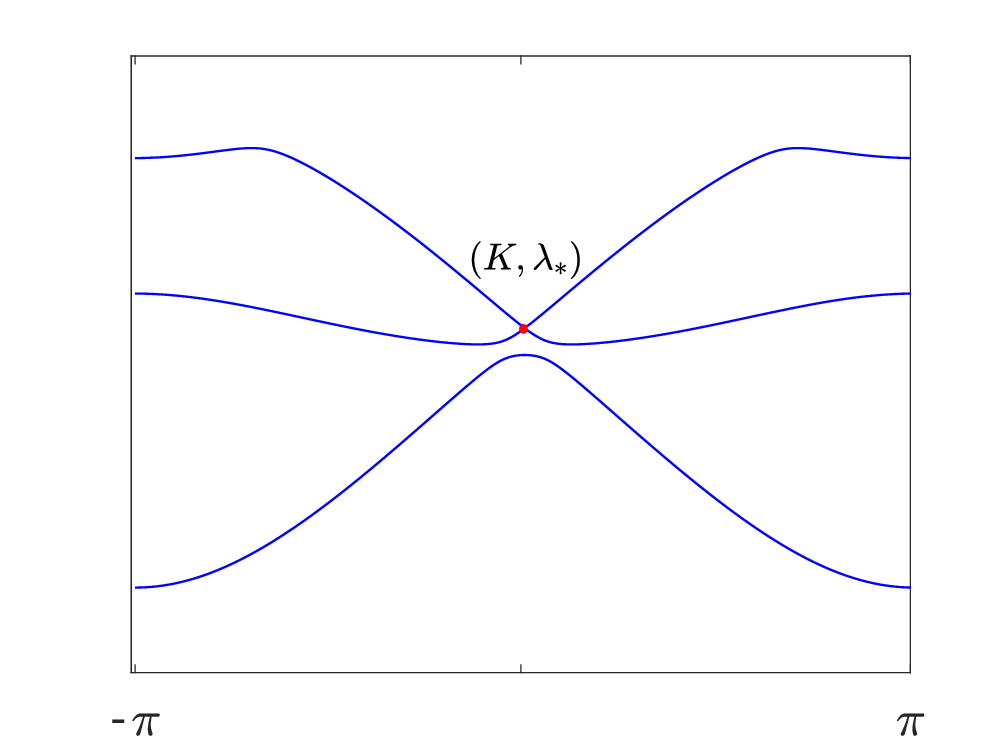}
\includegraphics[width=5cm]{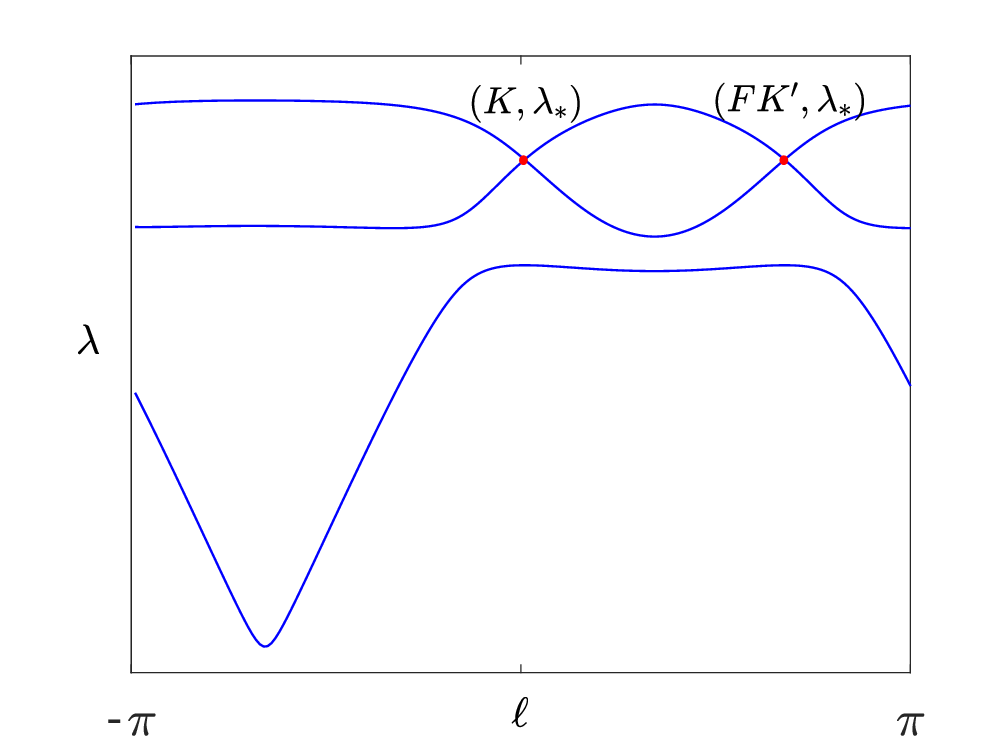}
\includegraphics[width=5cm]{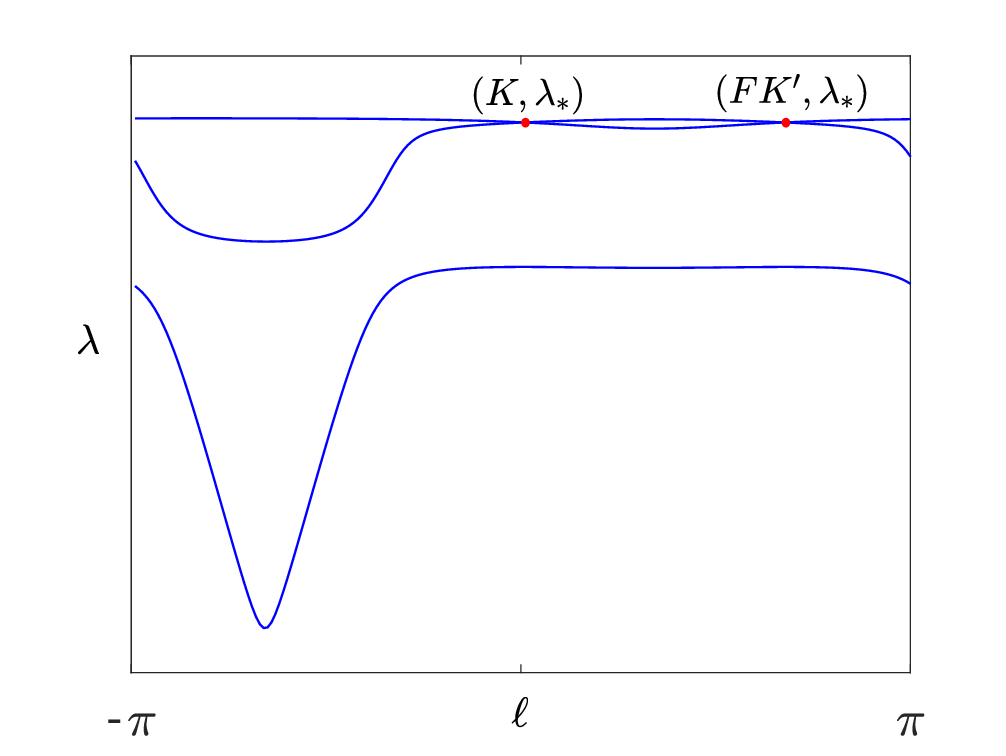}
\includegraphics[width=5cm]{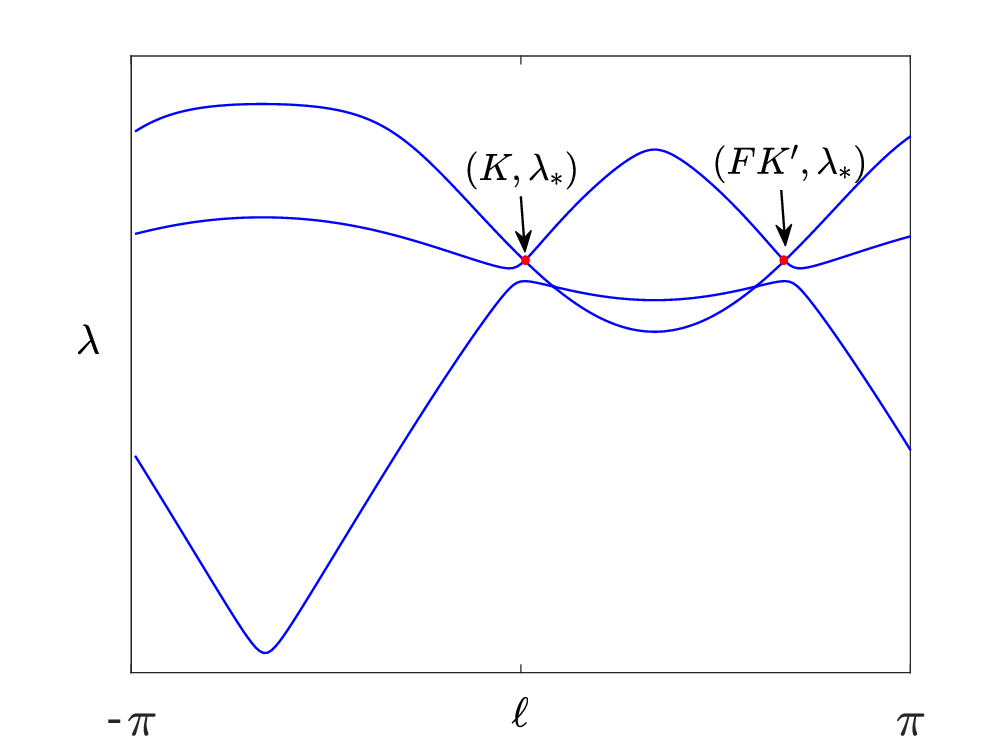}
\caption{The spectral band of $\cL_0$
when $\bp \in \{\Kone+\ell\bbeta$, $\ell\in [-\pi, \pi]\}$ for a photonic crystal consisting of an infinite array of inclusions shown in Figure \ref{fig:periodic_media} (left). The
material tensor $\A(0,0;\bx) = a\, \chi_D(\bx)$ for $\bx\in\cC_z$, where $D\in\cC_z$ is a Lipchitz domain that in invariant under $R$ and $\rflc$.
Top: $\bbeta=\bbeta_1:=(\frac{1}{\sqrt3},-1)^T$; 
Bottom: $\bbeta=\bbeta_1^a:=(0, -2)^T$.
The no-fold condition holds for the spectrum in  the first two columns when $a=\frac{1}{30}$ and $\frac{1}{100}$ respectively. The no-fold condition does not hold for the spectrum in  the last column when $a=\frac{1}{15}$.
}
\vspace*{-20pt}
\label{fig:no_fold}
\end{center}
\end{figure}

\begin{remark}
 Numerical evidence indicates that  
 Assumption~\ref{lem:assNoFold} holds for photonic crystals consisting of reasonably high-contrast dielectric materials; see Figure \ref{fig:no_fold}. In Theorem \ref{thm:edge} and Corollary \ref{thm:edgearm}, we assume that the no-fold condition holds when $\bbeta=\bbeta_1=(\frac{1}{\sqrt3},-1)^T$ and $\bbeta=\bbeta_1^a:=(0, -2)^T$, respectively. 
\end{remark}

The periodic elliptic operator $\cL(\eps,\delta)$ defined in \eqref{eq:pert} is perturbed from $\cL_0$.
At the leading order, the perturbations of the tensor $\A(\eps,\delta;\bx)$ in $\eps$ and $\delta$ are described by the derivatives
\begin{equation}\label{eq:BC}
\C(\bx):= \partial_\eps \A(\eps,0;\bx)\big|_{\eps=0},\quad 
\B(\bx):= \partial_\delta \A(0,\delta;\bx)\big|_{\delta=0}.
\end{equation}
The matrices $C$ and $B$ preserve the Hermiticity and $2\pi/3-$rotational symmetry as summarized in Assumption~\ref{ass:symFull}. We further assume that $C$ and $B$ attain the following properties:

\begin{assumption}\label{ass:symBC}
The matrices $\C,\B\in L^{\infty}(\mathbb R^2, \mathbb M_{2\times2})$ satisfy
\begin{enumerate}
\item \textbf{Reflection anti-invariance:} 
$
\C(\rflc^{-1}\bx) = -\rflc^{-1}\C(\bx) \rflc$, $\B(\rflc^{-1}\bx) = -\rflc^{-1}\B(\bx) \rflc$.
 
\item \textbf{Time-reversal invariance/anti-invariance:} 
$\overline{\C} = \C$, $\overline{\B} = -\B$. 
\end{enumerate}
\end{assumption}

\subsubsection{The interface modes for the joint operators}\label{sec:interfaces}
Let $\gamma:=\{t\be_2: t\in\mathbb{R}\}$ be the line along $\be_2$. We consider an elliptic operator defined as follows over the left and right side of $\gamma$:
\begin{equation}\label{eq:edge_operator}
    \Lint(\eps_L,\delta_L,\eps_R,\delta_R):= \left\{
    \begin{matrix}
        \cL(\eps_L,\delta_L), & \bx\cdot\be_2^\perp<0; \\
        \cL(\eps_R,\delta_R), &  \bx\cdot\be_2^\perp>0.
    \end{matrix}
    \right.
\end{equation}
The operator $\Lint(\eps_L,\delta_L,\eps_R,\delta_R)$ is associated with the joint medium formed by gluing two photonic crystals with the coefficients $ \A(\eps_L,\delta_L;\cdot)$ and $\A(\eps_R,\delta_R;\cdot)$ along a zigzag interface (see Figure \ref{fig:periodic_media}, right).
It is convenient to consider weak solutions associated with the operator $\Lint(\eps_L,\delta_L,\eps_R,\delta_R)$. For this purpose, we introduce the following sesquilinear form
\begin{equation}\label{eq:interfaceopgen}
\begin{aligned}
\langle v, \Lint(\eps_L,\delta_L,\eps_R,\delta_R)  u\rangle_{\mathbb R^2} := & \int_{\bx\cdot\be_2^\perp<0} \overline{ \nabla v(\bx)}\cdot  \A(\eps_L,\delta_L;\bx) \nabla  u(\bx))\,d\bx \\
& +
\int_{\bx\cdot\be_2^\perp>0} \overline{ \nabla v(\bx)}\cdot  \A(\eps_R,\delta_R;\bx) \nabla  u(\bx))\, d\bx.    
\end{aligned}
\end{equation}

\begin{figure}[!htbp]
\begin{center}
\includegraphics[height=4.8cm]{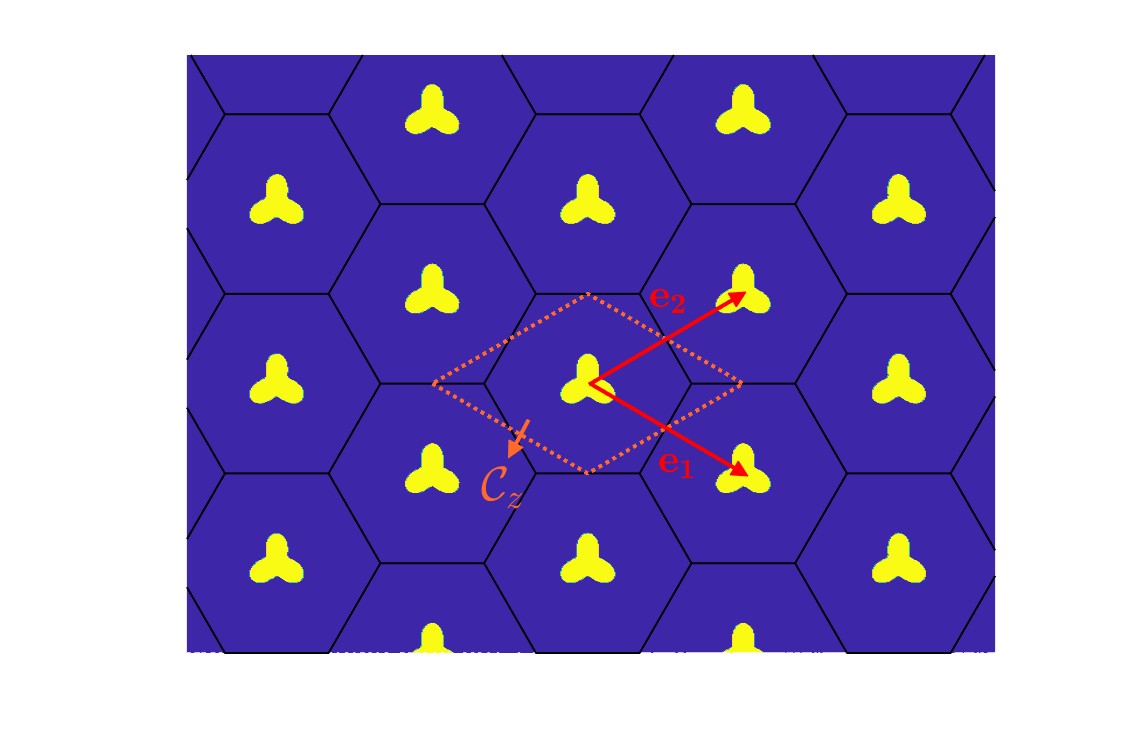} \hspace*{-1cm}
\includegraphics[height=5.2cm]{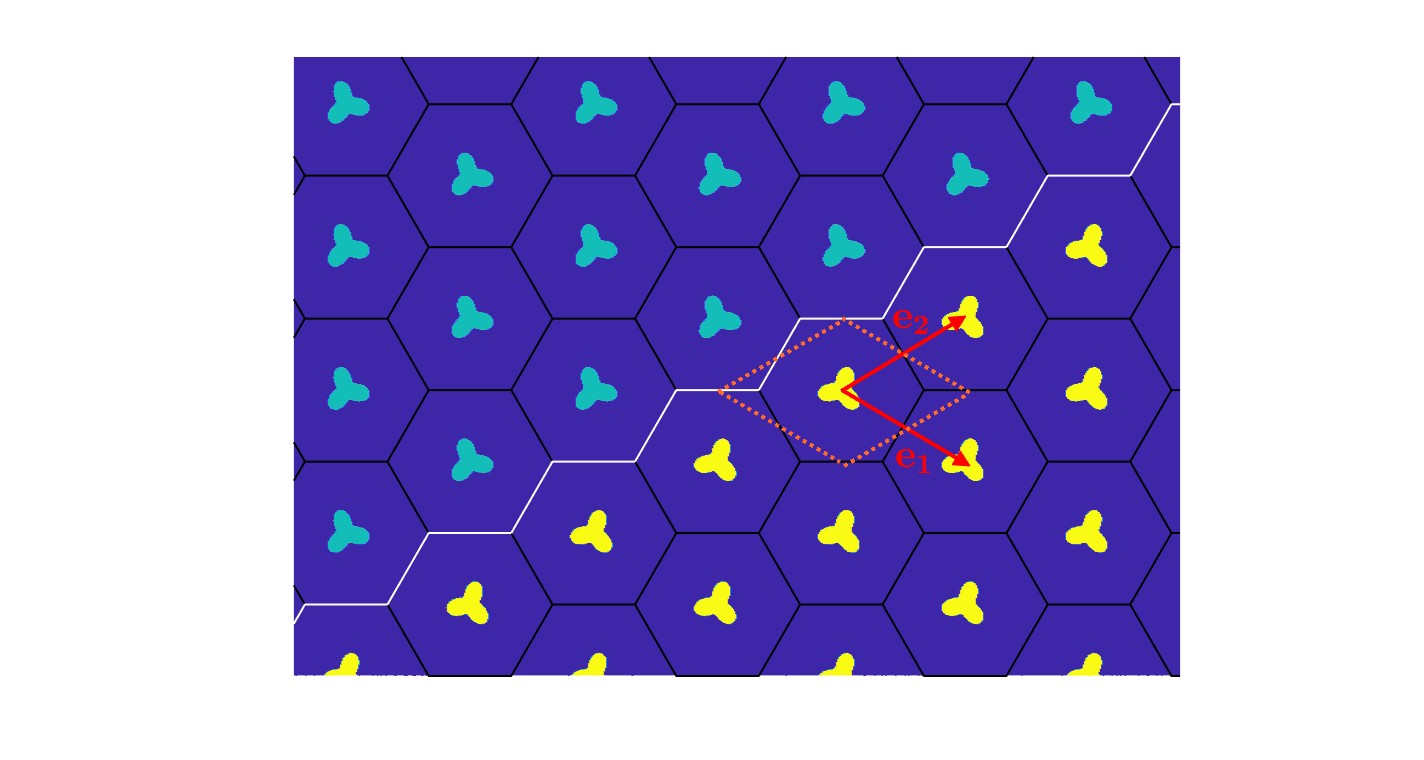}
\end{center}
\vspace*{-20pt}
\caption{Left: The medium for the elliptic operator $\cL_0:=\cL(0,0)$, for which the tensor $\A_0(\bx):=A(0,0,;\cdot)$ satisfies Assumptions \ref{ass:symFull} and \ref{ass:symA}.
Right: The medium of the joint honeycomb structure over which the elliptic operator $\Lint(\eps_L,\delta_L,\eps_R,\delta_R)$ is defined. The interface $\gamma$ of two periodic media is along the $\be_2$ direction. 
The perturbations of $\cL_0$, $(\eps_L,\delta_L)$ and $(\eps_R,\delta_R)$, are on either sides of of the interface.
}
\label{fig:periodic_media}
\end{figure}

The interface is periodic along $\be_2$. The unit cell with respect to this periodicity is $\Omega:= \bigcup_{m\in\mathbb Z}(\cC_z + m\be_1)$, an infinite strip region along the $\be_1$ direction. Define the function space
\begin{equation}\label{eq:edgespace}
H:=\left\{u\in H^1_{\text{loc}}(\mathbb R^2):
u\in H^1(\Omega)\right\}.
\end{equation}
An \emph{interface mode} for the operator $\Lint(\eps_L,\delta_L,\eps_R,\delta_R)$ is a function $u\in H$ satisfying
\begin{equation}\label{eq:edgedef}
\begin{cases}
\Lint(\eps_L,\delta_L,\eps_R,\delta_R) u = \lambda u,\\
u(\bx+\be_2)=e^{i\kp}u(\bx)
\end{cases}
\end{equation}
for some $\lambda,\kp\in\mathbb R$. In the above, $\kp$ is the quasi-momentum along the interface direction $\be_2$ and $\lambda$ is the eigenvalue. Note that for $u\in H$, there holds $|u(\bx)|\to0$
as $|\bx\cdot\be_1|\to\infty$. Namely, an interface mode propagates long the interface direction $\be_2$ while decays along the transverse direction.

Throughout the work, for clarity, we assume that the perturbation parameters satisfy $\eps_j\cdot\delta_j = 0$ for $j=L,R$.
The main result of this work is summarized as follows:

\begin{thm}\label{thm:edge} 
Let  $\kp^*:=K\cdot\be_2$ and  $\mathfrak d$ be an arbitrary constant in $(0,1)$.
Suppose Assumption~\ref{lem:assNoFold} hold along $\bbeta_1$, and the two constants $t_1$ and $t_2$ defined in Proposition~\ref{lem:Tderiv} are nonzero. 

\noindent
\begin{enumerate}
    \item [Case 1.] $(\eps_L,\delta_L)=(\eps, 0)$ and $(\eps_R,\delta_R)=(-\eps, 0)$: If $\kp=\kp^*$ (or $\kp=-\kp^*$), there exists exactly one interface mode in $H$ with the eigenvalue $\lambda\in (\lambda_* - \mathfrak d|t_1\eps|, \lambda_* + \mathfrak d|t_1\eps|)$ when $|\eps|\ll 1$.

    \item [Case 2.] $(\eps_L,\delta_L)=(0, \delta)$ and $(\eps_R,\delta_R)=(0, -\delta)$:  If $\kp=\kp^*$ (or $\kp=-\kp^*$), there exists exactly one interface mode in $H$ with the eigenvalue $\lambda\in (\lambda_* - \mathfrak d|t_2\delta|, \lambda_* + \mathfrak d|t_2\delta|)$ when $|\delta|\ll 1$.

    \item [Case 3.] $(\eps_L,\delta_L)=(0, \delta)$ and $(\eps_R,\delta_R)=(\eps, 0)$: Suppose $\eps$ and $\delta$ are sufficiently small while $\rho:=\frac{t_1\eps}{t_2\delta}$ is fixed. If $\rho>0$, there is no interface mode at $\kp=\kp^*$, but there is exactly one interface mode at $\kp=-\kp^*$ with an eigenvalue $\lambda\in (\lambda_* - \mathfrak d|t_1\eps|, \lambda_* + \mathfrak d|t_1\eps|)\cap (\lambda_* - \mathfrak d|t_2\delta|, \lambda_* + \mathfrak d|t_2\delta|)$.
    If $\rho<0$, there is exactly one edge mode at $\kp=\kp^*$ with an eigenvalue $\lambda\in (\lambda_* - \mathfrak d|t_1\eps|, \lambda_* + \mathfrak d|t_1\eps|)\cap (\lambda_* - \mathfrak d|t_2\delta|, \lambda_* + \mathfrak d|t_2\delta|)$, but there is no edge mode at $\kp=-\kp^*$. 
\end{enumerate}
All the above eigenvalues satisfy $\lambda= \lambda_* + o(\max(|\eps|,|\delta|)$.
\end{thm}

\begin{figure}[!htbp]
\begin{center}
\includegraphics[height=5cm]{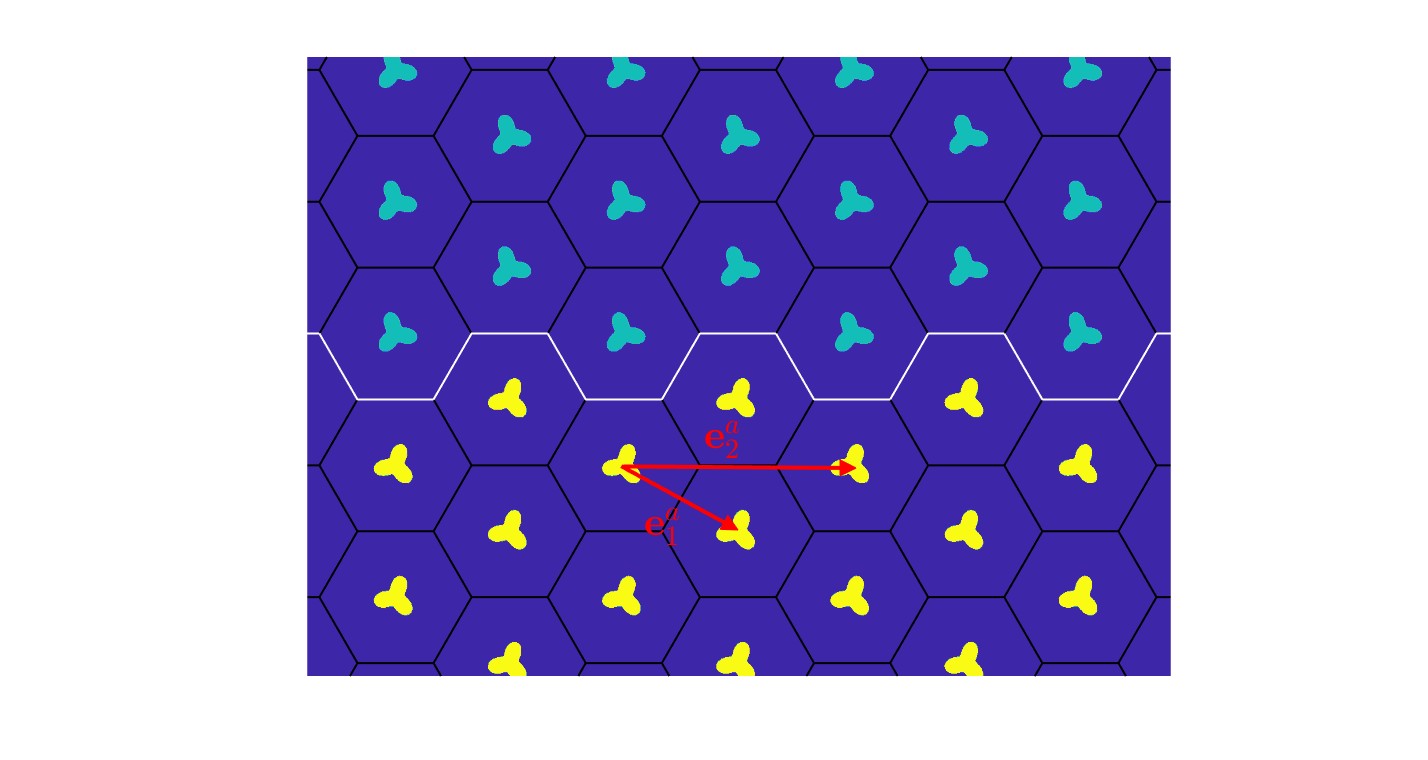}
\end{center}
\vspace*{-20pt}
\caption{The medium of the joint honeycomb structure with an armchair interface along the $\be_2^a$ direction.}
\label{fig:joint_medium_armchair}
\end{figure}

\medskip
The method for the analysis of interface modes along the zigzag interface can be extended to the joint medium with an armchair interface as shown in Figure \ref{fig:joint_medium_armchair}. To this end, let us express the honeycomb lattice as $\Lambda := \bbZ \be_1^a \oplus  \bbZ \be_2^a$, where the new lattice vectors 
\begin{equation*}
\be_1^a=\be_1=(\frac{\sqrt3}{2},-\frac{1}{2})^T, \quad \be_2^a:=\be_1+\be_2=(\sqrt{3},0)^T.
\end{equation*}
The armchair interface direction is along $\be_2^a$, hence the corresponding interface operator and the spectral problem are expressed in the form of \eqref{eq:edge_operator} and \eqref{eq:edgedef}, with $\be_2$ being replaced by $\be_2^a$. 
Note that $\Kone\cdot\be_2^a=2\pi$ and $\Ktwo \cdot\be_2^a = -2\pi$, which are both equivalent to $0$ by the periodicity of $e^{\im\kp}$. 
Define $\kp^{*,a}:=0$. Then the edge modes along the armchair interface at $\kp^{*,a}$ are given in the following corollary:
\begin{coro}\label{thm:edgearm} 
Let  $\kp^{*,a}=0$ and  $\mathfrak d$ be an arbitrary constant in $(0,1)$.
Suppose Assumption~\ref{lem:assNoFold} hold along $\bbeta_1$, and the two constants $t_1$ and $t_2$ defined in Proposition~\ref{lem:Tderiv} are nonzero.

\noindent
\begin{enumerate}
    \item [Case 1.] $(\eps_L,\delta_L)=(\eps, 0)$ and $(\eps_R,\delta_R)=(-\eps, 0)$: If $\kp=\kp^{*,a}$, there exist exactly two interface modes in $H$ with the eigenvalue $\lambda\in (\lambda_* - \mathfrak d|t_1\eps|, \lambda_* + \mathfrak d|t_1\eps|)$ when $|\eps|\ll 1$.

    \item [Case 2.] $(\eps_L,\delta_L)=(0, \delta)$ and $(\eps_R,\delta_R)=(0, -\delta)$:  If $\kp=\kp^{*,a}$, there exist exactly two interface modes in $H$ with the eigenvalue $\lambda\in (\lambda_* - \mathfrak d|t_2\delta|, \lambda_* + \mathfrak d|t_2\delta|)$ when $|\delta|\ll 1$.

    \item [Case 3.] $(\eps_L,\delta_L)=(0, \delta)$ and $(\eps_R,\delta_R)=(\eps, 0)$: Suppose $\eps$ and $\delta$ are sufficiently small while $\rho:=\frac{t_1\eps}{t_2\delta}$ is fixed. There  exists exactly one interface mode at $\kp=\kp^{*,a}$ with an eigenvalue $\lambda\in (\lambda_* - \mathfrak d|t_1\eps|, \lambda_* + \mathfrak d|t_1\eps|)\cap (\lambda_* - \mathfrak d|t_2\delta|, \lambda_* + \mathfrak d|t_2\delta|)$. 
\end{enumerate}
All the above eigenvalues satisfy $\lambda= \lambda_* + o(\max(|\eps|,|\delta|)$.
\end{coro}
For clarity of the presentation, in the rest of the paper, we present the spectral analysis for \eqref{eq:edgedef} and prove Theorem \ref{thm:edge}. The proof of Corollary \ref{thm:edgearm} follows the same lines.

\subsection{An example of the material fulfilling the assumptions}
Let us discuss one realization of the periodic media described above. The photonic crystal consists of an infinite array of dielectric inclusions that are embedded in a homogeneous background. More specifically, let \( D\subset \cC_z \) be a Lipchitz domain that is invariant under the rotation $R$ and the reflection~$\rflc$ such that $R(D)=D$ and $F(D)=D$. Consider the material tensor 
\begin{align*}
  & \A(\eps,\delta;\bx)=a\, \chi_D\big(R^{\eps}\bx\big) + \delta\, \begin{pmatrix}0 & -\im \\ \im & 0\end{pmatrix} \chi_D(\bx) \quad\mbox{for} \quad \bx \in \C_z, \\
  & \A(\eps,\delta;\bx+\be) = \A(\eps,\delta;\bx) \quad\mbox{for} \; \be \in \Lambda.
\end{align*}
In the above, \(a\) is a positive constant representing the dielectric constant of the inclusion, 
\(\chi_D\) is the characteristic function of the region \(D\), i.e., \(\chi_D(\bx)=1\) when \(\bx\in D\) and \(0\) otherwise, and \(R^\eps\) denotes the rotation matrix describing a clockwise rotation by an angle \(\eps\), given by
  \[
  R^\eps = \begin{pmatrix} \cos \eps & \sin \eps \\ -\sin \eps & \cos \eps \end{pmatrix}.
  \]
In particular, when there is no perturbation (\(\eps=\delta=0\)), we have
\[
\A(0,0;\bx) = a\, \chi_D(\bx), \quad\mbox{for} \; \bx \; \in \C_z;
\quad
\A(0,0;\bx+\be) = \A(0,0;\bx) \quad\mbox{for} \; \be \in \Lambda.
\]
This corresponds to a photonic crystal with dielectric inclusions defined by \( \bigcup_{\be\in\Lambda} (D+\be)\). The two types of perturbations are given as follows:
\begin{enumerate}
    \item 
The \(\eps\)-perturbation (geometric perturbation): 
   The term 
   \[
   a\, \chi_D\big(R^{\eps}\bx\big)
   \]
   represents a rotation of the inclusions by an angle of \(\eps\). This means that at the leading order, a small rotation is described by
   \[
   \C(\bx) = \frac{d}{d\eps}\A(\eps,0;\bx)\Big|_{\eps=0},
   \]
   which, by the chain rule, reads
   \[
   \C(\bx) = \nabla\A(0,0;\bx)\cdot \J\bx.
   \]
   In the above,
   \[
   \J=\begin{pmatrix}0 & 1 \\ -1 & 0\end{pmatrix}
   \]
   is the standard rotation matrix associated with an infinitesimal rotation. In other words, for each entry of \(\A(0,0;\bx)\), we have
   \[
   \C_{k,l}(\bx)= \nabla \A_{k,l}(0,0;\bx)\cdot \J\bx.
   \]

\item 
The \(\delta\)-perturbation (opto-magnetic perturbation):  
   The term
   \[
   \delta\, \begin{pmatrix}0 & -\im \\ \im & 0\end{pmatrix} \chi_D(\bx)
   \]
   a change in the material properties by doping the inclusions with opto-magnetic materials and applying a magnetic field whose strength is proportional to \(\delta\)~\cite{Wang-09}. At the first order, this yields the perturbation tensor
   \[
   \B(\bx) = \chi_D(\bx)\, \begin{pmatrix}0 & -\im \\ \im & 0\end{pmatrix}.
   \]
\end{enumerate}

To summarize,  the \(\eps\)-perturbation rotates the inclusions with the action of \( R^\eps \), with the leading-order described by \(\C\). The perturbation breaks the reflection symmetry of the elliptic operator.
The \(\delta\)-perturbation introduces an anti-Hermitian contribution reflecting a magneto-optical doping that is represented by \(\B\). This perturbation breaks the time-reversal and the reflection symmetry.

\begin{lemma}\label{lem:realization}
The matrices $\A$, $\B$ and $\C$ satisfy Assumptions~\ref{ass:symA} and \ref{ass:symBC}.
\end{lemma}
The proof of the lemma is provided in Appendix~\ref{sec:realization}.

\subsection{Outline}
The rest of the paper is organized as follows. In Section 2, we study the band structure of the periodic wave operator $\cL(\eps, \delta)$. In particular, we discuss the lift of the Dirac degeneracy and the spectral gap opening at the high-symmetry points of the Brillouin zone when $\eps$ or $\delta$ is nonzero. In Section 3, we introduce the Green functions in the infinite strip region for the perturbed periodic photonic structures and present its asymptotic behavior near the Dirac point. Section 4 is devoted to the interface modes for the joint wave operator $\Lint(\eps_L,\delta_L,\eps_R,\delta_R)$. We formulate the boundary integral equation for the eigenvalue problem and perform the asymptotic analysis of integral operators to prove Theorem \ref{thm:edge}.

\section{Band structure for the periodic operators and Berry curvature} 
\subsection{Floquet-Bloch theory}
The spectrum of the elliptic operator $\cL(\eps,\delta)$ is analyzed following the standard Floquet-Bloch theory. We give a brief summary as follows and refer to \cite{Kuchment-12} for more details.
To this end, let us define the space of functions with the quasi-momentum $\bp\in\cB$ as follows: 
\begin{equation*}
H^1_{\bp}:= \{ u\in H^1_{\text{loc}}(\mathbb R^2): u(\bx+\be_i) = e^{\im \bp\cdot \be_i} u(\bx)\}.
\end{equation*} 
There holds $u\in H^1_{\bp}$ if and only if $\tilde u (\bx):= e^{- \im \bp\cdot \bx} u(\bx) \in H^1_{\mathbf{0}}$. 
Introduce the self-adjoint operator $\cL(\eps,\delta,\bp): H^1_{\mathbf{0}} \times H^1_{\mathbf{0}}\to\mathbb C$ defined as
\begin{equation}\label{eq:twoquasi}
\cL(\eps,\delta,\bp) :=  - (\nabla + \im\bp)\cdot\A(\eps,\delta;\bx) (\nabla + \im\bp).
\end{equation}
Then for $u,v\in H^1_{\bp}$ and $ \tilde u(\bx) =e^{-\im \bp\cdot \bx} u(\bx)$ and $ \tilde v(\bx) =e^{-\im \bp\cdot \bx} v(\bx)$, we have
\begin{equation}\label{eq:L_and_Lp}
\begin{aligned}
\langle v, \cL(\eps,\delta)  u\rangle_{\cC_z} &:= \int_{\cC_z} \overline{ \nabla v(\bx)}\cdot \A(\eps,\delta;\bx) \nabla  u(\bx))\,d\bx \\
&= \int_{\cC_z} \overline{ (\nabla + \im \bp) \tilde v(\bx)}\cdot \A(\eps,\delta;\bx) (\nabla + \im \bp)\tilde u(\bx) )\,d\bx \\
&=: \langle \tilde v, \cL(\eps,\delta,\bp) \tilde u\rangle_{\cC_z}.
\end{aligned}
\end{equation}
The spectral band structure of the operator $\cL(\eps,\delta)$ is given by all the eigenpair $(\bp,\lambda)\in\cB\times\mathbb R$, for which there is a solution $(\lambda,u)\in\mathbb R\times H^1_{\bp}$ to the eigenvalue problem
\begin{equation}\label{eq:disp_zig}
\langle v, \cL(\eps,\delta) u\rangle_{\cC_z} = \lambda \langle v, u\rangle_{\cC_z} \quad \forall v\in H^1_{\bp}.
\end{equation}
In view of \eqref{eq:L_and_Lp}, this is equivalent to finding the eigenpair $(\bp,\lambda)\in\cB\times\mathbb R$ such that there exists $\tilde u\in H^1_{\mathbf 0}$, such that
\begin{equation}\label{eq:disp_zig2}
\langle \tilde v, \cL(\eps,\delta,\bp) \tilde u\rangle_{\cC_z} = \lambda \langle \tilde v, \tilde u\rangle_{\cC_z} \quad \forall \tilde v\in H^1_{\mathbf{0}}.
\end{equation}

For each $\bp\in\cB$, from the spectral theory of self-adjoint elliptic operators, the spectrum of $\cL(\eps,\delta,\bp)$ consists of a countable set of real eigenvalues $\lambda_{n,\eps,\delta}(\bp)$ (band functions) that are ordered increasingly:
\begin{equation}
\lambda_{1,\eps,\delta}(\bp)\leq \lambda_{2,\eps,\delta}(\bp)\leq \lambda_{3,\eps,\delta}(\bp)\leq\cdots.
\end{equation}
The corresponding eigenfunction $u_{n,\eps,\delta}(\bx;\bp)$ are normalized such that $\| u_{n,\eps,\delta} \|_{L^2(\cC_z)}=1$.
Here $\lambda_{n,\eps,\delta}(\bp)$ and $u_{n,\eps,\delta}(\cdot,\bp)$ are chosen to be piecewise analytic in $\bp$.

In the rest of this section, we first show when $\eps=\delta=0$, the band structure of $\cL_0$ has linear degeneracies at $K$ and $K'$, the high-symmetry points of the Brillouin zone. More precisely, two dispersion surfaces cross linearly at $\bp=K$ or $K'$, forming Dirac cones. Next we show when $\varepsilon\neq0$ or $\delta\neq0$, the Dirac degeneracies are lifted in the operator $\cL(\varepsilon,\delta)$ describing the perturbed medium.

\subsection{Dirac point of the operator $\cL_0$}
Define 
\begin{equation}
H_{\bp,n}^1:=\{u\in H^1_{\text{loc}}(\mathbb R^2): u(\bx+\be_i) = e^{\im \bp\cdot \be_i} u(\bx), \cR u(\bx) = \tau^n u(\bx)\}, \quad \tau = e^{\im2\pi/3}.
\end{equation}
Due to the symmetries of the operator $\cL_0$, Dirac points exist at the high symmetry points $\Kone$ and $\Ktwo$ under the following assumption.  
\begin{assumption}[Non-degeneracy condition]\label{ass:Dirac}
Suppose the following conditions are satisfied.
\begin{enumerate}
\item
$\lambda_*\in\mathbb R$ is a simple eigenvalue of $\cL_0$ restricted to $H_{K,1}^1$, with the eigenfunction $w_1(\bx)\in H_{K,1}^1$. Without loss of generality, $w_1$ can be chosen with $\|w_1\|_{L^2(\cC_z)}=1$.
\item 
$\lambda_*$ is a simple eigenvalue of $\cL_0$ restricted to $H_{K,2}^1$, with the eigenfunction $w_2(\bx)\in H_{K,2}^1$. Further more, $w_2(\bx)=Fw_1(\bx):=w_1(\rflc\bx)$. 
\item 
$\lambda_*$ is not an eigenvalue of $\cL_0$ restricted to $H_{K,0}^1$.
\item There holds
\begin{equation}\label{eq:slope}
m_*:=\frac{1}{2}\left|\overline{\left\langle w_1,(A_0\frac{1}{\text{\normalfont i}}\nabla + \frac{1}{\text{\normalfont i}}\nabla \cdot A_0)w_2\right\rangle}_{\cC_z}\cdot\left(
\begin{matrix}
1\\ \text{\normalfont i}     
\end{matrix}
\right)\right|\neq0.
\end{equation}
\end{enumerate}
\end{assumption}

Under Assumption~\ref{ass:Dirac}, we have
the following theorem on the existence of Dirac points. The local behavior of the dispersion relations and the eigenspaces near the Dirac points are also characterized. The proof follows from the perturbation argument and the Fredholm alternative. For conciseness of presentation, we refer to ~\cite{Berkolaiko-18, Lee-Thorp-Weinstein-Zhu-19, Cassier-Weinstein-21,Li2023} for details. 
\begin{thm}\label{lem:Dirac}
Suppose Assumption~\ref{ass:Dirac} holds. Then $(\Kone,\lambda_*)$ is a Dirac point of $\cL_0$. The dispersion relation \eqref{eq:disp_zig} of $\cL_0$ takes the form  
\begin{equation*}
(\lambda-\lambda_*)^2 = m_*^2 \, |\bp-\Kone|^2+O(|\bp-\Kone|^3),
\end{equation*}
for $m_*\in\mathbb R$ with $m_*>0$ as defined in \eqref{eq:slope}.
In addition, the basis of the eigenspace at the Dirac point $(\Kone,\lambda_*)$ can be chosen as $w_1$ and $w_2$ that satisfy $\|w_i\|_{L^2(\cC_z)=1}$ and 
\begin{equation*}
Rw_1(\bx):=w_1(R^{-1}\bx) = \tau w_1(\bx),\quad Rw_2(\bx):=w_2(R^{-1}\bx) = \overline{\tau}  w_2(\bx),\quad w_2(\bx)=Fw_1(\bx):=w_1(\rflc\bx),
\end{equation*}
in which $\tau = e^{\im\frac{2\pi}{3}}$. 

$(\Ktwo,\lambda_*)$ is also a Dirac point of $\cL_0$. 
The dispersion relation \eqref{eq:disp_zig} of $\cL_0$ takes the form  
\begin{equation*}
(\lambda-\lambda_*)^2 = m_*^2 \, |\bp-\Ktwo|^2+O(|\bp-\Ktwo|^3).
\end{equation*}
The eigenspace at the Dirac point $(\Kone,\lambda_*)$ is  spanned by
\begin{equation}\label{eq:onetwo}
w_1'(\bx) := \bar w_2(\bx),\quad w_2'(\bx):=\bar w_1(\bx),
\end{equation}
which attain the following symmetries:
\begin{equation*}
Rw_1'(\bx):=w_1'(R^{-1}\bx) = \tau w_1'(\bx),\quad Rw_2'(\bx):=w_2'(R^{-1}\bx) = \overline{\tau}w_2'(\bx),\quad w_2'(\bx)=w_1'(\rflc\bx).
\end{equation*}

\end{thm}

\subsection{Lifting of Dirac degeneracy for the operator $\cL(\varepsilon, \delta)$ and Berry curvature}\label{sec:bandstructure}

\subsubsection{Bandgap opening near the Dirac points for $\cL(\varepsilon,\delta)$}\label{sec:gap}
Let $(\bp, \lambda_*)$ be the Dirac point of the operator $\cL_0$ with $\bp=\Kone$ or $\Ktwo$. Without loss of generality, we assume that the Dirac point is formed by $n_0$-th and $(n_0+1)$-th dispersion surfaces of the operator. For ease of notation, let us label the eigenpair $(\lambda_{n,\eps,\delta}(\bp),u_{n,\eps,\delta}(\bp))$ by $(\lambda_{\tilde n,\eps,\delta}(\bp), u_{\tilde n,\eps,\delta}(\bp))$ in the following manner: 
\begin{equation*}
\begin{aligned}
&\lambda_{\tilde 1,\eps,\delta}(\bp) = \lambda_{n_0,\eps,\delta}(\bp), \quad \lambda_{\tilde 2,\eps,\delta}(\bp)= \lambda_{n_0+1,\eps,\delta}(\bp), \\
&\lambda_{\tilde n,\eps,\delta}(\bp) = \lambda_{n+2,\eps,\delta}(\bp) \text{ for } n< n_0, \\ &\lambda_{\tilde n,\eps,\delta}(\bp)= \lambda_{n,\eps,\delta}(\bp)  \text{ for } n> n_0+1.
\end{aligned}
\end{equation*}
Here and henceforth, we will adopt the labeling $(\lambda_{\tilde n,\eps,\delta}(\bp), u_{\tilde n,\eps,\delta}(\bp))$ above for the spectral bands of $\cL(\eps,\delta)$. However, with an abuse of notations, we drop the tildes and still denote them as $(\lambda_{n,\eps,\delta}(\bp),u_{n,\eps,\delta}(\bp))$.

\begin{prop}\label{lem:Tderiv}
Define $\tilde\bw:= (e^{-\im\Kone\cdot\bx} w_1, e^{-\im\Kone\cdot\bx} w_2)$ and $\tilde\bw':= (e^{-\im\Ktwo\cdot\bx} w_1', e^{-\im\Ktwo\cdot\bx} w_2')$. 
At $\Kone$, there exist $t_1, t_2 \in \mathbb{R}$ and $\theta_* \in \mathbb{C}$ such that
 \begin{subequations} \label{eq:Tderiv1}
\begin{align}
\langle\tilde \bw,  \bbeta_1\cdot\nabla_{\bp}\cL(0,0,\bp)\tilde\bw\rangle_{\cC_z}|_{\bp=\Kone}&= 
\left(\begin{matrix}
0&\overline{\theta_*}\\
\theta_*&0
\end{matrix}\right),  
\label{eq:w_b1_dpL_w_K}
\\
\langle\tilde\bw,  \bbeta_2\cdot\nabla_{\bp}\cL(0,0,\bp)\tilde\bw\rangle_{\cC_z}|_{\bp=\Kone}&= 
\left(\begin{matrix}
0&\tau\overline{\theta_*}\\
\overline{\tau}\theta_*&0
\end{matrix}\right),
\label{eq:w_b2_dpL_w_K}
\\
\langle \tilde\bw, \partial_{\eps}\cL (\eps,0,\Kone)\tilde\bw\rangle_{\cC_z}|_{\eps=0}&= 
\left(\begin{matrix}
t_1&0\\
0&-t_1
\end{matrix}\right),
\label{eq:w_depsL_w_K}
\\
\langle \tilde\bw, \partial_{\delta}\cL (0,\delta,\Kone)\tilde\bw\rangle_{\cC_z}|_{\delta=0}&= 
\left(\begin{matrix}
t_2&0\\
0&-t_2
\end{matrix}\right).
\label{eq:w_ddeltaL_w_K}
\end{align}
\end{subequations}

At $\Ktwo$,
\begin{subequations}\label{eq:Tderiv2}
\begin{align}
\langle \tilde\bw',  \bbeta_1\cdot\nabla_{\bp}\cL(0,0,\bp)\tilde\bw'\rangle_{\cC_z}|_{\bp=\Ktwo}&= 
\left(\begin{matrix}
0&-\overline{\theta_*}\\
-\theta_*&0
\end{matrix}\right),
\label{eq:w_b1_dpL_w_K'}
\\
\langle\tilde\bw',  \bbeta_2\cdot\nabla_{\bp}\cL(0,0,\bp)\tilde\bw'\rangle_{\cC_z}|_{\bp=\Ktwo}&= 
\left(\begin{matrix}
0&-\tau\overline{\theta_*}\\
-\overline{\tau}\theta_*&0
\end{matrix}\right), \label{eq:w_b2_dpL_w_K'}
\\
\langle \tilde\bw', \partial_{\eps}\cL (\eps,0,\Ktwo)\tilde\bw'\rangle_{\cC_z}|_{\eps=0}&= 
\left(\begin{matrix}
-t_1&0\\
0&t_1
\end{matrix}\right), 
\label{eq:w_depsL_w_K'}
\\
\langle \tilde\bw', \partial_{\delta}\cL (0,\delta,\Ktwo)\tilde\bw'\rangle_{\cC_z}|_{\delta=0}&= 
\left(\begin{matrix}
t_2&0\\
0&-t_2
\end{matrix}\right).
\label{eq:w_ddeltaL_w_K'}
\end{align}
\end{subequations}

\end{prop}

\begin{proof}
By the definition of $\cL(\eps,\delta,\bp)$ in \eqref{eq:twoquasi}, we have
\begin{equation}\label{eq:dL}
\begin{aligned}
\nabla_{\bp}\cL(0,0,\bp)|_{\bp=\Kone} &= - \im\A(\bx) (\nabla + \im \Kone) - \im \big((\nabla + \im \Kone)\cdot \A(\bx)\big)^t ,\\
\partial_{\eps}\cL (\eps,0,\Kone)|_{\eps=0} &= - (\nabla + \im \Kone) \cdot \C(\bx) (\nabla + \im\Kone),\\
\partial_{\delta}\cL (0,\delta,\Kone)|_{\delta=0} &= - (\nabla + \im \Kone) \cdot \B(\bx) (\nabla + \im \Kone).
\end{aligned}
\end{equation}
These relations hold when $\Kone$ is replaced by $\Ktwo$.

We first consider $\bp=\Kone$. Using \eqref{eq:dL}, we have
\begin{equation}
\langle \tilde w_i,  \partial_{\delta}\cL (0,\delta,\Kone) \tilde w_j\rangle_{\cC_z}|_{\delta=0} = \int_{\cC_z} (\overline{\nabla w_i(\bx)})^t \B(\bx) \nabla w_j(\bx) \, d\bx.
\end{equation}
It follows from the Hermiticity of Assumption~\ref{ass:symFull} that
\begin{equation}\label{eq:deltaherm}
\begin{aligned}
\overline{\langle \tilde w_i,  \partial_{\delta}\cL (0,\delta,\Kone) \tilde w_j\rangle_{\cC_z}|_{\delta=0}}  = \int_{\cC_z} (\overline{\nabla w_j(\bx)})^t \B(\bx) \nabla w_i(\bx) \, d\bx = \langle \tilde w_j,  \partial_{\delta}\cL (0,\delta,\Kone) \tilde w_i\rangle_{\cC_z}|_{\delta=0} .
\end{aligned}
\end{equation}
The diagonal entries of $\langle \tilde\bw, \partial_{\delta}\cL (0,\delta,\Kone)\tilde\bw\rangle_{\cC_z}|_{\delta=0}$ are opposite to each other because 
\begin{equation*}
\begin{aligned}
\langle \tilde w_2,  &\partial_{\delta}\cL (0,\delta,\Kone) \tilde w_2\rangle_{\cC_z}|_{\delta=0} 
= \int_{\cC_z} (\overline{\nabla w_2(\bx)})^t \B(\bx) \nabla w_2(\bx) \, d\bx 
= \int_{\cC_z} (\overline{\nabla \big(w_1(\rflc \bx)\big)})^t \B(\bx) \nabla \big(w_1(\rflc\bx)\big) \, d\bx\\
&= \int_{\cC_z} (\overline{\nabla w_1|_{\rflc \bx}})^t \rflc \B(\bx) \rflc  \nabla w_1|_{\rflc\bx} \, d\bx
= - \int_{\cC_z} (\overline{\nabla w_1|_{\rflc \bx}})^t \B(F\bx) \nabla w_1|_{\rflc\bx} \, d\bx \\
&= - \langle \tilde w_1,  \partial_{\delta}\cL (0,\delta,\Kone) \tilde w_1\rangle_{\cC_z}|_{\delta=0} .
\end{aligned}
\end{equation*}

To work with the off-diagonal entries, we notice that due to the $\Lambda$-periodicity of $H^1_{\bf 0}$ and $\cL(\eps,\delta,\bp)$, it holds 
\begin{equation}
\langle \tilde w_2,  \partial_{\delta}\cL (0,\delta,\Kone) \tilde w_1\rangle_{\cC_z}|_{\delta=0}
=\langle \tilde w_2, \partial_{\delta}\cL (0,\delta,\Kone) \tilde w_1\rangle_{E}|_{\delta=0}. 
\end{equation}
The off-diagonal entries are zero because
\begin{equation*} 
\begin{aligned}
\langle \tilde w_2, & \partial_{\delta}\cL (0,\delta,\Kone) \tilde w_1\rangle_{\cC_z}|_{\delta=0} 
= \int_{E} (\overline{\nabla w_2(\bx)})^t \B(\bx) \nabla w_1(\bx) \, d\bx = \int_{E} (\overline{\nabla w_2|_{R^{-1}\bx}})^t \B(R^{-1}\bx) \nabla w_1|_{R^{-1}\bx} \, d\bx\\
&= \int_{E} (\overline{R^{-1}\nabla \big(w_2(R^{-1}\bx)\big)})^t \B(R^{-1}\bx) R^{-1}\nabla \big(w_1(R^{-1}\bx)\big) \, d\bx \\
&= \tau^2\int_{E} (\overline{\nabla \big(w_2(\bx)\big)})^t R\B(R^{-1}\bx) R^{-1}\nabla \big(w_1(\bx)\big) \, d\bx
=\tau^2\int_{E} (\overline{\nabla \big(w_2(\bx)\big)})^t \B(\bx)\nabla \big(w_1(\bx)\big) \, d\bx \\
&=\tau^2 \langle \tilde w_2,  \partial_{\delta}\cL (0,\delta,\Kone) \tilde w_1\rangle_{\cC_z}|_{\delta=0}  .
\end{aligned}
\end{equation*}
Thus \eqref{eq:w_ddeltaL_w_K} holds.

The $\eps$ derivative takes the form
\begin{equation*}
\langle \tilde w_i,  \partial_{\eps}\cL (\eps,0,\Kone) \tilde w_j\rangle_{\cC_z}|_{\eps=0} = \int_{\cC_z} (\overline{\nabla w_i(\bx)})^t \C(\bx) \nabla w_j(\bx) \, d\bx.
\end{equation*}
The same calculations above leads to \eqref{eq:w_depsL_w_K}.

\begin{figure}[!htbp]
    \centering
\includegraphics[width=0.24\linewidth]{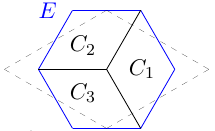}
    \caption{The hexagon $E$ represents another fundamental domain of the honeycomb lattice $\Lambda$ (cf. Figure \ref{fig:periodic_cell}). The vertices of the hexagon $E$ are given by $\frac{\sqrt{3}}{3}(\cos\theta_j, \sin\theta_j)$, wherein $\theta_j=\frac{j\pi}{3}$ for $j= 1, \cdots, 6$. $E$ can be decomposed as $E=C_1\sqcup C_2\sqcup C_3$ shown above, wherein $C_j=R^{-(j-1)}C_1$. It is clear that $E$ invariant under the rotation map $R$. }
    \label{fig:placeholder}
\end{figure}

For the $\bp$ derivatives,
\begin{equation*}
\begin{aligned}
\langle \tilde w_i,  \nabla_{\bp}\cL (0,0,\bp) \tilde w_j\rangle_{\cC_z}|_{\bp=\Kone} 
=\int_{E} -\im \overline{w_i(\bx)} \A(\bx) \nabla w_j(\bx) + \im  w_j(\bx)\A(\bx)^t  \overline{ \nabla w_i(\bx)}\, d\bx\\
= \int_{E} -\im \overline{w_i(\bx)} \A(\bx) \nabla w_j(\bx) + \im  w_j(\bx)\A(\bx)  \overline{ \nabla w_i(\bx)}\, d\bx.
\end{aligned}
\end{equation*}
We first show the diagonal entries are zero.  Let us decompose $E$ as $E = C_1\sqcup C_2\sqcup C_3$, where $\bx\in C_2$ if and only if $R\bx \in C_1$, and  $\bx\in C_3$ if and only if $R^2\bx \in C_1$. We have
\begin{equation*}
\begin{aligned}
\int_{C_2}& \overline{w_i(\bx)} \A(\bx) \nabla w_i(\bx) \, d\bx
=  \int_{C_1} \overline{w_i(R^{-1}\bx)} \A(R^{-1}\bx) \nabla w_i|_{R^{-1}\bx} \, d\bx\\
&= \int_{C_1} \overline{w_i(R^{-1}\bx)} \A(R^{-1}\bx) R^{-1} \nabla \big( w_i(R^{-1}\bx)\big) \, d\bx 
=  \int_{C_1} \overline{w_i(\bx)} R^{-1} \A(\bx) \nabla \big( w_i(\bx)\big) \,d\bx
\end{aligned}
\end{equation*}
Similarly, 
\begin{equation*}
\begin{aligned}
\int_{C_3} \overline{w_i(\bx)} \A(\bx) \nabla w_i(\bx) \, d\bx
=  \int_{C_1} \overline{w_i(\bx)} (R^{-1})^2 \A(\bx) \nabla \big( w_i(\bx)\big) \,d\bx
\end{aligned}
\end{equation*}
Thus
\begin{equation*}
\begin{aligned}
\int_E \overline{w_i(\bx)} \A(\bx) \nabla w_i(\bx) \, d\bx= (I + R^{-1} + (R^{-1})^2)\int_{C_1}& \overline{w_i(\bx)} \A(\bx) \nabla w_i(\bx) \, d\bx =0.
\end{aligned}
\end{equation*}
The off-diagonal entries are complex conjugate of each other because
\begin{equation*}
\overline{\int_{\cC_z} \overline{w_i(\bx)}  \A(\bx) \nabla w_j(\bx) \, d\bx } 
= \int_{\cC_z}  w_i(\bx) \overline{\A(\bx)} \overline{\nabla w_j(\bx)} \, d\bx 
= \int_{\cC_z} w_i(\bx) \A(\bx) \overline{\nabla w_j(\bx)} \,d\bx.
\end{equation*}
Next, using $\bbeta_2 = R^{-1}\bbeta_1$, we have the following relation between the $\bp$ derivatives in the directions of $\beta_1$ and $\beta_2$
\begin{equation*}
\begin{aligned}
& \langle \tilde w_i,  \bbeta_2\cdot\nabla_{\bp}\cL (0,0,\bp) \tilde w_j\rangle_{\cC_z}|_{\bp=\Kone} \\
= & \int_E -\im \overline{w_i(\bx)} (R^{-1}\bbeta_1)^t  \A(\bx) \nabla w_j(\bx) + \im  w_j(\bx) (R^{-1}\bbeta_1)^t  \A(\bx)  \overline{ \nabla w_i(\bx)}\, d\bx\\
= & \int_E -\im \overline{w_i(R^{-1}\bx)} (R^{-1}\bbeta_1)^t  \A(R^{-1}\bx) \nabla w_j|_{R^{-1}\bx} + \im  w_j(R^{-1}\bx) (R^{-1}\bbeta_1)^t  \A(R^{-1}\bx)  \overline{ \nabla w_i|_{R^{-1}\bx}}\, d\bx\\
= & \int_E -\im \overline{w_i(R^{-1}\bx)} (R^{-1}\bbeta_1)^t  \A(R^{-1}\bx) R^{-1}\nabla \big(w_j(R^{-1}\bx)\big) + \im  w_j(R^{-1}\bx) (R^{-1}\bbeta_1)^t  \A(R^{-1}\bx)  \overline{ R^{-1}\nabla \big(w_i(R^{-1}\bx)\big)}\, d\bx\\
= & \tau \int_E -\im \overline{w_i(\bx)} \bbeta_1^t  \A(\bx) \nabla \big(w_j(\bx)\big) + \im  w_j(\bx) \bbeta_1^t  \A(\bx)  \overline{\nabla \big(w_i(\bx)\big)}\, d\bx\\
= & \tau \langle \tilde w_1, \bbeta_1\cdot \nabla_{\bp}\cL (0,0,\bp) \tilde w_2\rangle_{\cC_z}|_{\bp=\Kone} 
\end{aligned}
\end{equation*}

When $\bp = \Ktwo$, using Assumption~\ref{ass:symBC} and following the same procedure above, it can be shown that the matrices in \eqref{eq:Tderiv2} share the same form as \eqref{eq:Tderiv1}. Furthermore, in light of \eqref{eq:onetwo}, 
we obtain 
\begin{align*}
\langle \tilde w_1',  \nabla_{\bp}\cL (0,0,\bp) \tilde w_2'\rangle_{\cC_z}|_{\bp=\Ktwo} 
& = \int_{\cC_z} -\im \overline{w_1'(\bx)} \A(\bx) \nabla w_2'(\bx) + \im  w_2'(\bx)\A(\bx)  \overline{ \nabla w_1'(\bx)}\, d\bx\\
& = \int_{\cC_z} -\im w_2(\bx) \A(\bx) \overline{\nabla w_1(\bx)} + \im  \overline{w_1(\bx)}\A(\bx) \nabla w_2(\bx)\, d\bx \\
& = - \langle \tilde w_1,  \nabla_{\bp}\cL (0,0,\bp) \tilde w_2\rangle_{\cC_z}|_{\bp=\Kone}. \\
\langle \tilde w_1',  \partial_{\delta}\cL (0,\delta,\Ktwo) \tilde w_1'\rangle_{\cC_z}|_{\delta=0} 
& = \int_{\cC_z} (\overline{\nabla w_1'(\bx)})^t \B(\bx) \nabla w_1'(\bx) \, d\bx 
= \int_{\cC_z} (\nabla w_2(\bx))^t \B(\bx) \overline{\nabla w_2(\bx)} \, d\bx \\
& = - \int_{\cC_z}  \overline{ ( \overline{\nabla w_2(\bx)})^t \B(\bx)\nabla w_2(\bx)} \, d\bx \,  = -  \overline{\langle \tilde w_2,  \partial_{\delta}\cL (0,\delta,\Kone) \tilde w_2\rangle_{\cC_z}|_{\delta=0}}. \\
\langle \tilde w_1',  \partial_{\eps}\cL (\eps,0,\Ktwo) \tilde w_1'\rangle_{\cC_z}|_{\eps=0} 
& = \overline{\langle \tilde w_2,  \partial_{\eps}\cL (\eps,0,\Kone) \tilde w_2\rangle_{\cC_z}|_{\eps=0}}.
\end{align*}

\end{proof}

\begin{assumption}
Here and henceforth, we assume $t_1\neq0$ and $t_2\neq0$, which ensures spectral band gap opening at the Dirac points when $\eps$ or $\delta$ is nonzero as discussed below.
\end{assumption}


We now investigate the spectral bands for the operator $\cL(\eps,\delta)$ when $\delta\neq0$ or $\eps\neq0$.
Parameterize the quasi-momentum near $\Kone$ and $\Ktwo$ by
\begin{equation*}
  \bp(\Kone; \ell,\mu):=\Kone + \ell \bbeta_1 + \mu\bbeta_2 \quad\mbox{and}\quad
  \bp(\Ktwo; \ell,\mu):=\Ktwo + \ell \bbeta_1 + \mu\bbeta_2.
\end{equation*}
Introduce two functions
\begin{equation}\label{eq:Lmu}
L_1(\eps,\ell,\mu):= \frac{\theta_*(\ell+\mu\overline{\tau})}{|\eps t_1|+ \sqrt{\eps^2 t_1^2 + |\theta_*|^2  | \ell+\mu\bar\tau |^2}}, 
\end{equation}
and
\begin{equation}\label{eq:L2mu}
L_2(\delta,\ell,\mu):= \frac{\theta_*(\ell+\mu\overline{\tau})}{|\delta t_2|+ \sqrt{\delta^2 t_2^2 + |\theta_*|^2  | \ell+\mu\bar\tau |^2}}. 
\end{equation}

The dispersion relations and the corresponding eigenfunctions close to $\Kone$ and $\Ktwo$ are described in the following propositions when perturbations are introduced. More specifically, Propositions~\ref{lem:Koneeps} and~\ref{lem:Ktwoeps} describe the perturbed spectral bands with respect to $\eps$ when $\delta=0$.  Propositions~\ref{lem:Konedelta} and~\ref{lem:Ktwodelta} describe the perturbation of the spectral bands with respect to $\delta$ when $\eps=0$. The proofs of Propositions~\ref{lem:Koneeps}-\ref{lem:Ktwodelta} are presented in Appendix~\ref{sec:pertDirac}.

Recall that that the two branches of dispersion relations perturbed from the Dirac points are denotd by $\lambda_{1,\eps,\delta}(\bp)$ and $ \lambda_{2,\eps,\delta}(\bp)$. The corresponding normalized eigenfunctions are denoted by $u_{n,\eps,\delta}(\bp)$ for $n=1,2$.
The big-$O$ notations in the propositions have the following meaning. The first $O(\eps,\ell,\mu)$ term in the eigenfunction expansion \eqref{eq:uepsp} and \eqref{eq:uepsm} (or the first $O(\delta,\ell,\mu)$ term in the eigenfunction expansion \eqref{eq:udeltap} and \eqref{eq:udeltam} respectively) attains an order of $\max(|\eps
|,|\ell|,|\mu|)$ (or $\max(|\delta
|,|\ell|,|\mu|)$ respectively) in $H^1(\cC_z)$ norms. The other $O(\eps,\ell,\mu)$ term (or the $O(\delta,\ell,\mu)$ term respectively) in the expansions are complex numbers with an order of $\max(|\eps
|,|\ell|,|\mu|)$ (or $\max(|\delta
|,|\ell|,|\mu|)$ respectively). 

\begin{prop}\label{lem:Koneeps} 
If $\delta=0$ and $|\eps|\ll 1$, the following holds for the spectral problem \eqref{eq:disp_zig} near $\Kone$:
\begin{itemize}
    \item [(i)] The dispersion relations attain the expansions
\begin{equation}\label{eq:lameps}
\begin{aligned}
\lambda_{1,\eps,0}(\bp(\Kone;\ell,\mu)) &= \lambda_*  - \sqrt{\eps^2 t_1^2 + |\theta_*|^2 | \ell+\mu\bar\tau |^2}(1+O(\eps,\ell,\mu)), \\
\lambda_{2,\eps,0}(\bp(\Kone;\ell,\mu)) &= \lambda_*  + \sqrt{\eps^2 t_1^2 + |\theta_*|^2 | \ell+\mu\bar\tau |^2}(1+O(\eps,\ell,\mu)). 
\end{aligned}
\end{equation}

\item[(ii)] 
When $t_1\eps>0$, the corresponding Bloch modes take the form
\begin{equation}\label{eq:uepsp}
\begin{aligned}
u_{1,\eps,0}(\bx;\bp(\Kone;\ell,\mu))&=\left( - \overline{L_1(\eps,\ell,\mu)}w_1 + w_2  +O(\eps,\ell,\mu) \right)\frac{1}{\sqrt{1+|L_1(\eps,\ell,\mu)|^2+O(\eps,\ell,\mu)}},\\
u_{2,\eps,0}(\bx;\bp(\Kone;\ell,\mu))&=\left(w_1 + L_1(\eps,\ell,\mu)w_2  +O(\eps,\ell,\mu) \right)\frac{1}{\sqrt{1+|L_1(\eps,\ell,\mu)|^2+O(\eps,\ell,\mu)}}.
\end{aligned}
\end{equation}
When $t_1\eps<0$, the corresponding Bloch modes take the form
\begin{equation}\label{eq:uepsm}
\begin{aligned}
u_{1,\eps,0}(\bx;\bp(\Kone;\ell,\mu))&=\left(w_1 - L_1(\eps,\ell,\mu)w_2  +O(\eps,\ell,\mu) \right)\frac{1}{\sqrt{1+|L_1(\eps,\ell,\mu)|^2+O(\eps,\ell,\mu)}},\\
u_{2,\eps,0}(\bx;\bp(\Kone;\ell,\mu))&=\left( \overline{L_1(\eps,\ell,\mu)}w_1 + w_2  +O(\eps,\ell,\mu)\right)\frac{1}{\sqrt{1+|L_1(\eps,\ell,\mu)|^2+O(\eps,\ell,\mu)}}.
\end{aligned}
\end{equation}

\end{itemize}
\end{prop}

\begin{prop}\label{lem:Ktwoeps}
If $\delta=0$ and $|\eps|\ll 1$, the following holds for the spectral problem \eqref{eq:disp_zig} near $\Ktwo$:

\begin{itemize}
    \item [(i)] The dispersion relations attain the  expansions
\begin{equation}\label{eq:lamepstwo}
\begin{aligned}
\lambda_{1,\eps,0}(\bp(\Ktwo;\ell,\mu)) &= \lambda_*  - \sqrt{\eps^2 t_1^2 + |\theta_*|^2 | \ell+\mu\bar\tau |^2}(1+O(\eps,\ell,\mu)), \\
\lambda_{2,\eps,0}(\bp(\Ktwo;\ell,\mu)) &= \lambda_*  +\sqrt{\eps^2 t_1^2 + |\theta_*|^2 | \ell+\mu\bar\tau |^2}(1+O(\eps,\ell,\mu)). 
\end{aligned}
\end{equation}

\item[(ii)] 
When $t_1\eps>0$, the corresponding Bloch modes take the form
\begin{equation}\label{eq:uepsptwo}
\begin{aligned}
u_{1,\eps,0}(\bx;\bp(\Ktwo;\ell,\mu))&=\left(w_1' + L_1(\eps,\ell,\mu)w_2'  +O(\eps,\ell,\mu) \right)\frac{1}{\sqrt{1+|L_1(\eps,\ell,\mu)|^2+O(\eps,\ell,\mu)}},\\
u_{2,\eps,0}(\bx;\bp(\Ktwo;\ell,\mu))&=\left( - \overline{L_1(\eps,\ell,\mu)}w_1' + w_2'  +O(\eps,\ell,\mu) \right)\frac{1}{\sqrt{1+|L_1(\eps,\ell,\mu)|^2+O(\eps,\ell,\mu)}}.
\end{aligned}
\end{equation}
When $t_1\eps<0$, the corresponding Bloch modes take the form
\begin{equation}\label{eq:uepsmtwo}
\begin{aligned}
u_{1,\eps,0}(\bx;\bp(\Ktwo;\ell,\mu))&=\left( \overline{L_1(\eps,\ell,\mu)}w_1' + w_2'  +O(\eps,\ell,\mu)\right)\frac{1}{\sqrt{1+|L_1(\eps,\ell,\mu)|^2+O(\eps,\ell,\mu)}},\\
u_{2,\eps,0}(\bx;\bp(\Ktwo;\ell,\mu))&=\left(w_1' - L_1(\eps,\ell,\mu)w_2'  +O(\eps,\ell,\mu) \right)\frac{1}{\sqrt{1+|L_1(\eps,\ell,\mu)|^2+O(\eps,\ell,\mu)}}.
\end{aligned}
\end{equation}

\end{itemize}
\end{prop}

\begin{prop}\label{lem:Konedelta} 
If $\eps=0$ and $\delta\ll 1$, the following holds for the spectral problem \eqref{eq:disp_zig} near $\Kone$:

\begin{itemize}
    \item [(i)] The dispersion relations attain the  expansions
\begin{equation}\label{eq:lamdelta}
\begin{aligned}
\lambda_{1,0,\delta}(\bp(\Kone;\ell,\mu)) &= \lambda_*  - \sqrt{\delta^2 t_2^2 + |\theta_*|^2 | \ell+\mu\bar\tau |^2}(1+O(\delta,\ell,\mu)), \\
\lambda_{2,0,\delta}(\bp(\Kone;\ell,\mu)) &= \lambda_*  + \sqrt{\delta^2 t_2^2 + |\theta_*|^2 | \ell+\mu\bar\tau |^2}(1+O(\delta,\ell,\mu)). 
\end{aligned}
\end{equation}

\item[(ii)] 
When $t_2\delta>0$, the corresponding Bloch modes take the form
\begin{equation}\label{eq:udeltap}
\begin{aligned}
u_{1,0,\delta}(\bx;\bp(\Kone;\ell,\mu))&=\left( - \overline{L_2(\delta,\ell,\mu)}w_1 + w_2  +O(\delta,\ell,\mu) \right)\frac{1}{\sqrt{1+|L_2(\delta,\ell,\mu)|^2+O(\delta,\ell,\mu)}},\\
u_{2,0,\delta}(\bx;\bp(\Kone;\ell,\mu))&=\left(w_1 + L_2(\delta,\ell,\mu)w_2  +O(\delta,\ell,\mu) \right)\frac{1}{\sqrt{1+|L_2(\delta,\ell,\mu)|^2+O(\delta,\ell,\mu)}}.
\end{aligned}
\end{equation}
When $t_2\delta<0$, the corresponding Bloch modes take the form
\begin{equation}\label{eq:udeltam}
\begin{aligned}
u_{1,0,\delta}(\bx;\bp(\Kone;\ell,\mu))&=\left(w_1 - L_2(\delta,\ell,\mu)w_2  +O(\delta,\ell,\mu) \right)\frac{1}{\sqrt{1+|L_2(\delta,\ell,\mu)|^2+O(\delta,\ell,\mu)}},\\
u_{2,0,\delta}(\bx;\bp(\Kone;\ell,\mu))&=\left( \overline{L_2(\delta,\ell,\mu)}w_1 + w_2  +O(\delta,\ell,\mu)\right)\frac{1}{\sqrt{1+|L_2(\delta,\ell,\mu)|^2+O(\delta,\ell,\mu)}}.
\end{aligned}
\end{equation}

\end{itemize}
\end{prop}

\begin{prop}\label{lem:Ktwodelta} 
If $\eps=0$ and $\delta\ll 1$, the following holds for the spectral problem \eqref{eq:disp_zig} near $\Ktwo$:

\begin{itemize}
    \item [(i)] The dispersion relations for the spectral problem \eqref{eq:disp_zig} attain the following expansions:
\begin{equation}\label{eq:lamdeltatwo}
\begin{aligned}
\lambda_{1,0,\delta}(\bp(\Ktwo;\ell,\mu)) &= \lambda_*  - \sqrt{\delta^2 t_2^2 + |\theta_*|^2 | \ell+\mu\bar\tau |^2}(1+O(\delta,\ell,\mu)), \\
\lambda_{2,0,\delta}(\bp(\Ktwo;\ell,\mu)) &= \lambda_*  + \sqrt{\delta^2 t_2^2 + |\theta_*|^2 | \ell+\mu\bar\tau |^2}(1+O(\delta,\ell,\mu)). 
\end{aligned}
\end{equation}

\item[(ii)] 
When $t_2\delta>0$, the corresponding Bloch modes take the form
\begin{equation}\label{eq:udeltaptwo}
\begin{aligned}
u_{1,0,\delta}(\bx;\bp(\Ktwo;\ell,\mu))&=\left(  \overline{L_2(\delta,\ell,\mu)}w_1' + w_2'  +O(\delta,\ell,\mu) \right)\frac{1}{\sqrt{1+|L_2(\delta,\ell,\mu)|^2+O(\delta,\ell,\mu)}},\\
u_{2,0,\delta}(\bx;\bp(\Ktwo;\ell,\mu))&=\left(w_1' - L_2(\delta,\ell,\mu)w_2'  +O(\delta,\ell,\mu) \right)\frac{1}{\sqrt{1+|L_2(\delta,\ell,\mu)|^2+O(\delta,\ell,\mu)}}.
\end{aligned}
\end{equation}
When $t_2\delta<0$, the corresponding Bloch modes take the form
\begin{equation}\label{eq:udeltamtwo}
\begin{aligned}
u_{1,0,\delta}(\bx;\bp(\Ktwo;\ell,\mu))&=\left(w_1' + L_2(\delta,\ell,\mu)w_2'  +O(\delta,\ell,\mu) \right)\frac{1}{\sqrt{1+|L_2(\delta,\ell,\mu)|^2+O(\delta,\ell,\mu)}},\\
u_{1,0,\delta}(\bx;\bp(\Ktwo;\ell,\mu))&=\left( -\overline{L_2(\delta,\ell,\mu)}w_1' + w_2'  +O(\delta,\ell,\mu)\right)\frac{1}{\sqrt{1+|L_2(\delta,\ell,\mu)|^2+O(\delta,\ell,\mu)}}.
\end{aligned}
\end{equation}

\end{itemize}
\end{prop}

\begin{remark}\label{lem:gap}
In view of Propositions~\ref{lem:Koneeps}-\ref{lem:Ktwodelta} and Assumption~\ref{lem:assNoFold}, we observe that for an arbitrary $\mathfrak d\in(0,1)$, along the $\bbeta_1$ direction in Brillouin zone with $\bp(\ell) = \Kone + \ell \bbeta_1$, a band gap of $(\lambda_* - \mathfrak d|t_1\eps|, \lambda_* + \mathfrak d|t_1\eps|)$
is opened for the spectrum of $\cL(\eps,\delta)$
when $\eps\neq0$ and $\delta=0$, and a band gap of $ (\lambda_* - \mathfrak d|t_2\delta|, \lambda_* + \mathfrak d|t_2\delta|) $
is opened for the spectrum of $\cL(\eps,\delta)$ when $\eps=0$ and $\delta\neq0$.
The same holds along $\bp(\ell) = \Ktwo + \ell \bbeta_1$.
\end{remark}

\begin{remark}\label{lem:periodic_Bloch}
We denote the periodic part the Bloch modes by $\tilde u_{n,\eps,\delta}:=e^{-\im\bp\cdot\bx}u_{n,\eps,\delta}\in H^1_{\mathbf 0}$. As shown in Appendix~\ref{sec:pertDirac}, their asymptotic expansions take the same from as the expansions of $u_{n,\eps,\delta}$ in the above propositions, with $w_1$ and $w_2$ replaced by their periodic parts $\tilde w_1$ and $\tilde w_2$, respectively.
For instance, when $t_1\eps>0$ and $\delta=0$,
\begin{equation}\label{eq:uepsptilde}
\tilde u_{1,\eps,0}(\bx;\bp(\Kone;\ell,\mu))=\left( - \overline{L_1(\eps,\ell,\mu)}\tilde w_1 + \tilde w_2  +O(\eps,\ell,\mu) \right)\frac{1}{\sqrt{1+|L_1(\eps,\ell,\mu)|^2+O(\eps,\ell,\mu)}}.
\end{equation}

\end{remark}

\subsubsection{Berry curvature at $\Kone$ and $\Ktwo$}\label{sec:Berry}
Using the asymptotic expansion of Bloch modes in Propositions ~\ref{lem:Koneeps} - \ref{lem:Ktwodelta}, we can compute
the Berry curvature at the high symmetry points $\Kone$ and $\Ktwo$. In particular, the sign of the Berry curvature can be determined from the perturbation of the medium as described in what follows.

Recall that the periodic part of the eigenfunctions $\tilde u(\cdot; \bp)\in H^1_{\mathbf 0}$ is parameterized by the quasi-momentum $\bp=(p_1, p_2)$. The Berry connection
is defined by $\bA(\bp) = (A_1(\bp), A_2(\bp))$, wherein $A_j(\bp) = \im \langle \tilde u(\bp), \partial_{p_j}\tilde u(\bp) \rangle_{\cC_z}$ for $j=1, 2$ \cite{Vanderbilt}. The Berry phase along a closed loop $\ell$ in the momentum space is  the line integral 
\begin{equation*}
    \phi :=  \oint_\ell \bA \cdot d\bp.
\end{equation*}
The Berry curvature, which is the Berry phase per unit area, is given by (cf. \cite{Vanderbilt})
\begin{equation}\label{eq:Berry_curvature}
    \Theta(\bp) = \partial_{p_1} A_2(\bp) - \partial_{p_2} A_1(\bp) = - 2 \; \text{Im}\langle \partial_{p_1} \tilde u(\bp), \partial_{p_2} \tilde u(\bp) \rangle_{\cC_z}.
\end{equation}


\begin{prop}\label{lem:Berry}
Consider the periodic operators $\cL(\eps,0)$ and $\cL(0,\delta)$, which attain spectral band gap at $\lambda_*$ when $\eps$ and $\delta$ is nonzero. The signs of the Berry curvatures $\Theta(\Kone)$ and $\Theta(\Ktwo)$ associated with the eigenfunctions for the spectral band immediately below $\lambda_*$ is summarized in the following table:
\begin{center}
    \begin{tabular}{c|c|c|c}
Operator & &$\Theta(\Kone)$ &$\Theta(\Ktwo)$ \\
\hline
$\cL(\eps,0)$&$t_1\eps>0$ & $+$ & $-$\\ 
\cline{2-4}
&$t_1\eps<0$ & $-$ & $+$\\ 
\hline
$\cL(0,\delta)$&$t_2\delta>0$ & $+$ & $+$\\ 
\cline{2-4}
&$t_2\delta<0$ & $-$ & $-$
\end{tabular}
\end{center}

\end{prop}
\begin{proof}
Let us denote the periodic part of the Bloch mode associated with the spectral band immediately below $(\Kone,\lambda_*)$ by
$\tilde u(\bx;\bp(\Kone, \ell,\mu))$, where $ \bp(\Kone; \ell,\mu):=\Kone + \ell \bbeta_1 + \mu\bbeta_2$.
From \eqref{eq:Berry_curvature} and a change of variable,
the Berry curvature reads
\begin{equation}
\Theta(\bp)=\frac{-2}{\det[\bbeta_1, \bbeta_2]} \; \text{Im}\left\langle \partial_{\ell} \tilde u(\cdot,\bp), \partial_{\mu} \tilde u(\cdot, \bp) \right\rangle_{\cC_z} 
\end{equation}
when $\ell, \mu \ll 1$.

We consider the operator $\cL(\eps,0)$ and the proof for the operator $\cL(0,\delta)$ is similar.
In view of Proposition~\ref{lem:Koneeps} and Remark~\ref{lem:periodic_Bloch}, when $t_1 \eps>0$, the Bloch modes attains the expansion
\begin{equation*}
\tilde u(\cdot, \bp(\Kone, \ell,\mu) ) =\left( - \overline{L_1(\eps,\ell,\mu)} \tilde w_1 + \tilde w_2  +O(\eps,\ell,\mu) \right)\frac{1}{\sqrt{1+|L_1(\eps,\ell,\mu)|^2 +O(\eps,\ell,\mu)}}. 
\end{equation*}
Elementary calculation yields 
\begin{equation}\label{eq:BcepsKonep}
\begin{aligned}
\langle \partial_{\ell} \tilde u,&\partial_{\mu}\tilde u\rangle_{\cC_z} = (\partial_{\mu}\frac{1}{\sqrt{1+|L_1|^2}})(\partial_{\ell}\frac{1}{\sqrt{1+|L_1|^2}})\\
&+\left(\overline{L_1}\partial_{\mu}\frac{1}{\sqrt{1+|L_1|^2}} + \frac{1}{\sqrt{1+|L_1|^2}}\partial_{\mu}\overline{L_1}\right)
\left(L_1\partial_{\ell}\frac{1}{\sqrt{1+|L_1|^2}} + \frac{1}{\sqrt{1+|L_1|^2}}\partial_{\ell}L_1\right)  +O(\eps,\ell,\mu). 
\end{aligned}
\end{equation}
In contrast, when $t_1\eps<0$, the Bloch modes attains the expansion
\begin{equation*}
\tilde u=\left(\tilde w_1 - L_1(\eps,\ell,\mu) \tilde w_2  +O(\eps,\ell,\mu) \right)\frac{1}{\sqrt{1+|L_1(\eps,\ell,\mu)|^2 +O(\eps,\ell,\mu)}}. 
\end{equation*}
This leads to 
\begin{equation}\label{eq:BcepsKonem}
\begin{aligned}
\langle\partial_{\ell} \tilde u,& \partial_{\mu}\tilde u\rangle_{\cC_z} = (\partial_{\mu}\frac{1}{\sqrt{1+|L_1|^2}})(\partial_{\ell}\frac{1}{\sqrt{1+|L_1|^2}})\\
&+\left(L_1\partial_{\mu}\frac{1}{\sqrt{1+|L_1|^2}} + \frac{1}{\sqrt{1+|L_1|^2}}\partial_{\mu}L_1\right)
\left(\overline{L_1}\partial_{\ell}\frac{1}{\sqrt{1+|L_1|^2}} + \frac{1}{\sqrt{1+|L_1|^2}}\partial_{\ell}\overline{L_1}\right) +O(\eps,\ell,\mu). 
\end{aligned}
\end{equation}
Since the first terms in \eqref{eq:BcepsKonep} and \eqref{eq:BcepsKonem} are real, and the second terms are complex conjugates, the Berry curvatures $\Theta(K)$ take opposite signs when $t_1\eps>0$ and $t_1\eps<0$.  A parallel calculation show that the signs of Berry curvatures $\Theta(K')$ attain the opposite sign when $t_1\eps>0$ and $t_1\eps<0$.

We now compute the sign of $\Theta(\Kone)$. In fact, a straightforward calculation for the second term in \eqref{eq:BcepsKonep} shows that
\begin{equation*}
\begin{aligned}
&|L_1|^2\partial_{\mu}\frac{1}{\sqrt{1+|L_1|^2}}\partial_{\ell}\frac{1}{\sqrt{1+|L_1|^2}}
+\frac{1}{1+|L_1|^2}\partial_{\mu}\overline{L_1}\partial_{\ell}L_1\\
&+\frac{\overline{L_1}}{\sqrt{1+|L_1|^2}}\partial_{\mu}\frac{1}{\sqrt{1+|L_1|^2}}\partial_{\ell}L_1
+ \frac{L_1}{\sqrt{1+|L_1|^2}}\partial_{\ell}\frac{1}{\sqrt{1+|L_1|^2}}\partial_{\mu}L_1\\
=&\frac{1}{(1+|L_1|^2)^2}\partial_{\mu}\overline{L_1}\partial_{\ell}L_1 + \text{real terms}\\
=&\frac{\overline{\tau}|\theta_*|^2|\eps t_1|}{\sqrt{\eps^2 t_1^2 + |\theta_*|^2  | \ell+\mu\bar\tau |^2}}\left(|\eps t_1|+\sqrt{\eps^2 t_1^2 + |\theta_*|^2  | \ell+\mu\bar\tau |^2}\right)^2+ \text{real terms}.
\end{aligned}
\end{equation*}
Hence $\Theta(\Kone)>0$ when $t_1\eps>0$ and $\delta=0$.
\end{proof}

\subsubsection{Smooth parametrization of the dispersion relations}\label{sec:smooth}
For the quasi-momentum $\kp^*=\Kone\cdot\be_2$ along the zigzag interface, we parameterize 
the set of Bloch wave vectors $\bp\in\mathbb R^2$ satisfying $\bp\cdot\be_2=\kp^*$ by 
$\bp(\ell):=\Kone + \ell \bbeta_1$. 
Recall that the spectral bands of the operator $\cL(\eps,\delta)$ along the direction $\bp(\ell)$  are denoted by
\begin{equation}\label{eq:pwanalyticpair}
\lambda_{n,\eps,\delta}(\bp(\ell)),u_{n,\eps,\delta}(\bx;\bp(\ell)),
\end{equation}
where $\lambda_{n,\eps,\delta}(\bp(\ell))$ are piecewise analytic for $\ell\in \mathbb R$  and $u_{n,\eps,\delta}(\bp(\ell))$ are piecewise analytic in $H^1(\cC_z)$. Next we introduce another parameterization of dispersion surfaces,
\begin{equation}\label{eq:analyticpair}
\mu_{n,\eps,\delta}(\bp(\ell)),v_{n,\eps,\delta}(\bx;\bp(\ell)),
\end{equation}
such that $\mu_{n,\eps,\delta}(\bp(\ell))$ are analytic and $v_{n,\eps,\delta}(\bp)$ are analytic in $H^1(\cC_z)$ for $\ell\in(-\pi,\pi)$.

When $\eps\neq0$ or $\delta\neq0$, since $\mu_{1,\eps,\delta}(\bp(\ell)) < \mu_{2,\eps,\delta}(\bp(\ell))$ for $\ell\in(-\pi,\pi)$,
the two labelings coincide: $\mu_{n,\eps,\delta}(\bp(\ell))=\lambda_{n,\eps,\delta}(\bp(\ell))$, $v_{n,\eps,\delta}(\bx;\bp(\ell))=u_{n,\eps,\delta}(\bx;\bp(\ell))$ for $n=1,2$. 

 When $\eps=\delta=0$, 
$\mu_{1,\eps,\delta}(\bp(\ell))$ and $\mu_{2,\eps,\delta}(\bp(\ell))$ intersect at the Dirac point $\bp(0)=\Kone$, thus $\lambda_{1,0,0}(\bp(\ell,0))$ and $\lambda_{2,0,0}(\bp(\ell,0))$ are not analytic. In the following lemma, we display the spectral bands $\mu_{1,0,0}(\bp(\ell,0))$ and $\mu_{2,0,0}(\bp(\ell,0))$ and the corresponding Bloch modes  $v_{n,\eps,\delta}(\bx;\bp(\ell))$ such that they are analytic with respect to $\ell$. For simplicity of notation, we drop the subscripts $\eps,\delta$ when they are both zero. 

\begin{lemma}\label{lem:0strip}
Define $\bp(\ell):=\Kone + \ell \bbeta_1 $. For sufficiently small $\ell\in\mathbb R$,
\begin{equation}\label{eq:mu0}
\begin{aligned}
\mu_1(\bp(\ell)) &= \lambda_*  + |\theta_*|\ell(1+O(\ell)) \quad\text{(increasing in $\ell$)}, \\
\mu_2(\bp(\ell)) &= \lambda_*  - | \theta_*|\ell(1+O(\ell)) \quad\text{(decreasing in $\ell$)}. 
\end{aligned}
\end{equation}
The corresponding Bloch modes are chosen as
\begin{equation}
\begin{aligned}
v_1(\bx;\bp(\ell))&=\left(  \frac{\overline{\theta_*}}{|\theta_*|}w_1 +  w_2  +O(\ell) \right)\frac{1}{\sqrt{2+O(\ell)}},\\
v_2(\bx;\bp(\ell))&=\left( \frac{\overline{\theta_*}}{|\theta_*|}w_1 - w_2  +O(\ell) \right)\frac{1}{\sqrt{2+O(\ell)}}.
\end{aligned}
\end{equation}
\end{lemma}

\begin{figure}
\begin{center}
\includegraphics[scale=0.8]{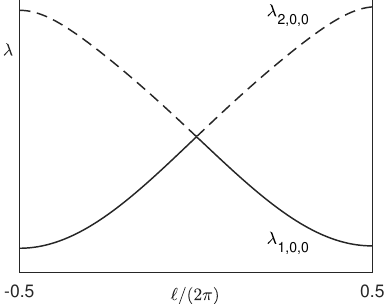}
\includegraphics[scale=0.8]{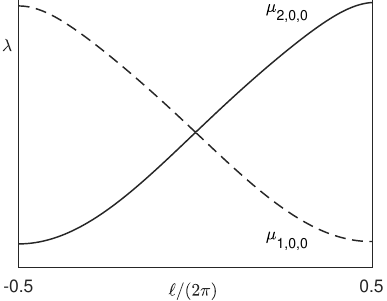}
\includegraphics[scale=0.8]{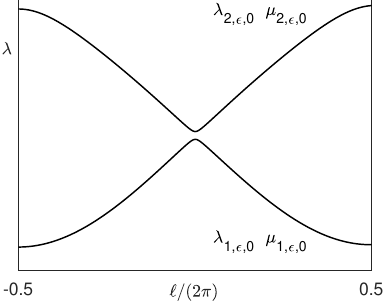}
\end{center}
\caption{The labelings of the band eigenvalues when $\delta=0$: $\mu_{n,\eps,0}$ are smooth branches, while $\lambda_{n,\eps,0}$ are piecewise smooth. When $\eps\neq0$, $\mu_{n,\eps,0}=\lambda_{n,\eps,0}$ for $n=1,2$.}
\label{fig:bandsplit}
\end{figure}

In the above lemma, the $O(\ell)$ terms following $w_2$ and $w_2'$ are functions whose $H^1(\cC_z)$ norms are of order $|\ell|$. The $O(\ell)$ other terms are complex numbers whose $\mathbb C$ norms are of order $|\ell|$. Let $v_i := v_i(\bx;\Kone)$. It follows that
\begin{equation}\label{eq:winv}
\begin{cases}
v_1=\frac{1}{\sqrt2}\left(  \frac{\overline{\theta_*}}{|\theta_*|}w_1 +  w_2 \right)\\
v_2=\frac{1}{\sqrt2}\left( \frac{\overline{\theta_*}}{|\theta_*|} w_1 - w_2\right)
\end{cases}
,\quad
\begin{cases}
w_1 = \frac{1}{\sqrt2}\frac{\theta_*}{|\theta_*|} (v_1  +  v_2)\\
w_2 = \frac{1}{\sqrt2} (v_1 - v_2)
\end{cases}.
\end{equation}

\begin{remark}\label{lem:phase}
When exactly one of $\eps$ and $\delta$ is zero, for $n=1,2$, there exist $\ell$-dependent phase factors $\alpha_n$ such that $\|u_n(\cdot,\bp(\ell,0))- \alpha_n u_{n,\eps,\delta}(\cdot,\bp(\ell,0))\|_{H^1(\chomo)} = O(\eps,\delta)$ uniformly for $\ell$ that is sufficiently small.
\end{remark}

Similar relations hold near the high-symmetry point $\Ktwo$, as described in the following lemma. 

\begin{lemma}\label{lem:0stripKtwo}
Define $\bp(\ell):=\Ktwo + \ell \bbeta_1 $. For sufficiently small $\ell\in\mathbb R$,
\begin{equation}\label{eq:mu0two}
\begin{aligned}
\mu_1(\bp(\ell)) &= \lambda_*  + |\theta_*|\ell(1+O(\ell)) \quad\text{(increasing in $\ell$)}, \\
\mu_2(\bp(\ell)) &= \lambda_*  - |\theta_*|\ell(1+O(\ell)) \quad\text{(decreasing in $\ell$)}. 
\end{aligned}
\end{equation}
The corresponding Bloch modes are chosen as
\begin{equation}
\begin{aligned}
v_1(\bx;\bp(\ell))&=\left(  \frac{\overline{\theta_*}}{|\theta_*|}w_1' - w_2'  +O(\ell) \right)\frac{1}{\sqrt{2+O(\ell)}},\\
v_2(\bx;\bp(\ell))&=\left( \frac{\overline{\theta_*}}{|\theta_*|} w_1' + w_2'  +O(\ell) \right)\frac{1}{\sqrt{2+O(\ell)}}.
\end{aligned}
\end{equation}
\end{lemma}

Let $v_i' := v_i(\bx;\Ktwo)$. It follows that
\begin{equation}\label{eq:winvtwo}
\begin{cases}
v_1'=\frac{1}{\sqrt2}\left(  \frac{\overline{\theta_*}}{|\theta_*|} w_1' -  w_2' \right)\\
v_2'=\frac{1}{\sqrt2}\left( \frac{\overline{\theta_*}}{|\theta_*|}  w_1' + w_2'\right)
\end{cases}
,\quad
\begin{cases}
w_1' = \frac{1}{\sqrt2}\frac{\theta_*}{|\theta_*|} (v_1'  +  v_2')\\
w_2' = \frac{1}{\sqrt2} (-v_1' + v_2')
\end{cases}.
\end{equation}

\section{The Green functions in the infinite strip}\label{sec:Green}
\subsection{Representation of the Green's function}
Consider the infinite strip region $\Omega:= \cup_{m\in\mathbb Z}(\cC_z + m\be_1)$ along the $\be_1$ direction for the joint medium with a zigzag interface. 
To incorporate the quasi-periodic boundary condition along the interface direction $\be_2$ in \eqref{eq:edgedef}, we introduce the function space
\begin{equation}\label{eq:QPedge}
\cH^1_{\kp}(\Omega):=\{u|_{\Omega}, u\in H^1_{\text{loc}}(\mathbb R^2), u(\bx+\be_2)=e^{\im \kp}u(\bx)\}
\quad \mbox{for} \; \kp\in [0, 2\pi].
\end{equation}
Let $G^{\eps,\delta}(\bx,\by; \lambda)$ be
the Green function in $\Omega$ with the quasi-periodic conditions along the $\be_2$ direction. We represent the Green functions in terms of the Bloch modes shown in Sections~\ref{sec:gap} and~\ref{sec:smooth} using the limiting absorption principle, following the derivation in~\cite{Fliss-16}.

Consider the following problem in $\Omega$ with absorption:
\begin{equation}
\langle v, \cL(\eps,\delta)  u\rangle_{\Omega} - (\lambda +\im \sigma )\langle v,   u\rangle_{\Omega}= \langle f,   u\rangle_{\Omega} \quad\text{for all }v\in \cH^1_{\kp}(\Omega).
\end{equation}
Here $\langle \cdot, \cdot \rangle_{\Omega}$ represents the $H^{1}(\Omega)$-$H^{-1}(\Omega)$ pairing, 
$f\in L^2(\Omega)$, and $\sigma$ is a positive constant. As $\sigma \to 0^+$, the kernel of the resolvent  $(\cL(\eps,\delta)-(\lambda+\im\sigma))^{-1}$ defines  
the Green function $G^{\eps,\delta}(\bx,\by; \lambda)$, which satisfies
\begin{equation}\label{eq:phyGeps}
\begin{aligned}
\begin{cases}
(-\nabla\cdot \A(\bx)\nabla  - \lambda) G^{\eps,\delta}(\bx,\by; \lambda)  = \delta(\bx-\by) \quad &\bx\in \Omega,\\
G^{\eps,\delta}(\bx+ \be_2,\by; \lambda)= e^{\im \kp}G^{\eps,\delta}(\bx,\by; \lambda) \quad &\text{for }\bx\in\Gamma_-, \\
\partial_{\nuGb}  G^{\eps,\delta}(\bx+ \be_2,\by; \lambda) =e^{\im \kp}\partial_{\nuGb}  G^{\eps,\delta}(\bx,\by; \lambda)  \quad &\text{ for }\bx\in\Gamma_-.
\end{cases}
\end{aligned}
\end{equation}
Here $\Gamma_-:=\{-\frac{1}{2}\be_2+\ell\be_1, \ell\in\mathbb R\}$ and $\nu_2=(\frac{1}{2},\frac{2}{\sqrt3})$.

The Green functions that are $\kp=\kp^*:=\Kone\cdot\be_2$ quasiperiodic take the following forms.
Let $\bp(\ell):=\Kone+\ell\beta_1$, which  parametrizes all $\bp$ values such that $\bp\cdot\be_2=\kp^*$. 
The Green function when $\eps=\delta=0$ takes the form
\begin{equation}
\begin{aligned}\label{eq:G0}
G^{0,0}(\bx,\by; \lambda_*) = &
\sum_{n\geq3}\frac{1}{2\pi}\int_{[-\pi,\pi ]} 
\frac{\overline{v_{n}(\by;\bp(\ell))} v_{n}(\bx;\bp(\ell))}{ \mu_{n}(\bp(\ell)) - \lambda_* } \, \dpt +\sum_{n=1,2}\frac{1}{2\pi}\text{p.v.}\int_{[-\pi,\pi ]} 
\frac{\overline{v_{n}(\by;\bp(\ell))} v_{n}(\bx;\bp(\ell))}{ \mu_{n}(\bp(\ell)) - \lambda_* } \, \dpt \\
&+ \frac{\im}{2|\theta_*|}\overline{v_{1}(\by;\Kone)} v_{1}(\bx;\Kone) + \frac{\im}{2|\theta_*|}\overline{v_{2}(\by;\Kone)} v_{2}(\bx;\Kone),\quad \bx,\by\in\Omega.
\end{aligned}
\end{equation}
In the above, $\mu_n$ and $v_n$ are the eigenvalues and eigenfunctions that are analytic in $\ell$ as introduced in \eqref{eq:analyticpair} and Lemma \ref{lem:0strip}. We denote the integrals in this Green function by
\begin{equation}\label{eq:tildeG0}
\tilde{G}^{0,0}(\bx,\by; \lambda_*) := 
\sum_{n\geq3}\frac{1}{2\pi}\int_{[-\pi,\pi ]} 
\frac{\overline{v_{n}(\by;\bp(\ell))} v_{n}(\bx;\bp(\ell))}{ \mu_{n}(\bp(\ell)) - \lambda_* } \, \dpt  +\sum_{n=1,2}\frac{1}{2\pi}\text{p.v.}\int_{[-\pi,\pi ]} 
\frac{\overline{v_{n}(\by;\bp(\ell))} v_{n}(\bx;\bp(\ell))}{ \mu_{n}(\bp(\ell)) - \lambda_* } \, \dpt.
\end{equation}
When exactly one of $\eps$ and $\delta$ is $0$, and when $\lambda$ is in the bandgap  of $\cL(\eps,\delta)$, $(\lambda_* - |t_1\eps|, \lambda_* +  |t_1\eps|)$ or $(\lambda_* - |t_2\delta|, \lambda_* +  |t_2\delta|)$, 
the Green function takes the form
\begin{equation}\label{eq:Green}
G^{\eps,\delta}(\bx,\by; \lambda) = 
 \sum_{n\geq1}\frac{1}{2\pi}\int_{[-\pi,\pi ]}
\frac{\overline{v_{n,\eps,\delta}(\by;\bp(\ell))} v_{n,\eps,\delta}(\bx;\bp(\ell))}{ \mu_{n,\eps,\delta}(\bp(\ell)) - \lambda } \, \dpt ,\quad \bx,\by\in\Omega.
\end{equation}
Again, $\mu_n$ and $v_n$ are the eigenvalues and eigenfunctions that are analytic in $\ell$. 
It is known that $G^{\eps,\delta}(\bx,\by; \lambda)$ decays exponentially as $|\bx\cdot\be_1|\to\infty$~\cite{Fliss-16}.

\begin{remark}\label{lem:Greenp}
The Green functions that are $\kp^{*,\prime}=\Ktwo\cdot\be_2$ quasiperiodic along $\be_2$ are similarly defined by doing the two replacements: first, use $\bp(\ell):=\Ktwo+\ell\beta_1$, and second use $v_1'(\cdot,\Ktwo)$ and $v_2'(\cdot,\Ktwo)$ in place of $v_1(\cdot,\Kone)$ and $v_2(\cdot,\Kone)$. The resulting Green functions are denoted by $G^{0,0,\prime}(\bx,\by; \lambda_*)$, $\tilde G^{0,0,\prime}(\bx,\by; \lambda_*)$ and $G^{\eps,\delta,\prime}(\bx,\by; \lambda)$ respectively. 
\end{remark}

\subsection{Asymptotic behavior of  the layer potentials}\label{sec:asympGreen}
We derive the asymptotics of the single layer potential using the expansions of Green functions as shown in Section~\ref{sec:Green}. More precisely, let $\lambda = \lambda_* + |\eps|\, h$, with $|\eps|$ small and $h\in\mathbb{C}$ satisfying $|h|\ll \mathfrak d |t_1|$, where $\mathfrak d\in(0,1)$ is a fixed constant. In this regime, the singularity of the Green's function appears when 
\[
\mu_{n,\eps,\delta}(\bp(\ell)) \approx \lambda \approx \lambda_*, \,\,\, \mbox{for} \; n=1, 2
\]
and $\bp(\ell)$ is near $\Kone$ or $\Ktwo$. Around these singularities, the change of the sign in $t_1\eps$ or $t_2\delta$ leads to significant changes in the eigenmodes $v_{n,\eps,\delta}(\bp(\ell))$ as shown in Propositions~\ref{lem:Koneeps} through \ref{lem:Ktwodelta}. 

To investigate the behavior of the layer potentials when $\eps$ or $\delta$ is small, we study the boundary integral operators $\mathcal{S}^{\eps, 0}(\lambda)$ and $\mathcal{S}^{0,\delta}$ 
with the kernel $G^{\eps,0}(\mathbf{x},\mathbf{y}; \lambda)$ and $G^{0,\delta}(\mathbf{x},\mathbf{y};\lambda)$ respectively. By \eqref{eq:Green}, the single layer potential $\mathcal{S}^{\eps, 0}(\lambda)$ attains the following spectral representation:
\begin{equation}
\mathcal{S}^{\eps, 0}(\lambda)\phi:= \int_{\Gamma} G^{\eps,\delta}(\bx,\by; \lambda)\phi(\by)\,d s_{\by} = \frac{1}{2\pi}\sum_{n\geq1} \int_{-\pi}^{\pi}\!
\frac{\overline{\langle \phi,v_{n,\eps, 0}(\cdot;\mathbf{p}(\ell))\rangle}\; v_{n,\eps, 0}(\mathbf{x};\mathbf{p}(\ell))}
{\mu_{n,\eps, 0}(\mathbf{p}(\ell))-\lambda}\,d\ell.
\end{equation}
Here the index \(n=1,2\) corresponds to the two bands that meet at \(\lambda_*\) while \(n\ge3\) denote the other bands, as described in Section~\ref{sec:gap}. The pairing $\langle \phi,v_{n,\eps, 0}(\cdot;\mathbf{p}(\ell))\rangle$ represents the $\cH^{-1/2}(\Gamma)$-$\cH^{1/2}(\Gamma)$ pairing, where the spaces are defined in~\eqref{eq:HsG2}. We can prove that as $t_1\eps\to0$ with a fixed sign, the following operator limit holds:
\begin{equation}\label{eq:PTeps}
\begin{aligned}
\mathcal{S}^{\eps, 0}(\lambda_*+|\eps| h)\phi &= \tilde{\mathcal{S}}^{0,0}(\lambda_*)\phi \\
&\quad + \beta_1(h)\Bigl[\overline{\langle \phi,v_1\rangle}\, v_1 + \overline{\langle \phi,v_2\rangle}\, v_2 \Bigr] \\
&\quad +\sgn(t_1\eps) \xi_1(h) \Bigl[\overline{\langle \phi,v_1\rangle}\, v_2 + \overline{\langle \phi,v_2\rangle}\, v_1\Bigr] + o(1),
\end{aligned}
\end{equation}
where $\tilde{\mathcal{S}}^{0,0}$ is the boundary integral operator with kernel $\tilde{G}^{0,0}$ and 
\begin{equation}
\begin{aligned}
\beta_1(h) &=  \frac{1}{2|\theta_*|}\frac{h}{\sqrt{t_1^2-h^2}}, \qquad
\xi_1(h) = \left|\frac{t_1}{2\theta_*}\right|\frac{1}{\sqrt{t_1^2-h^2}}. 
\end{aligned}
\end{equation}

The derivation of \eqref{eq:PTeps} is based on the decomposition of the \(\ell\)-integral in \eqref{eq:Green} into three parts:

\begin{enumerate}
  \item \textbf{First two bands near the Dirac point \(\boldsymbol{(|\ell|\le|\eps|^{1/3})}\):} \\
    In the regime where
    \[
    |\ell|\le|\eps|^{1/3},
    \]
    one exploits the local expansion
    \[
    \mu_{n,\eps, 0}(\mathbf{p}(\ell)) = \lambda_* + (-1)^{n-1} \sqrt{t_1^2+|\theta_*|^2\ell^2}\,\Bigl(1+O(\eps,\ell)\Bigr).
    \]
    Therefore, after subtracting the spectral parameter, we have
    \[
    \mu_{n,\eps, 0}(\mathbf{p}(\ell)) - (\lambda_*+\eps h)= (-1)^{n-1}\sqrt{t_1^2+|\theta_*|^2\ell^2} - \eps h + O(\eps,\ell).
    \]
    By the change of variables and the symmetry arguments (in particular, noticing that the terms involving \(L_1(1,\ell)\) are odd in \(\ell\) and cancel accordingly), there holds
    \[
    \begin{aligned}
    \frac{1}{2\pi}\sum_{n=1}^2 &\int_{-|\eps|^{1/3}}^{|\eps|^{1/3}} 
    \frac{\overline{\langle \phi, v_{n,\eps, 0}(\cdot;\mathbf{p}(\ell))\rangle}\; 
    v_{n,\eps, 0}(\mathbf{x};\mathbf{p}(\ell))}
    {\mu_{n,\eps, 0}(\mathbf{p}(\ell))-\lambda_*-\eps h} \,d\ell\\[1mm]
    &\longrightarrow 
     \beta_1(h)\Bigl[\overline{\langle \phi,v_1\rangle}\, v_1 + \overline{\langle \phi,v_2\rangle}\, v_2 \Bigr]
    +\sgn(t_1\eps) \xi_1(h) \Bigl[\overline{\langle \phi,v_1\rangle}\,v_2 + \overline{\langle \phi,v_2\rangle}\,v_1\Bigr].
    \end{aligned}
    \]
    
  \item \textbf{First two bands away from the Dirac point \(\boldsymbol{(|\ell|>|\eps|^{1/3})}\):} \\
    For values of \(\ell\) away from the degeneracy the denominators remain uniformly bounded away from \(0\). A standard application of the dominated convergence theorem gives
    \[
    \frac{1}{2\pi}\sum_{n=1}^2\int_{|\ell|>|\eps|^{1/3}} 
    \frac{\overline{\langle \phi, v_{n,\eps, 0}(\cdot;\mathbf{p}(\ell))\rangle}\; 
    v_{n,\eps, 0}(\mathbf{x};\mathbf{p}(\ell))}
    {\mu_{n,\eps, 0}(\mathbf{p}(\ell))-\lambda_*-\eps h}\,d\ell
    \longrightarrow \sum_{n=1,2}\frac{1}{2\pi}\text{p.v.}\int_{[-\pi,\pi ]} 
\frac{\overline{\langle \phi,v_{n}(\by;\bp(\ell))\rangle} v_{n}(\bx;\bp(\ell))}{ \mu_{n}(\bp(\ell)) - \lambda_* } \, \dpt .
    \]
    
  \item \textbf{Higher bands (\(\boldsymbol{n\ge3}\)):} \\
    Since the spectral parameter \(\lambda_*+\eps h\) remains uniformly separated from these bands, a direct perturbative argument shows that their contributions converge to the corresponding terms of the unperturbed operator:
    \[
    \frac{1}{2\pi}\sum_{n\ge3}\int_{-\pi}^{\pi}
    \frac{\overline{\langle \phi, v_{n,\eps, 0}(\cdot;\mathbf{p}(\ell))\rangle}\; 
    v_{n,\eps, 0}(\mathbf{x};\mathbf{p}(\ell))}
    {\mu_{n,\eps, 0}(\mathbf{p}(\ell))-\lambda_*-\eps h}\,d\ell
    \longrightarrow \sum_{n\geq3}\frac{1}{2\pi}\int_{[-\pi,\pi ]} 
\frac{\overline{\langle \phi,v_{n}(\by;\bp(\ell))\rangle} v_{n}(\bx;\bp(\ell))}{ \mu_{n}(\bp(\ell)) - \lambda_* } \, \dpt
    \]
\end{enumerate}

Collecting the limit operator for each one above, we deduce the asymptotic expansion \eqref{eq:PTeps}. Note that the uniform convergence is achieved in the appropriate Sobolev spaces (from \(\widetilde{H}^{-1/2}(\Gamma)\) to \(H^{1/2}(\Gamma)\)) by standard analytic and perturbative techniques.

A similar analysis applies to the boundary integral operator $\mathcal{S}^{0,\delta}(\lambda)$. Its spectral representation is analogous to \eqref{eq:Green}, and one obtains an asymptotic expansion of the form:
\begin{equation}\label{eq:PTdelta}
\begin{aligned}
\mathcal{S}^{0,\delta}(\lambda_*+|\delta| h)\phi &= \tilde{\mathcal{S}}^{0,0}(\lambda_*)\phi \\
&\quad + \beta_2(h)\Bigl[\overline{\langle \phi,v_1\rangle}\, v_1 + \overline{\langle \phi,v_2\rangle}\, v_2 \Bigr] \\
&\quad +\sgn(t_2\delta) \xi_2(h) \Bigl[\overline{\langle \phi,v_1\rangle}\, v_2 + \overline{\langle \phi,v_2\rangle}\, v_1\Bigr] + o(1),
\end{aligned}
\end{equation}
where 
\begin{equation}
\begin{aligned}
\beta_2(h) &=  \frac{1}{2|\theta_*|}\frac{h}{\sqrt{t_2^2-h^2}}, \qquad
\xi_2(h) = \left|\frac{t_2}{2\theta_*}\right|\frac{1}{\sqrt{t_2^2-h^2}}. 
\end{aligned}
\end{equation}

\begin{remark}
The above analysis shows that the Green function $G^{\eps, 0}$ attains a phase transition around $\eps=0$, which is characterized by the coefficient $\xi_1(h)$ in~\eqref{eq:PTeps}; while the Green function $G^{0, \delta}$ attains a phase transition around $\delta=0$, characterized by the coefficient $\xi_2(h)$ in~\eqref{eq:PTdelta}.
\end{remark}

\section{The interface mode for the joint photonic crystal}\label{sec:interface}
In this section, we consider the spectral problem \eqref{eq:edgedef} and prove Theorem~\ref{thm:edge}.
The proof for each case in Theorem~\ref{thm:edge} is similar.
For clarity we only give the proof for Case 3 and discuss the proof of the other two cases briefly.

\subsection{The spectral problem}
\label{sec:zig}

As $\Lint(\eps_L,\delta_L,\eps_R,\delta_R)$ is periodic along the direction $\be_2$, the Floquet-Bloch theory implies that it suffices to consider the quasi-periodic problem along $\be_2$. Introduce the operator $\Lint_{\kp}(\eps_L,\delta_L,\eps_R,\delta_R)$, which is the restriction of $\Lint(\eps_L,\delta_L,\eps_R,\delta_R)$ on $\cH^1_{\kp}(\Omega)$.

We consider the interface modes bifurcating from the Dirac points, which will attain quasi-momenta near $\kp^*:= \Kone\cdot \be_2=\frac{4\pi}{3}$ or $\kp^{*,\prime} = \Ktwo\cdot \be_2=- \frac{4\pi}{3}$.
An interface mode bifurcating from $\Kone$ is an eigenfunction $u\in \cH^1_{\kp^*}(\Omega)$ such that for some $\lambda\in\mathbb R$,
\begin{equation}\label{eq:interfaceeig}
\langle v, \Lint(\eps_L,\delta_L,\eps_R,\delta_R)  u\rangle_{\Omega} = \lambda \langle v,  u\rangle_{\Omega}, \quad \text{for all }  v(\bx)\in \cH^1_{\kp^*}(\Omega).
\end{equation}
and an interface mode bifurcating from $\Ktwo$ is an eigenfunction $u\in \cH^1_{\kp^{*,\prime}}(\Omega)$ such that for some $\lambda\in\mathbb R$,
\begin{equation}\label{eq:interfaceeigtwo}
\langle v, \Lint(\eps_L,\delta_L,\eps_R,\delta_R)  u\rangle_{\Omega} = \lambda \langle v,  u\rangle_{\Omega}, \quad \text{for all }  v(\bx) \in \cH^1_{\kp^{*,\prime}}(\Omega).
\end{equation}
We prove case 3 of Theorem~\ref{thm:edge}, wherein $(\eps_L,\delta_L)=(0, \delta)$ and $(\eps_R,\delta_R)=(\eps, 0)$.

\subsection{Integral-equation formulation for the spectral problem}\label{sec:characterization}
Let $\Gamma:=\{t\be_2, -\frac{1}{2}\leq t < \frac{1}{2}\}$ be the interface of the two photonic crystals in the strip region $\Omega$. For $s\in\mathbb R$, define the quasi-periodic Sobolev space on $\Gamma$, $\cH^s(\Gamma)$,  by 
\begin{equation}\label{eq:HsG2}
\cH^s(\Gamma):= \left\{ u(\bx_0+ t\be_2) = \sum_{n\in\mathbb Z} a_n e^{\im \kp^* t} e^{\im 2\pi n t}: \|u\|_{\cH^s(\Gamma)}^2:=\sum_{n\in\mathbb Z} |a_n|^2 (1+|n|^2)^s\right\},
\end{equation}
Here $\bx_0 = -\frac{1}{2}\be_1 -\frac{1}{2}\be_2$.

Let $\lambda$ be located in the common gap of $\cL(\eps,0)$ and $\cL(0,\delta)$ along $\bp(\ell)=\Kone + \ell \bbeta_1 $ when $\eps$ or $\delta$ is nonzero (cf.  Remark~\ref{lem:gap}). That is, $\lambda \in (\lambda_* - \mathfrak d|t_1\eps|, \lambda_* + \mathfrak d|t_1\eps|)\cap (\lambda_* - \mathfrak d|t_2\delta|, \lambda_* + \mathfrak d|t_2\delta|) $, where $\mathfrak d\in(0,1)$. Let $\bn=(\frac{1}{2},-\frac{\sqrt3}{2})$ be the unit normal vector of $\Gamma$ pointing to the right. For $(\psi,\phi)\in \HhalfG\times\HmhalfG$, we define the single and double layer potentials: 
\begin{equation}\label{eq:lpotentialeps}
\begin{aligned}
\cS^{\eps,\delta}(\lambda)\phi(\bx)&:=\int_\Gamma  G^{\eps,\delta}(\bx,\by;\lambda) \phi(\by)\, ds_{\by}, \quad \bx\notin \Gamma, \\
\cD^{\eps,\delta}(\lambda)\psi(\bx)&:=\int_\Gamma  \partial_{n_\by} G^{\eps,\delta}(\bx,\by;\lambda) \psi(\by)\, ds_{\by} \quad \bx\notin \Gamma,  
\end{aligned}
\end{equation}
where $G^{\eps,\delta}(\bx,\by;\lambda)$ are the  Green functions that are $\kp^*$ quasiperiodic along $\be_2$ as defined in \eqref{eq:Green}. 
The single layer potential $\cS^{\eps,\delta}(\lambda)\phi(\bx)$ can be continuously extended to $\Gamma$ and it defines an bounded integer operator from $\HmhalfG$ to $\HhalfG$, which we still denote by $\cS^{\eps,\delta}$. Given $(\psi,\phi)\in \HhalfG\times\HmhalfG$, we also define the integral operators
\begin{equation}\label{eq:int_opt_K}
\begin{aligned}
\cK^{\eps,\delta}(\lambda)\psi(\bx)&:=\int_\Gamma  \partial_{n_\by} G^{\eps,\delta}(\bx,\by;\lambda) \psi(\by)\, ds_{\by} \quad \bx\in \Gamma,\\
\cK^{*,\eps,\delta}(\lambda)\phi(\bx)&:=\int_\Gamma  \partial_{n_\bx} G^{\eps,\delta}(\bx,\by;\lambda) \phi(\by)\, ds_{\by} \quad \bx\in \Gamma.\\
\end{aligned}
\end{equation}
It can be shown that $\cK^{\eps,\delta}:\HhalfG\to\HhalfG$ and $\cK^{*,\eps,\delta}:\HmhalfG\to\HmhalfG$ are bounded.

By taking the limit of the layer potentials as $\bx\to\Gamma$,
the following jump relationship holds~\cite{colton2013integral}: 
\begin{equation}\label{eq:jumpeps}
\begin{aligned}
[\cS^{\eps,\delta}\psi (\lambda) ]_{\pm} &= \cS^{\eps,\delta} (\lambda) \psi, \\
[\partial_n \cS^{\eps,\delta}(\lambda)  \psi ]_{\pm} &= \mp\frac{1}{2}\psi+\cK^{*,\eps,\delta}(\lambda) \psi, \\
[ \cD^{\eps,\delta}\phi (\lambda) ]_{\pm} &= \pm\frac{1}{2}\phi + \cK^{\eps,\delta} (\lambda) \phi, \\
[\partial_n \cD^{\eps,\delta} (\lambda) \phi ]_{\pm}& =: \cN^{\eps,\delta}\phi. 
\end{aligned}
\end{equation}
In the above, the subscripts $-$ and $+$ represent the limit of the layer potentials as $\bx\to\Gamma$ from the left and right side respectively, the symbol $\partial_n$ represents the normal derivative, and $\cN^{\eps,\delta}: H^{1/2}(\Gamma)\to H^{-1/2}(\Gamma)$ is a well-defined bounded operator.

Assume that $u(\bx)$ is an interface mode of $\Lint_{\kp^*}(0,\delta,\eps,0)$ with the eigenvalue $\lambda\in (\lambda_* - \mathfrak d|t_1\eps|, \lambda_* + \mathfrak d|t_1\eps|)\cap (\lambda_* - \mathfrak d|t_2\delta|, \lambda_* + \mathfrak d|t_2\delta|)$. 
Let  $u|_\Gamma \in \HhalfG$ and $\partial_n u|_\Gamma \in \HmhalfG$ be the traces of $u$ and the normal derivatives of $u$ on $\Gamma$. Then it follows from the Green's formula that $u$ attains the following representation in the infinite strip $\Omega$:
\begin{equation}\label{eq:tracetoedge}
u(\bx) = 
\begin{cases}
\big[\cD^{\eps,0} (\lambda)  u|_\Gamma \big](\bx) - \big[\cS^{\eps,0}(\lambda) \partial_n u|_\Gamma  \big] (\bx)  \quad \text{for}\; \bx \; \text{on the right of }\Gamma,\\
-\big[ \cD^{0,\delta} (\lambda)  u|_\Gamma \big] (\bx) +  \big[\cS^{0,\delta} \partial_n u|_\Gamma (\lambda)  \big](\bx)  \quad \text{for}\; \bx \; \text{on the left of }\Gamma.
\end{cases}
\end{equation}
Here we use the fact that $u\in\cH^1_{\kp^*}(\Omega)$, especially the decay of $u$ when $|\bx\cdot\be_1| \to\infty$ when applying the Green's formula.
 Taking the limit from either side of $\Gamma$, we obtain the following two systems of integral equations:

 \begin{equation}\label{eq:densuright}
 \left(\begin{matrix}
 u|_\Gamma\\
 \partial_n u|_\Gamma
 \end{matrix}\right)
  =
 \left(\begin{matrix}
 \cK^{\eps,0}(\lambda)  + \frac{1}{2}\cI  & -\cS^{\eps,0}(\lambda)   \\
 \cN^{\eps,0}(\lambda)  & -\cK^{*,\eps,0}(\lambda)  + \frac{1}{2}\cI
 \end{matrix}\right) 
 \left(\begin{matrix}
 u|_\Gamma\\
 \partial_n u|_\Gamma
 \end{matrix}\right),
 \end{equation}
 and
 \begin{equation}\label{eq:densuleft}
 \left(\begin{matrix}
 u|_\Gamma\\
 \partial_n u|_\Gamma
 \end{matrix}\right)
 =
 \left(\begin{matrix}
 - \cK^{0,\delta}(\lambda)  + \frac{1}{2}\cI & \cS^{0,\delta}(\lambda)    \\
  - \cN^{0,\delta}(\lambda)  & \cK^{*,0,\delta}(\lambda)  + \frac{1}{2}\cI
 \end{matrix}\right) 
 \left(\begin{matrix}
 u|_\Gamma\\
 \partial_n u|_\Gamma
 \end{matrix}\right).
 \end{equation}


Conversely, assume $\lambda \in (\lambda_* - \mathfrak d|t_1\eps|, \lambda_* + \mathfrak d|t_1\eps|)\cap (\lambda_* - \mathfrak d|t_2\delta|, \lambda_* + \mathfrak d|t_2\delta|) $. Let $(\psi,\phi)\in \HhalfG\times\HmhalfG$, which is not necessarily the Cauchy data of an interface mode on the interface $\Gamma$. We define $u(\bx)$ in the infinite strip $\Omega$ as a combination of single and double layer potentials:
\begin{equation}\label{eq:denstoedge}
u(\bx)= 
\begin{cases}
  [\cD^{\eps,0}(\lambda)  \psi](\bx)  - [\cS^{\eps,0}(\lambda)  \phi](\bx)  \quad \text{on the right of }\Gamma, \\
-[\cD^{0,\delta}(\lambda)  \psi](\bx)  + [\cS^{0,\delta}(\lambda)  \phi](\bx)  \quad \text{on the left of }\Gamma.
\end{cases}
\end{equation}
For $u$ defined above to be an interface mode, we only need the continuity of the value and the normal derivative across $\Gamma$, expressed as:
 \begin{equation}\label{eq:dataright}
 \left(\begin{matrix}
  \cK^{\,0}(\lambda) + \frac{1}{2}\cI & -\cS^{\eps,0}(\lambda) \\
  \cN^{\eps,0}(\lambda) & -\cK^{*,\eps,0}(\lambda) + \frac{1}{2}\cI 
 \end{matrix}\right)  
 \left(\begin{matrix}
 \psi\\
 \phi
 \end{matrix}\right)
 =
 \left(\begin{matrix}
 - \cK^{0,\delta}(\lambda) + \frac{1}{2}\cI & \cS^{0,\delta}(\lambda)   \\
 - \cN^{0,\delta}(\lambda) & \cK^{*,0,\delta}(\lambda) + \frac{1}{2}\cI
 \end{matrix}\right) 
 \left(\begin{matrix}
 \psi\\
 \phi
 \end{matrix}\right)\neq0.
 \end{equation}

%

Define the integral operators on $\HhalfG\times\HmhalfG$ 
\begin{equation}\label{eq:bTeps}
\mathbb T^{\eps,\delta}
(\lambda):=
\left(\begin{matrix}
-\cK^{\eps,\delta}(\lambda) & \cS^{\eps,\delta}(\lambda) \\
 -\cN^{\eps,\delta}(\lambda) & \cK^{*,\eps,\delta}(\lambda)
\end{matrix}\right),
\end{equation}
and
\begin{equation}\label{eq:bTs}
\mathbb T_s^{\eps,\delta}(\lambda):=\mathbb T^{\eps,0}(\lambda) + \mathbb T^{0,\delta}(\lambda), \quad
\mathbb T_t^{\eps,\delta}(\lambda):=- \mathbb T^{\eps,0}(\lambda) + \mathbb T^{0,\delta}(\lambda) + \mathbb I, \quad
\mathbb T_n^{\eps,\delta}(\lambda):=\mathbb T^{\eps,0}(\lambda) - \mathbb T^{0,\delta}(\lambda) + \mathbb I,
\end{equation}
where $\mathbb I$ is the identity operator.
By virtue of \eqref{eq:densuright}, \eqref{eq:densuleft} and \eqref{eq:dataright}, we obtain the following lemma for the interface modes.

\begin{lemma}\label{lem:edgestate} Assume that $(\eps_L,\delta_L)=(0, \delta)$ and $(\eps_R,\delta_R)=(\eps, 0)$.
Let $\mathfrak d\in(0,1)$ and $\lambda \in (\lambda_* - \mathfrak d|t_1\eps|, \lambda_* + \mathfrak d|t_1\eps|)\cap (\lambda_* - \mathfrak d|t_2\delta|, \lambda_* + \mathfrak d|t_2\delta|) $.

\begin{itemize}
    \item [(i)] There exists an interface mode $u$ satisfying \eqref{eq:interfaceeig} if and only if there exists $(\psi,\phi)\in \HhalfG\times\HmhalfG$ such that
\begin{equation}\label{eq:sufficient}
\mathbb T_s^{\eps,\delta}
(\lambda)
\left(\begin{matrix}
\psi\\
\phi
\end{matrix}\right)
=0,\quad  \mathbb T_t^{\eps,\delta}
(\lambda)
\left(\begin{matrix}
\psi\\
\phi
\end{matrix}\right)
\neq0.
\end{equation}
Furthermore, each solution to \eqref{eq:sufficient} yields an interface mode expressed by \eqref{eq:denstoedge}.
\item [(ii)] If $u$ is an interface mode satisfying \eqref{eq:interfaceeig}, then $0\neq( u|_\Gamma, \partial_n u|_\Gamma) \in \HhalfG\times \HmhalfG$ satisfies
\begin{equation}\label{eq:necessary}
\mathbb T_s^{\eps,\delta}
(\lambda)
\left(\begin{matrix}
u|_\Gamma\\
\partial_n u|_\Gamma
\end{matrix}\right)
=0,\quad
\mathbb T_n^{\eps,\delta}
(\lambda)
\left(\begin{matrix}
u|_\Gamma\\
\partial_n u|_\Gamma
\end{matrix}\right)
=0.
\end{equation}
\end{itemize}
\end{lemma}

\begin{remark}\label{lem:T}
At $\Ktwo$, we define $\mathbb T^{\eps,\delta,\prime}
(\lambda)$,  $\mathbb T^{\eps,\delta,\prime}_s$, $\mathbb T^{\eps,\delta,\prime}_t$ and $\mathbb T^{\eps,\delta,\prime}_n$ in parallel to those defined in \eqref{eq:bTeps} and \eqref{eq:bTs}, where the Green functions  $G^{\eps,\delta}$ in the layer potentials are replaced by $G^{\eps,\delta,\prime}$.
We also define $\tilde{\mathbb T}^{0,0,\prime} (\lambda_*)$ and $\tilde{\mathbb T}^{0,0,\prime} (\lambda_*)$ parallel to \eqref{eq:bTeps} when the Green functions in the layer potentials are replaced by 
$\tilde{G}^{0,0,\prime}(\bx,\by; \lambda_*)$ and $ \tilde{G}^{0,0,\prime}(\bx,\by; \lambda_*)$, respectively.
\end{remark}

\begin{remark}
The integral-equation formulation can also be set up when $(\eps_L,\delta_L)=(0, \delta)$ and $(\eps_R,\delta_R)=(0, -\delta)$, or when
$(\eps_L,\delta_L)=(\eps,0)$ and $(\eps_R,\delta_R)=(-\eps, 0)$. For the former, Lemma~\ref{lem:edgestate} holds with the integral operators defined by
\begin{equation*}
\mathbb T_s^{\eps,\delta}(\lambda):=\mathbb T^{-\eps,0}(\lambda) + \mathbb T^{\eps,0}(\lambda), \quad
\mathbb T_t^{\eps,\delta}(\lambda):=- \mathbb T^{-\eps,0}(\lambda) + \mathbb T^{\eps,0}(\lambda) + \mathbb I, \quad
\mathbb T_n^{\eps,\delta}(\lambda):=\mathbb T^{-\eps,0}(\lambda) - \mathbb T^{\eps,0}(\lambda) + \mathbb I.
\end{equation*}
For the latter, Lemma~\ref{lem:edgestate} holds with the integral operators defined by
\begin{equation*}
\mathbb T_s^{\eps,\delta}(\lambda):=\mathbb T^{0,-\delta}(\lambda) + \mathbb T^{0,\delta}(\lambda), \quad
\mathbb T_t^{\eps,\delta}(\lambda):=- \mathbb T^{0,-\delta}(\lambda) + \mathbb T^{0,\delta}(\lambda) + \mathbb I, \quad
\mathbb T_n^{\eps,\delta}(\lambda):=\mathbb T^{0,-\delta}(\lambda) - \mathbb T^{0,\delta}(\lambda) + \mathbb I.
\end{equation*}
\end{remark}

\subsection{The limiting operators}
We derive asymptotic expansions for the integral operators $\mathbb T^{\eps,\delta}$, $\mathbb T_s^{\eps,\delta}$, $\mathbb T_t^{\eps,\delta}$,  $\mathbb T_n^{\eps,\delta}$, $\mathbb T^{\eps,\delta,\prime}$, $\mathbb T_s^{\eps,\delta,\prime}$, $\mathbb T_t^{\eps,\delta,\prime}$, and $\mathbb T_n^{\eps,\delta,\prime}$ in this subsection. To this end, we first introduce several notations. For $\vec\phi = (\psi,\phi)\in\HhalfG\times\HmhalfG$, at $\Kone$, let
\begin{equation}\label{eq:ci}
c_i(\vec\phi):= \overline{\langle \phi,v_i\rangle}_\Gamma - \langle\partial_n v_i, \psi\rangle_\Gamma,
\end{equation}
where $v_i$ are defined in Remark~\ref{lem:phase}.
We also denote
\begin{equation}\label{eq:vecv}
\vec v_i:=
\left(\begin{matrix}v_i|_\Gamma \\\partial_n v_i|_\Gamma\end{matrix}\right),\quad
i=1,2,
\end{equation}
and define the operators
\begin{equation}\label{eq:P}
\mathbb P
\vec\phi : =
c_1(\vec\phi)\vec v_1
+c_2(\vec\phi)\vec v_2,
\end{equation}
\begin{equation}\label{eq:Q}
\mathbb Q
\vec\phi 
:=
c_2(\vec\phi)\vec v_1
+c_1(\vec\phi)\vec v_2.
\end{equation}
%
Similarly at $\Ktwo$, define
\begin{equation}\label{eq:citwo}
c_i'(\vec\phi):= \overline{\langle \phi,v_i'\rangle}_\Gamma - \langle\partial_n v_i', \psi\rangle_\Gamma,
\end{equation}
where $v_i'$ are defined in Remark~\ref{lem:phase}.
We also denote
\begin{equation}\label{eq:vecvtwo}
\vec v_i':=
\left(\begin{matrix}v_i'|_\Gamma \\\partial_n v_i'|_\Gamma\end{matrix}\right),\quad
i=1,2,
\end{equation}
and define the operators
\begin{equation}\label{eq:Ptwo}
\mathbb P'
\vec\phi : =
c_1'(\vec\phi)\vec v_1'
+c_2'(\vec\phi)\vec v_2',
\end{equation}
\begin{equation}\label{eq:Qtwo}
\mathbb Q'
\vec\phi 
:=
c_2'(\vec\phi)\vec v_1'
+c_1'(\vec\phi)\vec v_2'.
\end{equation}

Define the functions
\begin{equation}\label{eq:betaxi}
\begin{aligned}
    & \beta_1(h):=   \frac{1}{2|\theta_*|}\frac{h}{\sqrt{t_1^2-h^2}}, \quad  \beta_2(h):=   \frac{1}{2|\theta_*|}\frac{h}{\sqrt{t_2^2-h^2}}, \\
    & \xi_1(h):=  \left|\frac{t_1}{2\theta_*}\right|\frac{1}{\sqrt{t_1^2-h^2}},\quad \xi_2(h):=  \left|\frac{t_2}{2\theta_*}\right|\frac{1}{\sqrt{t_2^2-h^2}}. 
\end{aligned}
\end{equation}
%

Following the procedure in Section~\ref{sec:asympGreen}, one can obtain the limiting operators of $\cS^{\eps,\delta}$,  $\cK^{\eps,\delta}$, $\cK^{*,\eps,\delta}$, and $\cN^{\eps,\delta}$ when $\eps\to0$ or $\delta\to0$. In particular, the limit of the integral operator $\mathbb T^{\eps,\delta}$ is summarized in the following propositions. The readers are referred to \cite{li2024} for detailed calculations when $\delta=0$.
\begin{prop} \label{lem:oplim}
Let Assumption~\ref{lem:assNoFold} holds along $\bbeta_1$ and $t_1\neq0$. Let $\mathfrak d\in(0,1)$ be a constant. 
Then the following limit holds uniformly for $h\in\mathbb C$ that satisfy $|h|< \mathfrak d|t_1|$ as $t_1\eps\to 0^\pm$: 
\begin{equation}\label{eq:oplim}
\mathbb T^{\eps,0}(\lambda_*+ |\eps| h) 
\to \tilde{\mathbb T}^{0,0} (\lambda_*) 
+ \beta_1(h)\mathbb P \pm \xi_1(h)\mathbb Q ,
\end{equation}
where
the convergence is understood with the operator norm from $\HhalfG\times\HmhalfG$ to $\HhalfG\times\HmhalfG$. 
\end{prop}

\begin{prop} \label{lem:oplimtwo}
Let Assumption~\ref{lem:assNoFold} holds along $\bbeta_1$ and $t_1\neq0$. Let $\mathfrak d\in(0,1)$ be a constant. 
Then the following limit holds uniformly for $h\in\mathbb C$ that satisfy $|h|< \mathfrak d|t_1|$ as $t_1\eps\to 0^\pm$: 
\begin{equation}\label{eq:oplimtwo}
\mathbb T^{\eps,0,\prime}(\lambda_*+ |\eps| h) 
\to \tilde{\mathbb T}^{0,0,\prime} (\lambda_*) 
+ \beta_1(h)\mathbb P' \mp \xi_1(h)\mathbb Q',
\end{equation}
where
the convergence is understood with the operator norm from $\HhalfG\times\HmhalfG$ to $\HhalfG\times\HmhalfG$. 
\end{prop}

\begin{prop} \label{lem:oplimdone}
Let Assumption~\ref{lem:assNoFold} holds along $\bbeta_1$ and $t_2\neq0$. Let $\mathfrak d\in(0,1)$ be a constant. 
Then the following limit holds uniformly for $h\in\mathbb C$ that satisfy $|h|< \mathfrak d|t_2|$ as $t_2\delta\to 0^\pm$: 
\begin{equation}\label{eq:oplimdone}
\mathbb T^{0,\delta}(\lambda_*+ |\delta| h) 
\to \tilde{\mathbb T}^{0,0} (\lambda_*) 
+ \beta_2(h)\mathbb P \pm \xi_2(h)\mathbb Q,
\end{equation}
where
the convergence is understood with the operator norm from $\HhalfG\times\HmhalfG$ to $\HhalfG\times\HmhalfG$. 
\end{prop}

\begin{prop} \label{lem:oplimdtwo}
Let Assumption~\ref{lem:assNoFold} holds along $\bbeta_1$ and $t_2\neq0$. Let $\mathfrak d\in(0,1)$ be a constant. 
Then the following limit holds uniformly for $h\in\mathbb C$ that satisfy $|h|< \mathfrak d|t_2|$ as $t_2\delta\to 0^\pm$: 
\begin{equation}\label{eq:oplimdtwo}
\mathbb T^{0,\delta,\prime}(\lambda_*+ |\delta| h) 
\to \tilde{\mathbb T}^{0,0,\prime} (\lambda_*) 
+ \beta_2(h)\mathbb P' \pm \xi_2(h)\mathbb Q', 
\end{equation}
where
the convergence is understood with the operator norm from $\HhalfG\times\HmhalfG$ to $\HhalfG\times\HmhalfG$. 
\end{prop}

\begin{remark} \label{lem:formcorres}
It is interesting to note that
the form of the limiting operators in Propositions~\ref{lem:oplim} through \ref{lem:oplimdtwo} are closely related to the $\eps$ and $\delta$ derivatives of the elliptic operator given in \eqref{eq:Tderiv1} and \eqref{eq:Tderiv2}, or the relative signs of the Berry curvatures at $\Kone$ and $\Ktwo$ by Proposition \ref{lem:Berry}.
More precisely, the sign for the multiplication operators $\xi_1(h)$ and $\xi_2(h)$, which sit in front of the operators $\mathbb Q$ and $\mathbb Q'$ in \eqref{eq:oplimtwo} - \eqref{eq:oplimdtwo}, depends on the sign of the $(1,1)$ entry of the matrices in \eqref{eq:Tderiv1} and \eqref{eq:Tderiv2}. 
For example, at $\Kone$ and $\Ktwo$, the $(1,1)$ entry of the $\eps$ derivatives are $t_1$ and $-t_1$ respectively, and that for the $\delta$ derivatives are $t_2$ and $t_2$ respectively. 
\end{remark}

\subsection{Properties of the limiting operators}
To facilitate the proof of Theorem~\ref{thm:edge}, we calculate the limits of the operators defined in \eqref{eq:bTs}.
For clarity, we give the details when $\rho = \frac{t_1\eps}{t_2\delta}>0$, $t_1\eps>0$ and $t_2\delta>0$. The derivations of the limiting operators for the case $\rho>0$, $t_1\eps<0$ and $t_2\delta<0$; and the case that $\rho<0$ are similarly.
\begin{lemma}
Let $(\eps_L,\delta_L)=(0, \delta)$ and $(\eps_R,\delta_R)=(\eps, 0)$, and $\mathfrak d\in(0,1)$ be a constant.
Assume that $\rho=\frac{t_1\eps}{t_2\delta}>0$, $t_1\eps>0$ and $t_2\delta>0$.
If Assumption~\ref{lem:assNoFold} holds along $\bbeta_1$, then at $\kp=\kp^*$, when $\eps$ and $\delta$ approach zero while $\rho$ is fixed, the following convergence holds uniformly for $h\in\mathbb C$ satisfying $|h|<\mathfrak d|t_1|$:
\begin{equation}\label{eq:neceKone}
T_n^{\eps,\delta}
( \lambda_*+ |\eps| h)\to  \left(\beta_1(h)-\beta_2(\left|\frac{\eps}{\delta}\right|h)\right)\mathbb P+\left(\xi_1(h)-\xi_2(\left|\frac{\eps}{\delta}\right|h)\right)\mathbb Q +\mathbb I =:\mathbb U_n(h).
\end{equation}
At $\kp=-\kp^*$, when where $\eps$ and $\delta$ approach zero while $\rho$ is fixed, the following convergence holds uniformly for $h\in\mathbb C$ satisfying $|h|<\mathfrak d|t_1|$:
\begin{equation}\label{eq:stnKtwo}
\begin{aligned}
T_s^{\eps,\delta,\prime}
( \lambda_*+ |\eps| h)&\to 2\tilde{\mathbb T}^{0,0,\prime} (\lambda_*)+ \left(\beta_1(h)+\beta_2(\left|\frac{\eps}{\delta}\right|h)\right)\mathbb P-\left(\xi_1(h)-\xi_2(\left|\frac{\eps}{\delta}\right|h)\right)\mathbb Q +\mathbb I =:\mathbb U_s'(h),\\
T_t^{\eps,\delta,\prime}
( \lambda_*+ |\eps| h)&\to  \left(\beta_1(h)-\beta_2(\left|\frac{\eps}{\delta}\right|h)\right)\mathbb P+\left(\xi_1(h)+\xi_2(\left|\frac{\eps}{\delta}\right|h)\right)\mathbb Q +\mathbb I =:\mathbb U_t'(h),\\
T_n^{\eps,\delta,\prime}
( \lambda_*+ |\eps| h)&\to  -\left(\beta_1(h)-\beta_2(\left|\frac{\eps}{\delta}\right|h)\right)\mathbb P-\left(\xi_1(h)+\xi_2(\left|\frac{\eps}{\delta}\right|h)\right)\mathbb Q +\mathbb I =:\mathbb U_n'(h).
\end{aligned}
\end{equation}
\end{lemma}
\begin{proof}
First, consider $\kp=\kp^*$.
Since $t_1\eps>0$ and $t_2\delta>0$, by Propositions~\ref{lem:oplim}-\ref{lem:oplimdone}, we obtain
\begin{equation}
\mathbb T^{\eps,0}(\lambda_*+ |\eps| h_1) 
\to \tilde{\mathbb T}^{0,0} (\lambda_*) 
+ \beta_1(h_1)\mathbb P + \xi_1(h_1)\mathbb Q ,
\end{equation}
\begin{equation}
\mathbb T^{0,\delta}(\lambda_*+ |\delta| h_2) 
\to \tilde{\mathbb T}^{0,0} (\lambda_*) 
+ \beta_2(h_2)\mathbb P + \xi_2(h_2)\mathbb Q .
\end{equation}
Choose 
\begin{equation}\label{eq:h1h2}
h_1=h,\quad
h_2=\left|\frac{\eps}{\delta}\right|h,
\end{equation}
so that the two energies coincide $\lambda_*+ |\eps| h_1=\lambda_*+ |\delta| h_2 = \lambda_*+ |\eps| h$. Note that since $t_1$ and $t_2$ are constants, $\rho$ being fixed implies that $\frac{\eps}{\delta}$ is also fixed. The relation~\eqref{eq:neceKone} follows from the definition of $\mathbb T_n^{\eps,\delta}(\lambda)$ in \eqref{eq:bTs}.

Next, consider $\kp=-\kp^*$. Since $t_1\eps>0$ and $t_2\delta>0$, by Propositions~\ref{lem:oplimtwo}-\ref{lem:oplimdtwo}, we obtain
\begin{equation}
\mathbb T^{\eps,0,\prime}(\lambda_*+ |\eps| h_1) 
\to \tilde{\mathbb T}^{0,0,\prime} (\lambda_*) 
+ \beta_1(h_1)\mathbb P' - \xi_1(h_1)\mathbb Q' ,
\end{equation}
\begin{equation}
\mathbb T^{0,\delta,\prime}(\lambda_*+ |\delta| h_2) 
\to \tilde{\mathbb T}^{0,0} (\lambda_*) 
+ \beta_2(h_2)\mathbb P' + \xi_2(h_2)\mathbb Q' .
\end{equation}
With the same choice as in~\eqref{eq:h1h2}, the relations in~\eqref{eq:stnKtwo} follow from the definitions of $T_s^{\eps,\delta,\prime}$, $T_t^{\eps,\delta,\prime}$ and $T_n^{\eps,\delta,\prime}$  in Remark~\ref{lem:T}.
\end{proof}

The properties of the limiting operators are stated in the following two lemmas.
\begin{lemma}\label{lem:U}
The operator $\mathbb U_n(h)$ defined in \eqref{eq:neceKone} is  a family of Fredholm operators with index zero that is analytic in $h$, and it attains no characteristic values.
\end{lemma}

\begin{proof}
Setting
\begin{equation}\label{eq:ab}
a=\frac{|t_1\eps|}{|t_2\delta|}=|\rho|,\quad\text{and} \quad b=\frac{h}{|t_1|},
\end{equation}
with the choice in \eqref{eq:h1h2}, we obtain
\begin{equation*}
\begin{aligned}
\beta_1(h_1) &= \frac{1}{2|\theta_*|}\frac{h}{\sqrt{t_1^2-h^2}} = \frac{1}{2|\theta_*|}\frac{b}{\sqrt{1-b^2}},\\
\beta_2(h_2) &= \frac{1}{2|\theta_*|}\frac{ah}{\sqrt{t_1^2-(ah)^2}} = \frac{1}{2|\theta_*|}\frac{ab}{\sqrt{1-(ab)^2}}, \\
\xi_1(h_1) &= \frac{|t_1|}{2|\theta_*|}\frac{1}{\sqrt{t_1^2-h^2}}= \frac{1}{2|\theta_*|}\frac{1}{\sqrt{1-b^2}} ,\\
\xi_2(h_2) &= \frac{|t_1|}{2|\theta_*|}\frac{1}{\sqrt{t_1^2-(ah)^2}} = \frac{1}{2|\theta_*|}\frac{1}{\sqrt{1-(ab)^2}}.
\end{aligned}
\end{equation*}
It is known that $c_i(\vec v_j)$ are related to the energy flux of $\vec v_j$ through $\Gamma$~\cite{Fliss-16}
\begin{equation}\label{eq:civj}
c_i(\vec v_j) = (-1)^{n-1}|\theta_*|\delta_{i,j}.
\end{equation}
On the basis of $\vec v_1, \vec v_2$, 
\begin{equation*}
\mathbb P (\vec v_1,\vec v_2) =  \im |\theta_*| (\vec v_1,\vec v_2) \left(\begin{matrix}1 & 0 \\ 0 & -1 \\ \end{matrix}\right), \quad \mathbb Q (\vec v_1,\vec v_2) = \im |\theta_*| (\vec v_1,\vec v_2) \left(\begin{matrix}0 & -1 \\ 1 & 0 \\ \end{matrix}\right).
\end{equation*}
As the ranges of $\mathbb P$ and $\mathbb Q$ are both $\text{span}\{\vec v_1,\vec v_2\}$, a necessary condition for $h$ to be a characteristic value of $\mathbb U_n(h)$ is $b=h/|t_1|$ is a characteristic value of $M_n$ defined in the following:
\begin{equation}\label{eq:Mn}
M_n(b)=\im|\theta_*|\left(\begin{matrix}\beta_1(h_1)-\beta_2(h_2) & -(\xi_1(h_1)-\xi_2(h_2)) \\ \xi_1(h_1)-\xi_2(h_2) & -(\beta_1(h_1)-\beta_2(h_2)) \\ \end{matrix}\right)
+ \left(\begin{matrix}1 & 0 \\ 0 & 1 \\ \end{matrix}\right)=\quad \frac{\im }{2}\left(\begin{matrix}x_- - 2\im & -y_- \\ y_- & -x_- -2\im \\ \end{matrix}\right), 
\end{equation}
where 
\begin{equation}\label{eq:xy}
x_{\pm}:=\frac{b}{\sqrt{1-b^2}} \pm \frac{ab}{\sqrt{1-(ab)^2}}, \quad
y_{\pm}:= \frac{1}{\sqrt{1-b^2}} \pm \frac{1}{\sqrt{1-(ab)^2}}.
\end{equation}
To find characteristic values of $M_n(h)$, we compute
\begin{equation}
\det(M_n(b))=0 \iff y_-^2=x_-^2+4.
\end{equation}
A necessary condition for the equation on the right is 
\begin{equation}
b^2(1-a)^2=0,
\end{equation}
which implies $b=0$ or $a=1$. However, in both cases, $x_-=y_-=0$ and $\det(M_n(b))\neq0$. Thus $M_n(b)$ has no characteristic values.
\end{proof}

\begin{lemma}\label{lem:Up}
The operators $\mathbb U_s'(h)$, $\mathbb U_n'(h)$ and $\mathbb U_t'(h)$ defined in \eqref{eq:stnKtwo} are families of Fredholm operators with index zero that are analytic in $h$. For each of these three operators, the only characteristic value is $h=0$, and its multiplicity is $2$.
\end{lemma}

\begin{proof}
Similar to \eqref{eq:civj}, we have the relations
\begin{equation}
c_i'(\vec v_j') = (-1)^{n-1}|\theta_*|\delta_{i,j}.
\end{equation}
On the basis of $\vec v_1', \vec v_2'$, 
\begin{equation}
\mathbb P' (\vec v_1',\vec v_2') =  \im |\theta_*| (\vec v_1',\vec v_2') \left(\begin{matrix}1 & 0 \\ 0 & -1 \\ \end{matrix}\right), \quad \mathbb Q' (\vec v_1',\vec v_2') = \im |\theta_*| (\vec v_1',\vec v_2') \left(\begin{matrix}0 & -1 \\ 1 & 0 \\ \end{matrix}\right).
\end{equation}
Furthermore, it was shown in~\cite{li2024} that
\begin{equation*}
\ker \; \tilde{\mathbb T}^{0,0,\prime} (\lambda_*) =X',\quad \text{Ran}  \; \tilde{\mathbb T}^{0,0,\prime} (\lambda_*) =Y', 
\end{equation*}
where
\begin{equation}\label{eq:XY}
 X':=\text{span}\{\vec v_1', \vec v_2'\},\quad Y':=\{\vec\phi \in \HhalfG\times\HmhalfG\}, c_i'(\vec\phi')=0, i=1,2\}.
\end{equation}
Thus the kernels of $\mathbb U_s'(h)$, $\mathbb U_n'(h)$ and $\mathbb U_t'(h)$ are subspaces of $X'$, and their characteristic problems are reduced to characteristic value problems of the 3-by-3 matrices:
\begin{equation}\label{eq:Mbp}
\begin{aligned}
M_s'(b)=& \im |\theta_*|\left(\begin{matrix}\beta_1(h_1)+\beta_2(h_2) & \xi_1(h_1)-\xi_2(h_2) \\ -(\xi_1(h_1)-\xi_2(h_2)) & -(\beta_1(h_1)+\beta_2(h_2)) \\ \end{matrix}\right)
, \\
M_t'(b) =& -\im |\theta_*|\left(\begin{matrix}\beta_1(h_1)-\beta_2(h_2) & \xi_1(h_1)+\xi_2(h_2) \\ -(\xi_1(h_1)+\xi_2(h_2)) & -(\beta_1(h_1)-\beta_2(h_2)) \\ \end{matrix}\right)
+\left(\begin{matrix}1 & 0 \\ 0 & 1 \\ \end{matrix}\right), \\
M_n'(b)=& \im|\theta_*|\left(\begin{matrix}\beta_1(h_1)-\beta_2(h_2) & \xi_1(h_1)+\xi_2(h_2) \\ -(\xi_1(h_1)+\xi_2(h_2)) & -(\beta_1(h_1)-\beta_2(h_2)) \\ \end{matrix}\right) + \left(\begin{matrix}1 & 0 \\ 0 & 1 \\ \end{matrix}\right)
\end{aligned}
\end{equation}
These matrices take the following forms in terms of the quantities defined in~\eqref{eq:xy}:
\begin{equation}
\begin{aligned}
M_s'(b)=&\quad  \frac{\im }{2}\left(\begin{matrix}x_+ & y_- \\ -y_- & -x_+ \\ \end{matrix}\right),\\
M_t'(b)=& -\frac{\im }{2}\left(\begin{matrix}x_- +2\im & y_+ \\ -y_+ & -x_- +2\im \\ \end{matrix}\right), \\
M_n'(b)=&\quad \frac{\im }{2}\left(\begin{matrix}x_- - 2\im & y_+ \\ -y_+ & -x_- -2\im \\ \end{matrix}\right) =\overline{M_t'(b)}.
\end{aligned}
\end{equation}
To find the characteristic values of $M_s'$, $M_t'$ and $M_n'$, we compute 
\begin{equation*}
\det(M_s'(b))=0 \iff (x_+)^2=(y_-)^2,
\end{equation*}
\begin{equation*}
\det(M_t'(b))=0 \iff \det(M_n)=0 \iff (x_-)^2-(y_+)^2 + 4=0.
\end{equation*}
Both equalities simplify to
\begin{equation}\label{eq:condmismatch}
b^2(1+a)^2=0,
\end{equation}
which holds and only holds when $b=0$. That is, $h=0$ is the only characteristic value to each of $M_s'(b)$, $M_t'(b)$ and $M_n'(b)$.

Next we find the kernels and multiplicities of the matrices at $h=0$. Compute
\begin{equation}\label{eq:xyderivatives}
\begin{aligned}
&\frac{d}{db}x_{\pm}=\frac{1}{\sqrt{1-b^2}^3} \pm \frac{a}{\sqrt{1-a^2b^2}^3}, \quad  
&&\frac{d^2}{db^2}x_{\pm}=\frac{3b}{\sqrt{1-b^2}^5} \pm \frac{3a^3b}{\sqrt{1-a^2b^2}^5}, \\
&\frac{d}{db}y_{\pm}=\frac{b}{\sqrt{1-b^2}^3} \pm \frac{a^2b}{\sqrt{1-a^2b^2}^3}, \quad  
&&\frac{d^2}{db^2}y_{\pm}=\frac{1}{\sqrt{1-b^2}^5}(1+2b^2) \pm \frac{a^2}{\sqrt{1-a^2b^2}^5}(1+2a^2b^2).
\end{aligned}
\end{equation}
Thus at $h=0$,
\begin{equation}
M_s'(0) =  \frac{\im }{2}\left(\begin{matrix}0 & 0 \\ 0 & 0 \\ \end{matrix}\right), \quad
\frac{d}{db}M_s'(0) =  \frac{\im(1+a) }{2}\left(\begin{matrix}1 & 0 \\ 0 & -1 \\ \end{matrix}\right), 
\end{equation}
\begin{equation*}
M_n'(0)= -\frac{\im }{2}\left(\begin{matrix}2\im & 2 \\ -2 & 2\im \\ \end{matrix}\right),\quad
\frac{d}{db}M_n'(0)= -\frac{\im(1-a) }{2}\left(\begin{matrix}1 & 0 \\ 0 & -1\\ \end{matrix}\right),\quad
\frac{d^2}{db^2}M_n'(0)= -\frac{\im(1+a^2) }{2}\left(\begin{matrix} 0& 1 \\ -1& 0 \\ \end{matrix}\right).
\end{equation*}

For $M_s(b)$, we observe $\ker(M_s(0)=\mathbb R^2$. For each nontrivial vector $(A,B)^t\in\mathbb R^2$, and an arbitrary vector function $(A(b),B(b))^t$ that is analytic in $b$ and satisfy $(A(0),B(0))^t=(A,B)^t$, we have
\begin{equation*}
\frac{d}{db} \left(M_s(h)
\left(\begin{matrix}A(b) \\ B(b) \end{matrix}\right)
\right)|_{b=0}=\frac{\im(1+a) }{2}\left(\begin{matrix} 1& 0 \\ 0 & -1\\ \end{matrix}\right)\left(\begin{matrix}A \\ B \end{matrix}\right) \neq0.
\end{equation*}
Thus each vector in $\ker(M_s'(0))$ is of rank 1, and the multiplicity of $M_s(b)$ is equal to 2.

For $M_n(b)$, we observe $\ker(M_n(0)=\text{span}\{(\im,1)^t\}$. For an arbitrary vector function $(A(b),B(b))^t$ that is analytic in $b$ and satisfy $(A(0),B(0))^t=(\im,1)^t$, we have
\begin{equation*}
\frac{d}{db} \left(M_n(h)
\left(\begin{matrix}A(b) \\ B(b) \end{matrix}\right)
\right)|_{b=0}= 
-\frac{\im }{2}\left(\begin{matrix}2\im & 2 \\ -2 & 2\im \\ \end{matrix}\right)
\left(\begin{matrix}\frac{d}{db}A(0) \\ \frac{d}{db}B(0) \end{matrix}\right) 
-\frac{\im (1-a)}{2}\left(\begin{matrix}1 & 0 \\ 0 & -1 \\ \end{matrix}\right)\left(\begin{matrix}\im \\ 1 \end{matrix}\right).
\end{equation*}
This expression is equal to zero if and only if
\begin{equation*}
\frac{d}{db}B(0)= -\im \frac{d}{db}A(0)-\im\frac{1-a}{2}.
\end{equation*}
Finally, observe that
\begin{equation*}
\begin{aligned}
\frac{d^2}{db^2}& \left(M_n(h)
\left(\begin{matrix}A(b) \\ B(b) \end{matrix}\right)
\right)|_{b=0}= 
-\frac{\im }{2}\left(\begin{matrix}2\im & 2 \\ -2 & 2\im \\ \end{matrix}\right)\left(\begin{matrix}\frac{d^2}{db^2}A(0) \\ \frac{d^2}{db^2}B(0) \end{matrix}\right)\\
&-\im(1-a)\left(\begin{matrix}1 & 0 \\ 0 & -1 \\ \end{matrix}\right)\left(\begin{matrix}\frac{d}{db}A(0) \\ -\im \frac{d}{db}A(0)-\im\frac{1-a}{2} \end{matrix}\right) 
-\frac{\im (1+a^2)}{2}
\left(\begin{matrix} 0& 1 \\ -1 & 0 \\ \end{matrix}\right)\left(\begin{matrix}\im \\ 1 \end{matrix}\right).
\end{aligned}
\end{equation*}
The above quantity is zero only when
\begin{equation*}
\left(\begin{matrix} 1+a^2 \\ \im((1-a)^2-(1+a^2)) \end{matrix}\right)\in\text{span}\{
\left(\begin{matrix} 1 \\ \im\end{matrix}\right)
\},
\end{equation*}
which is never satisfied by $a>0$. Thus the vector $(\im,1)^t$ has rank 2, and the multiplicity of $M_n'(b)$ is equal to 2.
Thus we conclude that the only vector $(\im,1)^t$ in $\ker(M_s'(0))$ is of rank 2, and the multiplicity of $M_n(b)$ is equal to two.

\end{proof}
\begin{remark}
When $a=1$, the matrices in \eqref{eq:Mbp} degenerate to 
\begin{equation*}
M_s'(b)= \im |\theta_*|\beta_1(h_1)\left(\begin{matrix}1& 0 \\ -0 & -1 \\ \end{matrix}\right)
, \quad
M_t'(b) = -\im |\theta_*|2\xi_1(h_1)\left(\begin{matrix}0&0 \\ -0 & 0\\ \end{matrix}\right)
+\left(\begin{matrix}1 & 0 \\ 0 & 1 \\ \end{matrix}\right).
\end{equation*}
This is the case considered in our previous work~\cite{li2024}.
\end{remark}


\subsection{Proof of the Main Result  Theorem \ref{thm:edge}}
\label{sec:proofmain}
For clarity, we first present the proof for Case 3. 
We also further focus on $\rho = \frac{t_1\eps}{t_2\delta}>0$, $t_1\eps>0$ and $t_2\delta>0$. 

First, consider $\kp=\kp^*$ and define $V:= \{h\in\mathbb C, |h|<\mathfrak d |t_1|\}$. 
By Lemma~\ref{lem:U}, $\mathbb U_n(h)$ attains no characteristic value in $V$. Since $T_s^{\eps,\delta} (\lambda_*+\eps h)$ is analytic in $\overline{V}$, and converges uniformly to $\mathbb U_n$ in a neighborhood of $V$,  the generalized Rouche theorem~\cite{Ammari-book} implies that $T_s^{\eps,\delta} (\lambda_*+\eps h)$ has no characteristic values in $V$. Thus there is no edge state at $\kp=\kp^*$.

Next, at $\kp=-\kp^*$, we have the following lemmas for the operators $T_s^{\eps,\delta,\prime} (\lambda_*+\eps h)$, $T_t^{\eps,\delta,\prime} (\lambda_*+\eps h)$, and $T_n^{\eps,\delta,\prime} (\lambda_*+\eps h)$ as defined in Remark~\ref{lem:T}. 
The proof bridges from the limiting operators in Lemma~\ref{lem:Up} following parallel lines as Lemma~7.8, Propositions~7.5, 7.6 and 7.11 in \cite{li2024}. We omit the proofs here for conciseness.

\begin{lemma}\label{lem:rank1}
Consider Case 3 of Theorem~\ref{thm:edge}. Let Assumption~\ref{lem:assNoFold} hold along $\bbeta_1$.
Suppose $\rho=\frac{t_1\eps}{t_2\delta}>0$, $t_1\eps>0$ and $t_2\delta>0$.
Let $\mathfrak d\in(0,1)$ be a constant.
When $\eps$ and $\delta$ are sufficiently small while $\rho$ is fixed, and for $|h_0|<\mathfrak d |t_1|$, every nontrivial $\vec\phi\in \ker \;(\mathbb T_s^{\eps,\delta,\prime} (\lambda_*+\eps h_0) )$ is of rank $1$.
Moreover, $\mathbb T_s^{\eps,\delta,\prime} (\lambda_*+\eps h_0)$ may have a nontrivial kernel only when $h_0=o(1)$ as $\eps\to0$.
\end{lemma}

\begin{lemma}\label{lem:uniquestn}
Consider Case 3 of Theorem~\ref{thm:edge}.  Let Assumption~\ref{lem:assNoFold} hold along $\bbeta_1$.
Suppose $\rho=\frac{t_1\eps}{t_2\delta}>0$, $t_1\eps>0$ and $t_2\delta>0$.
Let $\mathfrak d\in(0,1)$ be a constant.
When $\eps$ and $\delta$ are sufficiently small while $\rho$ is fixed, and for $|h_0|<\mathfrak d |t_1|$, the system
\begin{equation}\label{eq:simulsn}
\mathbb T_s^{\eps,\delta,\prime} (\lambda_*+\eps h)\vec\phi=0\quad\text{and} \quad\mathbb T_n^{\eps,\delta,\prime} (\lambda_*+\eps h)\vec\phi=0
\end{equation}
attains at most one solution $(h,\vec\phi)$.
The same holds for the system
\begin{equation}\label{eq:simulst}
\mathbb T_s^{\eps,\delta,\prime} (\lambda_*+\eps h)\vec\phi=0 \quad\text{and}\quad \mathbb T_t^{\eps,\delta,\prime} (\lambda_*+\eps h)\vec\phi=0.
\end{equation}
\end{lemma}

Now let us consider the region $V:= \{h\in\mathbb C, |h|<\mathfrak d |t_1|\}$. 
In view of Lemma~\ref{lem:Up}, $\mathbb U_s'(h)$ has multiplicity 2 in $V$. Since $\mathbb U_s'$ and $\mathbb T_s^{\eps,\delta,\prime}
( \lambda_*+ |\eps| h)$ are analytic in $\overline{V}$, and the convergence is uniform in a neighborhood of $V$,  the generalized Rouche theorem~\cite{Ammari-book} implies that the multiplicity is two in $V$ when $\eps$ is sufficiently small. By Lemma~\ref{lem:rank1}, $\mathbb T_s^{\eps,\delta,\prime}
( \lambda_*+ |\eps| h)$ has two distinct characteristic values in $V$, each of which has a kernel of dimension one. Denote theses two characteristic values and their kernels as $(h_{0,i},\vec \phi_{0,i})$, $i=1,2$. If none of $\vec \phi_{0,i}$ generates an edge state through \eqref{eq:denstoedge}, it contradicts that \eqref{eq:simulsn} has at most one solution.  If each of $\vec \phi_{0,i}$ generates an edge state through \eqref{eq:denstoedge}, it contradicts that \eqref{eq:simulst} has at most one solution. Thus there is excatly one edge mode with $\kp=-\kp^*$.

For the proof of Cases 1 and 2 in Theorem~\ref{thm:edge}, the interface modes can be analyzed following the same steps above through the asymptotic analysis of the integral operators over the interface. In particular, when the perturbations are opposite on two sides of the interface, either with $(\eps_L,\delta_L)=(\eps, 0)$ and $(\eps_R,\delta_R)=(-\eps, 0)$, or with
$(\eps_L,\delta_L)=(0, \delta)$ and $(\eps_R,\delta_R)=(0, -\delta)$,
using the observations in Remark~\ref{lem:formcorres}, we can arrive at the set of integral operators that attain the same structure as \eqref{eq:stnKtwo} with $|\rho|=1$ at both $\kp=\kp^*$ and $\kp=-\kp^*$. As a result, the same argument as above using the Gohberg-Sigal theory implies the existence of the interface modes and their multiplicity.

\section*{Acknowledgement}
W. Li was partially supported by the NSF grant DMS-2532730. J. Lin was partially supported by the NSF grant DMS-2410645 and the Research Support Program from Auburn University. H. Zhang was partially supported by the Hong Kong RGC grant GRF 16307024 and GRF 16301625.


\begin{appendices}

\section{Proof of Lemma~\ref{lem:realization}}
\label{sec:realization}
\begin{proof} 
We only show 
$$\rflc \C(\rflc\bx) \rflc =-\C(\bx), \quad R\C(R^{-1}\bx) R^{-1} = \C(\bx),$$
as all other relations are obvious.
First,
\begin{equation*}
\begin{aligned}
\C_{k,l}(\rflc\bx)&= \nabla \A_{k,l}|_{\rflc \bx} \cdot \J \rflc \bx = 
\left(\rflc  \nabla \big(\A_{k,l}(\rflc\bx) \big)\right)^t \J \rflc \bx \\
&=\left(  \nabla \big(\A_{k,l}(\rflc\bx) \big)\right)^t \rflc\J \rflc \bx
= -\nabla \big(\A_{k,l}(\rflc\bx) \big) \cdot J\bx.
\end{aligned}
\end{equation*}
Thus
\begin{equation*}
\begin{aligned}
\rflc \C(\rflc\bx) \rflc =  - \nabla \big(\rflc\A(\rflc\bx)\rflc \big) \cdot J\bx = - \nabla\A(\bx) \cdot J\bx=-\C(\bx),
\end{aligned}
\end{equation*}
and the first identity above holds.
For the second identity, a straightforward calculation leads to
\begin{equation*}
\begin{aligned}
\C_{k,l}(R^{-1}\bx)&= \nabla \A_{k,l}|_{R^{-1} \bx} \cdot \J R^{-1} \bx = 
\left(R^{-1}  \nabla \big(\A_{k,l}(R^{-1}\bx) \big)\right)^t \J R^{-1} \bx \\
& =\left(  \nabla \big(\A_{k,l}(R^{-1}\bx) \big)\right)^t R\J R^{-1} \bx 
= \nabla \big(\A_{k,l}(R^{-1}\bx) \big) \cdot J\bx.
\end{aligned}
\end{equation*}
Thus
\begin{equation*}
\begin{aligned}
R \C(R^{-1}\bx) R^{-1} =   \nabla \big(R\A(R^{-1}\bx)R^{-1} \big) \cdot J\bx =  \nabla\A(\bx) \cdot J\bx=\C(\bx)
\end{aligned}
\end{equation*}

\end{proof}
\section{Proof of Propositions~\ref{lem:Koneeps}-\ref{lem:Ktwodelta}}\label{sec:pertDirac}
The eigenvalue problem \eqref{eq:disp_zig} is equivalent to finding $(\lambda,\tilde u)\in \mathbb C\times H^1_{\mathbf{0}}(\cC_z)$, such that
\begin{equation}\label{eq:dispper}
\langle \tilde v, \cL(\eps,\delta,\bp) \tilde u\rangle_{\cC_z} = \lambda \langle \tilde v, \tilde u\rangle_{\cC_z} \quad \text{for all } v\in H^1_{\mathbf{0}}(\cC_z).
\end{equation}
Using Lyapunov-Schmidt reduction, following the proof of Proposition~4.3 in
\cite{li2024}, the leading order behaviors of the eigenpairs are determined by the matrix 
\begin{equation}\label{eq:Lyap}
\cM :=\eps\langle \tilde\bw, \partial_{\eps}\cL (\eps,0,K_*)\tilde\bw\rangle_{\cC_z}|_{\eps=0}+\delta\langle \tilde\bw, \partial_{\delta}\cL (0,\delta,K_*)\tilde\bw\rangle_{\cC_z}|_{\delta=0} + 
\langle\tilde \bw,  (\ell \bbeta_1 + \mu\bbeta_2))\cdot\nabla_{\bp}\cL(0,0,\bp)\tilde\bw\rangle_{\cC_z}|_{\bp=K_*}
- \lambda^{(1)} I,
\end{equation}
where $I$ is the 2-by-2 identity matrix, and $\lambda^{1} = \lambda - \lambda_*$. Here the leading order term of $\lambda$ is in the $\mathbb R$-norm, and that of the eigenfunctions $\tilde u_{n,\eps,\delta}$ are in the $H^1(\cC_z)$ norm. In the following subsections, we compute the leading order terms of the dispersion pairs when $\eps\neq0$ or $\delta\neq0$, depending on the signs of $t_1\eps$ and $t_2\delta$, at $\Kone$ and $\Ktwo$.

\subsection{Perturbed eigenpairs at $\Kone$, with respect to $\eps$}
At $\Kone$, when $\eps\neq0$ and $\delta=0$, the matrix in \eqref{eq:Lyap} takes the form
\begin{equation*}
\left(\begin{matrix}
t_1\eps - \lambda^{(1)} &\overline{\theta_*}(\ell+\mu\tau)\\
\theta_*(\ell+\mu\tau) &-t_1\eps - \lambda^{(1)}
\end{matrix}\right).
\end{equation*}
Its determinant is zero when $ \lambda^{(1)} =\pm\sqrt{\eps^2 t_1^2 + |\theta_*|^2 (\ell+\mu\tau)^2}$. 
At the higher eigenvalue $ \lambda^{(1)} = \sqrt{\eps^2 t_1^2 + |\theta_*|^2 (\ell+\mu\tau)^2}$, the matrix takes the form
\begin{equation*}
\left(\begin{matrix}
t_1\eps - \sqrt{\eps^2 t_1^2 + |\theta_*|^2 (\ell+\mu\tau)^2} &\overline{\theta_*}(\ell+\mu\tau)\\
\theta_*(\ell+\mu\tau) &-t_1\eps  - \sqrt{\eps^2 t_1^2 + |\theta_*|^2 (\ell+\mu\tau)^2}
\end{matrix}\right),
\end{equation*}
and the eigenspace is spanned by
\begin{equation*}
\left(t_1\eps +\sqrt{\eps^2 t_1^2 + |\theta_*|^2 (\ell+\mu\tau)^2} , \theta_*(\ell+\mu\tau) \right) \text{or}
\left(\overline{\theta_*}(\ell+\mu\tau), -t_1\eps+ \sqrt{\eps^2 t_1^2 + |\theta_*|^2 (\ell+\mu\tau)^2}\right).
\end{equation*}
At the lower eigenvalue $ \lambda^{(1)} =-\sqrt{\eps^2 t_1^2 + |\theta_*|^2 (\ell+\mu\tau)^2}$, the matrix takes the form
\begin{equation*}
\left(\begin{matrix}
t_1\eps + \sqrt{\eps^2 t_1^2 + |\theta_*|^2 (\ell+\mu\tau)^2} &\overline{\theta_*}(\ell+\mu\tau)\\
\theta_*(\ell+\mu\tau) &-t_1\eps  +\sqrt{\eps^2 t_1^2 + |\theta_*|^2 (\ell+\mu\tau)^2}
\end{matrix}\right),
\end{equation*}
and the eigenspace is spanned by
\begin{equation*}
\left(t_1\eps -\sqrt{\eps^2 t_1^2 + |\theta_*|^2 (\ell+\mu\tau)^2} , \theta_*(\ell+\mu\tau) \right)  \text{or}
\left(\overline{\theta_*}(\ell+\mu\tau), -t_1\eps- \sqrt{\eps^2 t_1^2 + |\theta_*|^2 (\ell+\mu\tau)^2}\right).
\end{equation*}

Observe $\sqrt{\eps^2 t_1^2 + |\theta_*|^2 (\ell+\mu\tau)^2} \geq |t_1\eps|$, with the equal sign attained when $\ell=\mu=0$.
Define 
\begin{equation*}
L_1(\eps,\ell,\mu):= \frac{\theta_*(\ell+\mu\tau)}{|\eps t_1|+ \sqrt{\eps^2 t_1^2 + |\theta_*|^2  (\ell+\mu\tau)^2}}.
\end{equation*}
Thus for all $t_1\eps$,
\begin{equation}\label{eq:tepslam}
\lambda_{1,\eps,0} \sim \lambda_* - \sqrt{\eps^2 t_1^2 + |\theta_*|^2 (\ell+\mu\tau)^2},  \quad \lambda_{2,\eps,0} \sim \lambda_* + \sqrt{\eps^2 t_1^2 + |\theta_*|^2 (\ell+\mu\tau)^2}.
\end{equation}
When $t_1\eps>0$,
\begin{equation}\label{eq:tepsp}
\tilde u_{1,\eps,0}(\cdot;\bp(\Kone;0,0)) \sim (- \overline{L_1}  \tilde w_1 + \tilde w_2)/\sqrt{1+|L_1|^2}, \quad \tilde u_{2,\eps,0}(\cdot;\bp(\Kone;0,0)) \sim (\tilde w_1 + L_1 \tilde w_2)/\sqrt{1+|L_1|^2}, 
\end{equation}
and when $t_1\eps<0$,
\begin{equation}\label{eq:tepsm}
\tilde u_{1,\eps,0}(\cdot;\bp(\Kone;0,0)) \sim  (\tilde w_1 -L_1 \tilde w_2)/\sqrt{1+|L_1|^2}, \quad \tilde u_{2,\eps,0}(\cdot;\bp(\Kone;0,0)) \sim ( \overline{L_1} \tilde w_1 + \tilde w_2)/\sqrt{1+|L_1|^2}.
\end{equation}
Here, the symbol $\sim$ represents equal to up to an order of $O(|\eps|,|\ell|,|\mu|)$. Finally, using 
\begin{equation}
u_{n,\eps,0}(\cdot;\bp(\Kone;0,0))=e^{\im\Kone\cdot\bx}\tilde u_{n,\eps,0}(\cdot;\bp(\Kone;0,0)), \quad w_n= e^{\im\Kone\cdot\bx}\tilde w_n,
\end{equation}
and
\begin{equation}
u_{n,\eps,0}(\cdot;\bp(\Kone;\ell,\mu)) = e^{\im\bp(\Kone;\ell,\mu)\cdot\bx}\tilde u_{n,\eps,0}(\cdot;\bp(\Kone;0,0)) = u_{n,\eps,0}(\cdot;\bp(\Kone;0,0))+O(\ell,
\mu),
\end{equation}
we obtain the relations shown in \eqref{eq:uepsp} and \eqref{eq:uepsm}.

\subsection{Perturbed eigenpairs at $\Kone$, with respect to $\delta$}
At $\Kone$, when $\eps=0$ and $\delta\neq0$, the matrix in \eqref{eq:Lyap} takes the form
\begin{equation*}
\left(\begin{matrix}
t_2\delta - \lambda^{(1)} &\overline{\theta_*}(\ell+\mu\tau)\\
\theta_*(\ell+\mu\tau) &-t_2\delta - \lambda^{(1)}
\end{matrix}\right).
\end{equation*}
Define 
\begin{equation*}
L_2(\delta,\ell,\mu):= \frac{\theta_*(\ell+\mu\tau)}{|\delta t_2|+ \sqrt{\delta^2 t_2^2 + |\theta_*|^2  (\ell+\mu\tau)^2}}.
\end{equation*}
We obtain that when $t_2\delta>0$,
\begin{equation}\label{eq:tdeltap}
\tilde u_{1,0,\delta}(\cdot;\bp(\Kone;0,0)) \sim(- \overline{L_2}  \tilde w_1 + \tilde w_2)/\sqrt{1+|L_2|^2} , \quad \tilde u_{2,0,\delta}(\cdot;\bp(\Kone;0,0)) \sim  (\tilde w_1 + L_2 \tilde w_2)/\sqrt{1+|L_2|^2},
\end{equation}
and when $t_2\delta<0$,
\begin{equation}\label{eq:tdeltam}
\tilde u_{1,0,\delta}(\cdot;\bp(\Kone;0,0)) \sim  (\tilde w_1 -L_2 \tilde w_2)/\sqrt{1+|L_2|^2}, \quad \tilde u_{2,0,\delta}(\cdot;\bp(\Kone;0,0)) \sim   (\overline{L_2} \tilde w_1 + \tilde w_2)/\sqrt{1+|L_2|^2}.
\end{equation}

\subsection{Perturbed eigenpairs at $\Ktwo$, with respect to $\eps$}
At $\Ktwo$, when $\eps\neq0$ and $\delta=0$, the matrix in \eqref{eq:Lyap} takes the form
\begin{equation*}
\left(\begin{matrix}
-t_1\eps - \lambda^{(1)} &-\overline{\theta_*}(\ell+\mu\tau)\\
-\theta_*(\ell+\mu\tau) &t_1\eps - \lambda^{(1)}
\end{matrix}\right).
\end{equation*}
For $ \lambda^{(1)} =+ \sqrt{\eps^2 t_1^2 + |\theta_*|^2 (\ell+\mu\tau)^2}$, which gives $u_{2,\eps,0}$, the matrix is 
\begin{equation*}
\left(\begin{matrix}
-t_1\eps - \sqrt{\eps^2 t_1^2 + |\theta_*|^2 (\ell+\mu\tau)^2} &-\overline{\theta_*}(\ell+\mu\tau)\\
-\theta_*(\ell+\mu\tau) &t_1\eps  - \sqrt{\eps^2 t_1^2 + |\theta_*|^2 (\ell+\mu\tau)^2}
\end{matrix}\right).
\end{equation*}
and eigenspace is spanned by
\begin{equation*}
\left(-t_1\eps +\sqrt{\eps^2 t_1^2 + |\theta_*|^2 (\ell+\mu\tau)^2} , -\theta_*(\ell+\mu\tau) \right) \text{or}
\left(-\overline{\theta_*}(\ell+\mu\tau), t_1\eps+ \sqrt{\eps^2 t_1^2 + |\theta_*|^2 (\ell+\mu\tau)^2}\right).
\end{equation*}

For $ \lambda^{(1)} = -\sqrt{\eps^2 t_1^2 + |\theta_*|^2 (\ell+\mu\tau)^2}$, which gives $u_{1,\eps,0}$, the matrix is 
\begin{equation*}
\left(\begin{matrix}
-t_1\eps + \sqrt{\eps^2 t_1^2 + |\theta_*|^2 (\ell+\mu\tau)^2} &-\overline{\theta_*}(\ell+\mu\tau)\\
-\theta_*(\ell+\mu\tau) &t_1\eps  +\sqrt{\eps^2 t_1^2 + |\theta_*|^2 (\ell+\mu\tau)^2}
\end{matrix}\right).
\end{equation*}
and the eigenspace is spanned by
\begin{equation*}
\left(-t_1\eps -\sqrt{\eps^2 t_1^2 + |\theta_*|^2 (\ell+\mu\tau)^2} , -\theta_*(\ell+\mu\tau) \right) \text{or}
\left(-\overline{\theta_*}(\ell+\mu\tau), t_1\eps- \sqrt{\eps^2 t_1^2 + |\theta_*|^2 (\ell+\mu\tau)^2}\right).
\end{equation*}
Hence when $t_1\eps>0$,
\begin{equation}\label{eq:twotepsp}
\tilde u_{1,\eps,0}(\cdot;\bp(\Ktwo;0,0)) \sim (\tilde w_1 + L_1 \tilde w_2)/\sqrt{1+|L_1|^2}, \quad \tilde u_{2,\eps,0}(\cdot;\bp(\Ktwo;0,0)) \sim  (- \overline{L_1}  \tilde w_1 + \tilde w_2)/\sqrt{1+|L_1|^2},
\end{equation}
and when $t_1\eps<0$,
\begin{equation}\label{eq:twotepsm}
\tilde u_{1,\eps,0}(\cdot;\bp(\Ktwo;0,0)) \sim (\overline{L_1} \tilde w_1 + \tilde w_2)/\sqrt{1+|L_1|^2} , \quad \tilde u_{2,\eps,0}(\cdot;\bp(\Ktwo;0,0)) \sim   (\tilde w_1 -L_1 \tilde w_2)/\sqrt{1+|L_1|^2}.
\end{equation}

\subsection{Perturbed eigenpairs at $\Ktwo$, with respect to $\delta$}
At $\Ktwo$, when $\eps=0$ and $\delta\neq0$, the matrix in \eqref{eq:Lyap} takes the form
\begin{equation}\label{}
\left(\begin{matrix}
t_2\delta + \lambda^{(1)} &-\overline{\theta_*}(\ell+\mu\tau)\\
-\theta_*(\ell+\mu\tau) &-t_2\delta + \lambda^{(1)}
\end{matrix}\right).
\end{equation}
We obtain that when $t_2\delta>0$,
\begin{equation}\label{eq:twotdeltap}
\tilde u_{1,\delta}(\cdot;\bp(\Ktwo;0,0)) \sim(\overline{L_2}  \tilde w_1 + \tilde w_2)/\sqrt{1+|L_2|^2} , \quad \tilde u_{2,\delta}(\cdot;\bp(\Ktwo;0,0)) \sim   (\tilde w_1 - L_2 \tilde w_2)/\sqrt{1+|L_2|^2},
\end{equation}
and when $t_2\delta<0$,
\begin{equation}\label{eq:twotdeltam}
\tilde u_{1,0,\delta}(\cdot;\bp(\Ktwo;0,0)) \sim (\tilde w_1 + L_2 \tilde w_2)/\sqrt{1+|L_2|^2} , \quad \tilde u_{2,0,\delta}(\cdot;\bp(\Ktwo;0,0)) \sim  ( -\overline{L_2} \tilde w_1 + \tilde w_2)/\sqrt{1+|L_2|^2}.
\end{equation}

\end{appendices}

\bibliographystyle{siam}

\end{document}